\documentclass[letterpaper]{sig-alternate-2013}
	
\pdfoutput=1

\usepackage{amsfonts}
\usepackage{amsmath}
\usepackage{amssymb}

\usepackage{algorithmic}
\usepackage[boxed]{algorithm}

\usepackage{epsfig}
\usepackage{epstopdf}

\usepackage[caption=false,font=footnotesize]{subfig}

\newcommand{\remove}[1]{}
\newcommand{\KDDinclude}[1]{#1}

\newcommand{\OnlyInFull}[1]{#1}
\newcommand{\OnlyInShort}[1]{}

\usepackage{color, xcolor}
\usepackage{ifthen}

\usepackage{hyperref}
\hypersetup{colorlinks=true,urlcolor=blue,linkcolor=blue,citecolor=blue}

\newtheorem{theorem}{Theorem}

\newtheorem{fact}{Fact}
\newtheorem{definition}{Definition}
\newtheorem{assumption}{Assumption}
\newtheorem{lemma}[fact]{Lemma}
\newtheorem{proposition}[fact]{Proposition}

\newtheorem{axiom}{Axiom}

\newcommand{\compilehidecomments}{false}

\ifthenelse{ \equal{\compilehidecomments}{true} }{%
	\newcommand{\wei}[1]{}
	\newcommand{\shanghua}[1]{}
}{
	\newcommand{\wei}[1]{{\color{blue!50!black}  [\text{Wei:} #1]}}
	\newcommand{\shanghua}[1]{{\color{brown!60!black} [\text{Shanghua:} #1]}}
}

\newcommand{\hide}[1]{}

\newcommand{\argmax}{\operatornamewithlimits{argmax}}

\newcommand{\calI}{{\mathcal{I}}}

\newcommand{\E}{\mathbb{E}}
\newcommand{\I}{\mathbb{I}}
\newcommand{\R}{\mathbb{R}}
\newcommand{\bT}{\mathbb{T}}

\newcommand{\bI}{\boldsymbol{I}}
\newcommand{\bL}{\boldsymbol{L}}

\newcommand{\bR}{\boldsymbol{R}}
\renewcommand{\bT}{\boldsymbol{T}}

\newcommand{\bu}{\boldsymbol{u}}

\newcommand{\bv}{\boldsymbol{v}}

\newcommand{\bW}{\boldsymbol{W}}
\newcommand{\bX}{\boldsymbol{X}}
\newcommand{\bY}{\boldsymbol{Y}}
\newcommand{\bx}{\boldsymbol{x}}
\newcommand{\by}{\boldsymbol{y}}
\newcommand{\bZ}{\boldsymbol{Z}}
\newcommand{\bpi}{\boldsymbol{\pi}}
\newcommand{\cA}{\mathcal{A}}
\newcommand{\cB}{\mathcal{B}}

\newcommand{\cV}{\mathcal{V}}

\newcommand{\btau}{\boldsymbol{\tau}}
\newcommand{\btheta}{\boldsymbol{\theta}}
\newcommand{\boldeta}{\boldsymbol{\eta}}

\newcommand{\epsi}{\hat{\boldsymbol{\psi}}}

\newcommand{\vv}{\vec{v}}

\newcommand{\vb}{\vec{b}}
\newcommand{\vzero}{\vec{0}}

\def\INPUT{\REQUIRE}
\def\OUTPUT{\ENSURE}

\def\est{\boldsymbol{\it est}}
\def\LB{\boldsymbol{\it LB}}
\def\ASVRR{{\sf ASV-RR}}
\def\ASVRRW{{\sf ASV-RR-W}}
\def\ASNIRR{{\sf ASNI-RR}}
\def\IMM{{\sf IMM}}
\def\Degree{{\sf Degree}}

\def\Shapley{\textit{Shapley}}
\def\SNI{\textit{SNI}}

\makeatletter
\def\@copyrightspace{\relax}
\makeatother

\permission{\copyright 2017 International World Wide Web Conference Committee \\ (IW3C2), published under Creative Commons CC BY 4.0 License.}
\conferenceinfo{WWW 2017,}{April 3--7, 2017, Perth, Australia.}
\copyrightetc{ACM \the\acmcopyr}
\crdata{978-1-4503-4913-0/17/04. \\
	http://dx.doi.org/10.1145/3038912.3052608 \\
	\includegraphics{cc.png} }
	
\begin{document}





\title{Interplay between Social Influence and Network Centrality:
	A Comparative Study 
	on Shapley Centrality and Single-Node-Influence Centrality
	}

\numberofauthors{2} 

\author{
	\alignauthor Wei Chen \\
	\affaddr{Microsoft Research}\\
	\affaddr{Beijing, China}\\
	\email{weic@microsoft.com}
	\and
	\alignauthor Shang-Hua Teng \\
	\affaddr{University of Southern California}\\
	\affaddr{Los Angeles, CA, U.S.A.}\\
	\email{shanghua@usc.edu}
}



\maketitle

\begin{abstract}
	We study network centrality
	based on dynamic influence propagation
	models in social networks.
To illustrate our integrated mathematical-algorithmic approach
  for understanding the fundamental 
  interplay between dynamic influence processes and static network structures,
  we focus on two basic centrality measures: 
	(a) {\em Single Node Influence} (SNI) {\em centrality}, 
	which measures each node's significance
  by its influence spread;\footnote{The influence spread of 
		a group is the expected number 
		of nodes this group 
		can activate as the initial active set.}
	and
	(b) {\em Shapley Centrality}, which uses the Shapley value 
	of the influence spread function ---
	formulated based on a fundamental 
	cooperative-game-theoretical concept --- to measure 
	the significance of nodes.
	We present a comprehensive comparative study of these 
	two centrality measures.
	Mathematically, we present 
	axiomatic characterizations, 
	which precisely capture the essence of these
	two centrality measures and their fundamental differences.
	Algorithmically,
	we provide scalable algorithms for approximating
	them for a large family of social-influence instances.
	Empirically, we demonstrate their similarity and differences
	in a number of real-world social networks,
	as well as the efficiency of our scalable algorithms.
	Our results shed light on their applicability:
	SNI centrality is suitable for assessing individual influence in isolation
	while Shapley centrality assesses individuals' performance in group influence settings.

\end{abstract}

%
%
%
%


\sloppy

\section{Introduction}\label{sec:Introduction}


Network science is a fast growing discipline 
  that uses mathematical graph structures to represent real-world
  networks --- such as the Web, Internet, social networks, 
  biological networks, and power grids ---
  in order to study fundamental network properties.
However, network phenomena are far more complex 
  than what can be captured only by nodes and edges, making
  it essential to formulate network concepts by incorporating network facets
  beyond graph structures  \cite{NetworkEssence}.
For example, network centrality is a key concept in network analysis.
The {\em centrality} of nodes, usually measured by a real-valued function,
	reflects their significance, importance, or crucialness 
	within the given network.
Numerous centrality measures have been proposed, 
  based on degree, closeness, betweenness and eigenvector (i.e., PageRank)
(cf. \cite{NewmanBook}).
However, most of these centrality measures focus
  only on the static topological structures of the networks, 
  while	real-world network data 
  include much richer interaction dynamics beyond static topology.
	
{\em Influence propagation} is a wonderful example of interaction dynamics in social networks.
As envisioned by 
Domingos and Richardson \cite{RichardsonDomingos,DomingosRichardson},
and beautifully formulated by Kempe, Kleinberg, and Tardos \cite{kempe03},
{\em social influence propagation}
can be viewed as a stochastic dynamic process
over an underlying  static graph:
After a group of nodes becomes {\em active}, 
these {\em seed nodes} propagate their influence 
through the graph structure.
Even when the static graph structure of a social network is fixed, 
	dynamic phenomena such as the spread of ideas, epidemics, and technological
	innovations can follow different processes.
Thus, network centrality, which aims to measure nodes' importance 
  in social influence,
  should be based not only
  on static graph structure, 
  but also on the dynamic influence propagation process.
	
In this paper, we address 
  the basic question of {\em how to formulate network centrality
		measures that reflect dynamic influence propagation}.
We will focus on the study 
  of the {\em interplay between social influence and network centrality}.

\KDDinclude{A social influence} instance\remove{\KDDinclude{$\calI$}}
specifies a directed graph $G=(V,E)$
and an influence model $P_{\calI}$ (see Section \ref{sec:prel}).
\remove{,  we will review various  network influence models introduced in 
	\cite{kempe03}.}
For each \remove{group}$S\subseteq V$, $P_{\calI}$ 
  defines a stochastic influence process on $G$ with $S$ as the initial active set,
which 
activates a random  set $\bI(S) \supseteq S$ with probability 
$P_{\calI}(S,\bI(S))$.
Then, $\sigma(S) = \E[|\bI(S)|]$
\KDDinclude{is the} {\em influence spread} \KDDinclude{of $S$.}
The question above can be restated as: 
Given a social-influence instance $(V,E,P_{\calI})$,
	how should we define the centrality of nodes in $V$?

A natural centrality measure for each node $v\in V$ is
    its influence spread $\sigma(\{v\})$.
However, this measure --- referred to as
  the {\em single node influence} (SNI)  {\em centrality} ---
  completely  ignores the influence profile of groups of nodes 
  and a node's role in such group influence.
Thus, other more sensible centrality measures accounting 
  for group influence may better capture nodes' roles in social influence.
As a concrete formulation of group-influence analyses, we
  apply Shapley value \cite{Shapley53} ---
  a fundamental concept from cooperative game theory ---
  to define a new centrality measure 
  for social influence.

Cooperative game theory is a mathematical theory studying people's performance and behavior in coalitions
	(cf. \cite{MyersonBook}).
Mathematically, an $n$-person {\em coalitional game}
is defined by 
a {\em characteristic function}
$\tau: 2^V \rightarrow \R$, where $V=[n]$, and $\tau(S)$ is the utility of the coalition $S$ \cite{Shapley53}.
In this game,
the {\em Shapley value} $\phi_v^{\Shapley}(\tau)$  of 
$v\in V$  is $v$'s {\em expected marginal contribution in a random order}. 
More precisely:
\begin{equation} \label{eq:shapleydef}
\phi_v^{\Shapley}(\tau) = \E_{\bpi}[\tau(S_{\bpi,v} \cup \{v\}) - \tau(S_{\bpi,v})],
\end{equation}
where $S_{\bpi, v}$ denotes the set of players preceding $v$ 
in a random permutation $\bpi$ of $V$:
The Shapley value enjoys an 
  axiomatic characterization (see Section \ref{sec:prel}), 
and 
is widely considered to be the {\em fairest} measure of 
a player's power in a cooperative~game.

Utilizing the above framework, 
  we view influence spread $\sigma(\cdot)$ as a characteristic function, and
  define the {\em Shapley centrality} of an influence instance 
  as the Shapley value of $\sigma$.

In this paper, we present a comprehensive comparative study 
 of SNI and Shapley centralities.
In the age of Big Data, networks are massive.
Thus, an effective solution concept in network science
should be both  {\em mathematically meaningful} and
{\em algorithmically efficient}.
In our 
  study, we will address both the
  conceptual and algorithmic questions.

Conceptually,
 influence-based centrality can be viewed as a {\em dimensional reduction} from the high dimensional
influence model $P_{\calI}$ to a low dimensional centrality measure.
Dimensional reduction of data is a challenging task, because 
inevitably some information is lost.
\OnlyInFull{As highlighted by Arrow's celebrated impossibility theorem on voting 
	\cite{ArrowBook},
	for various (desirable) properties, 
	{\em conforming} dimensional reduction scheme may not even exist. }Thus, it is fundamental 
to characterize what each centrality measure captures.

So, ``what do Shapley and SNI centralities capture?
what are their basic differences?''
Axiomatization is an instrumental approach for such characterization.
In Section \ref{sec:Axioms}, we present 
our axiomatic characterizations.
We present five axioms for Shapley centrality,
  and prove that it is the unique centrality measure satisfying these
axioms. 
We do the same for the SNI centrality with three axioms.
Using our axiomatic characterizations,
we then provide a detailed comparison of Shapley and SNI centralities.
Our characterizations show that 
(a) SNI centrality focuses on 
individual influence and 
would not be appropriate for models concerning group influence,
such as threshold-based models.
(b) Shapley centrality focuses on 
individuals' ``irreplaceable power" in group influence settings, 
but may not	be interpreted well if one 
prefer to focus on individual influence in isolation.

The computation of influence-based centralities is also a challenging 
  problem:
Exact computation of influence spread 
  in the basic {\em independent cascade}
  and {\em linear-threshold} models has been shown to be 
  $\#$P-complete \cite{wang2012scalable,ChenYuanZhang}.\OnlyInShort{ Shapley centrality computation 
  seems to be more challenging 
  since its definition as in Eq.~\eqref{eq:shapleydef}
 involves various influence spreads derived from $n!$ permutations. }\OnlyInFull{ Shapley centrality computation seems to be more challenging since its definition as in Eq.~\eqref{eq:shapleydef}
	involves $n!$ permutations, and existing Shapley value computation in several simple network games
	have quadratic or cubic time complexity 
	\cite{ShapleyValueForCentrality1}. }Facing these challenges, 
  in Section \ref{sec:ScalableAlgorithm}, 
  we present provably-good scalable
  algorithms for approximating both Shapley and SNI centralities 
   of a large family of social influence instances.
Surprisingly, both algorithms share 
  the same algorithm structure, which 
  extends techniques from the recent algorithmic breakthroughs
  in influence maximization \cite{BorgsBrautbarChayesLucier,tang14,tang15}.
We further 
  conduct empirical evaluation of Shapley and SNI centralities 
  in a number of real-world networks.
Our experiments --- see Section \ref{sec:experiments} --- 
 show that our  algorithms can scale up 
  to networks with tens of millions of nodes and edges,
  and these two centralities are similar 
  in several cases but also have noticeable differences.

These combined mathematical/algorithmic/empirical analyses together present
(a) a systematic case study of 
  the interplay between influence dynamics and network centrality
  based on Shapley and SNI centralities;
(b) axiomatic characterizations for two basic centralities that precisely 
  capture their similarities and differences; and
 (c)  new scalable algorithms 
  for influence models.
We believe that the dual 
  axiomatic-and-algorithmic characterization 
  provides a comparative framework for evaluating 
  other influence-based 
  network concepts in the future.
\OnlyInShort{Due to space constraint, 
 proofs and 
additional results are in \cite{CT16}.
	}
	
	\OnlyInFull{
		For presentation clarity, we move the technical proofs into the appendix,
		which also contains 
		additional technical materials for (algorithmic and axiomatic) generalization to weighted influence models.
	}

\subsection{Related Work}

Network centrality has been extensively studied 
(see~\cite{NewmanBook} and the references therein for
	a 
comprehensive introduction).
Most classical centralities, 
  based on degree, closeness, betweenness, eigenvector, 
  are defined on static graphs.
But some also have dynamic interpretations 
  based on random-walks or network flows~\cite{BorgattiCentrality}.
Eigenvector centrality~\cite{Bonacich1972} and 
its closely related Katz-\cite{Katz} and
  Alpha-centrality~\cite{Bonacich1987power}
  can be viewed as some forms of influence measures,
  since their dynamic processes are  non-conservative~\cite{GL14}, 
  meaning that 
  items could be replicated and propagated, 
  similar to diffusion of ideas, opinions, etc.
PageRank \OnlyInFull{\cite{PageRank,PageRankBMW98}}\OnlyInShort{\cite{PageRank}} and other random-walk related centralities 
correspond to conservative processes, 
  and thus may not be suitable for propagation dynamics.
Percolation centrality~\cite{PercolationCentrality} also 
  addresses diffusion process, 
  but its definition only involves static percolation. 
None of above maps specific propagation models to network centrality.
Ghosh et al.~\cite{GhoshInterplay} maps a 
  linear dynamic process characterized by parameterized Laplacian 
  to centrality 
  but the social influence models we consider in this paper 
  are beyond 
  such linear dynamic framework.
Michalak et al. use Shapley value as network centrality~\cite{ShapleyValueForCentrality1}, but they only consider five basic 
  network games based on local sphere of influence,
  and their algorithms run in (least) quadratic time.
To the best of our knowledge, 
  our study is the first to explicitly map general social network
  influence propagation models to network centrality.



Influence propagation has been extensively studied, 
 but most focusing on influence maximization
	tasks~\cite{kempe03,wang2012scalable,ChenYuanZhang},
which aims to efficiently select a set of nodes 
  with the largest influence spread.
The solution is not a centrality measure and 
  the seeds in the solution may not be the high centrality nodes.
Borgatti~\cite{borgatti06} provides clear conceptual discussions on the difference between
	centralities and such key player set identification problems.
Algorithmically, our construction extends the idea of reverse reachable sets,
  recently introduced
	 in~\cite{BorgsBrautbarChayesLucier,tang14,tang15}
  for scalable influence maximization.

In terms of axiomatic characterizations of network centrality, Sabidussi is the first
	who provides a set of axioms that a centrality measure should satisfy~\cite{Sabidussi66}.
A number of other studies since then either provide other axioms that a centrality measure
	should satisfy (e.g.~\cite{Nieminen73,BV14,SB16}) or a set of axioms that uniquely
	define a centrality measure (e.g. \cite{PageRankAxioms} on PageRank without the damping factor).
All of these axiomatic characterizations focus on
  static graph structures, while our axiomatization 
	focuses on the interplay between dynamic influence processes
  and static graph structures, and thus our
	study fundamentally
	differs from all the above characterizations.
While we are heavily influenced by the axiomatic characterization of the 
	Shapley value~\cite{Shapley53}, we are also inspired by
	social choice theory \cite{ArrowBook}, and particularly by
	\cite{IntellectualInfluence} on measures of intellectual influence and
	\cite{PageRankAxioms} on PageRank.

\OnlyInShort{\vspace{-2mm}}
\section{Influence and Centrality} \label{sec:prel}

\OnlyInFull{In this section, we review the basic concepts about 
	social influence models and Shapley value, and define
	the Shapley and single node influence centrality measures.
}


\OnlyInShort{\vspace{-2mm}}
\subsection{Social Influence Models} \label{sec:infmodel}

A network-influence instance is usually specified by
  a triple $\calI = (V,E,P_{\calI})$, 
  where a directed graph $G = (V,E)$ represents
  the structure of a social network, and 
  $P_{\calI}$ defines the influence model~\cite{kempe03}.
As an example, consider the classical discrete-time
  {\em independent cascade (IC) model},  
in which each directed  edge $(u,v) \in E$
  has an influence probability $p_{u,v}\in [0,1]$.
At time $0$,  nodes in a given seed set $S$ are activated while other nodes are inactive.
At time $t \ge 1$,
   for any node $u$ activated at time $t-1$, it has one
   chance to activate each of its inactive out-neighbor 
   $v$ with an independent probability $p_{u,v}$.
When there is no more activation,
  the stochastic process ends with a random set $\bI(S)$
  of nodes activated during the process.
\remove{We say}\remove{denote the random set of 
. 
Let 
  $\bI(S)$ denote the random set of 
  all nodes activated in a run of the influence process 
  with seed set $S$.}The {\em influence spread} of $S$ is $\sigma(S) = \E[|\bI(S)|]$,
  the expected number of nodes influenced by 
  $S$.
Throughout the paper, we use boldface symbols to represent 
   random variables.

Algorithmically, we will focus on
  the (random) {\em triggering model} \cite{kempe03}, which has IC model as a special case.
In this model, 
  each $v\in V$ has a random {\em triggering set} $\bT(v)$, 
  drawn	from a distribution defined by the influence model over the power set  of all in-neighbors of $v$.
At time $t=0$, triggering sets $\{\bT(v)\}_{v\in V}$ are drawn independently,
  and the seed set $S$ is activated.
At \remove{time }$t\ge 1$, if $v$ is not active,
  it becomes activated\remove{unless $v$ is already active}
  if some $u\in \bT(v)$ is activated at time $t-1$.
\OnlyInFull{The {\em influence spread} of $S$ is $\sigma(S) = \E[|\bI(S)|]$,
  where $\bI(S)$ denotes the random set activated by $S$.
IC is the triggering model that:
For each directed edge $(u,v)\in E$, 
  add $u$ to $\bT(v)$ with an independent probability of $p_{u,v}$.}
The triggering model can be equivalently viewed under the 
{\em live-edge graph model}:
(1) Draw independent random triggering sets~$\{\bT(v)\}_{v\in V}$;
(2) form a {\em live-edge graph} $\bL = (V,\{(u,v): u\in \bT(v) \})$, where $(u,v), u\in \bT(v)$ 
	is referred as a {\em live edge}.
For any subgraph $L$  of $G$ and $S\subseteq V$, 
let $\Gamma(L, S)$ be the set of nodes in $L$ 
	reachable from set $S$.
Then set of active nodes with seed set $S$ is $\Gamma(\bL, S)$, and influence spread
	$\sigma(S) = \E_{\bL}[|\Gamma(\bL, S)|] 
	= \sum_{L} \Pr(\bL = L)\cdot |\Gamma(L, S)|$.
We say a set function $f(\cdot)$ 
  is {\em monotone} if $f(S)\le f(T)$ 
   whenever $S\subseteq T$, and {\em submodular}
   if $f(S\cup \{v\}) - f(S) \ge f(T\cup \{v\}) - f(T)$ whenever
	$S \subseteq T$ and $v\not\in T$.
As shown in \cite{kempe03},  
  in any triggering model,  $\sigma(\cdot)$  is monotone and 
  \OnlyInShort{submodular.}\OnlyInFull{submodular,
  because $|\Gamma(L, S)|$ is monotone and submodular 
  for each graph $L$.}

More generally, we 
	define an {\em influence instance} 
  as a triple $\calI = (V,E,P_{\calI})$, where
$G = (V,E)$ represents the underlying network, and
$P_{\calI}: 2^V \times 2^V \rightarrow \R$ defines the probability  that
  in the influence process, 
any seed set $S\subseteq V$ activates {\em exactly} nodes in any target set $T\subseteq V$ and no other nodes: 
If $\bI_{\calI}(S)$ denotes the random set activated by seed set $S$, 
then $\Pr(\bI_{\calI}(S) = T) = P_{\calI}(S,T)$.
This probability profile 
  is commonly defined by a succinct influence model, 
  such as the triggering model, which interacts with network $G$.
We also require that: 
%
(a) $P_{\calI}(\emptyset, \emptyset) = 1$, 
$P_{\calI}(\emptyset, T) = 0$, $\forall T\neq \emptyset$, and
(b) if $S \not\subseteq T$ then $P_{\calI}(S, T) = 0$,
i.e., $S$ always activates itself ($S\subseteq \bI_{\calI}(S)$).
Such model is also referred to as the {\em progressive} influence model.
The {\em influence spread} of \remove{the seed set }$S$ is: 
$$\sigma_{\calI}(S) = \E[|\bI_{\calI}(S)|] = \sum_{T\subseteq V, S\subseteq T}P_{\calI}(S, T) \cdot |T|.$$

\vspace{-2mm}
\subsection{Coalitional Games and Shapley Values}

An $n$-person {\em coalitional game} over $V = [n]$
  is specified by a {\em characteristic function}
  $\tau: 2^V \rightarrow \R$, 
  where for any coalition $S\subseteq V$, 
  $\tau(S)$ denotes the {\em cooperative utility} 
  of  
   $S$. 
In cooperative game theory,
a {\em ranking function}  $\phi$ 
  is a mapping from a\remove{any} characteristic 
  function $\tau$ to a vector in $\R^n$. 
A fundamental solution concept of cooperative game theory 
  is the ranking function given by the {\em Shapley value} \cite{Shapley53}:
Let $\Pi$ be the set of all permutations of $V$.
For any $v\in V$ and $\pi \in \Pi$, 
  let $S_{\pi, v}$ denote the set of nodes in $V$
  preceding $v$ in permutation $\pi$.
Then, $\forall v\in V$:\remove{, its Shapley value $\phi^{\Shapley}_v(\tau)$ is:}
\vspace{-1mm}
\begin{align*}
\phi^{\Shapley}_v(\tau) = & \frac{1}{n!}\sum_{\pi \in \Pi} \left( \tau(S_{\pi,v} \cup \{v\}) - \tau(S_{\pi,v})\right) \\
	= & \sum_{S\subseteq V \setminus \{v\}} \frac{|S|!(n-|S|-1)!}{n!}
		\left( \tau(S \cup \{v\}) - \tau(S)\right).
\end{align*}
\remove{With a bit of abuse of notation,
  we will also use $\Pi$ to denote the 
  distribution
  that each permutation of $V$ has an equal probability
	of $1/n!$.
Let $\bpi$ be a random permutation drawn from $\Pi$.}We use $\bpi\sim \Pi$ to denote that 
 $\bpi$ is a random permutation uniformly drawn from $\Pi$.
Then: \remove{$\phi_v(\tau)$ can be equivalently defined as:}
\begin{equation} \label{eq:sprandom}
\phi^{\Shapley}_v(\tau) = \E_{\bpi\sim \Pi}[\tau(S_{\bpi,v} \cup \{v\}) - \tau(S_{\bpi,v})].
\end{equation}
The Shapley value of $v$ measures
  $v$'s marginal contribution over the set preceding $v$
	in a random permutation.

Shapley 
 \cite{Shapley53} proved a \remove{the following}remarkable representation theorem:
  The Shapley value is the unique ranking function
	that satisfies all the following four conditions:
\remove{\begin{itemize}
\item }(1) {\sf Efficiency}:  
  $\sum_{v\in V} \phi_{v}(\tau) = \tau(V)$.
\remove{\item }(2) {\sf Symmetry}:
For any $u,v\in V$, 
  if  $\tau(S\cup \{u\}) = \tau(S\cup \{v\})$, $\forall 
  S\subseteq V\setminus \{u,v\}$,
 then $\phi_u(\tau) = \phi_v(\tau)$.
\remove{\item }(3) {\sf Linearity}: 
For any two characteristic functions $\tau$ and $\omega$,  for any
 $\alpha, \beta > 0$, 
$\phi(\alpha\tau + \beta\omega) = \alpha\phi(\tau) + \beta\phi(\omega)$.
\remove{\item }(4) {\sf Null Player}:
For any $v\in V$, 
	if $\tau(S\cup \{v\}) - \tau(S) = 0$,  $\forall 
   S\subseteq V \setminus \{v\}$, 
   then~$\phi_v(\tau)=0$.
\remove{\end{itemize}}{\sf Efficiency} states that the total utility is fully distributed.  
{\sf Symmetry} states that two players' ranking values should be the same
 if they have the identical marginal utility profile. 
{\sf Linearity} states that the ranking values of the weighted sum of 
  two coalitional games is the same as the weighted
  sum of their ranking values.
{\sf Null Player}
  states that a player's ranking value should be zero if 
  the player has zero marginal utility to every subset.


\subsection{Shapley and SNI Centrality}

The influence-based centrality measure 
  aims at assigning a value for every node under every influence instance:
\OnlyInShort{\vspace{-2mm}}
\begin{definition}[Centrality Measure]\label{def:CM}
	An {\em (influence-based) centrality measure} $\psi$ is a mapping from 
	an influence instance $\calI = (V,E,P_{\calI})$ 
	to a real vector $(\psi_{v}(\calI))_{v\in V} \in \R^{|V|}$.
\end{definition}



\vspace{-2mm}

The {\em single node influence} (SNI) {\em centrality}, denoted by
  $\psi^{\SNI}_v(\calI)$,
  assigns the influence spread of node $v$ as $v$'s centrality measure:
$\psi^{\SNI}_v(\calI) = \sigma_{\calI}(\{v\})$.

The {\em Shapley centrality}, denoted by
	$\psi^{\Shapley}(\calI)$, is the Shapley value
	of the influence spread function $\sigma_{\calI}$:
$\psi^{\Shapley}(\calI) = \phi^{\Shapley}(\sigma_{\calI})$.
As a subtle point, note that  $\phi^{\Shapley}$ maps from a 
  $2^{|V|}$ dimensional $\tau$ to a $|V|$-dimensional vector, 
  while, formally, $\psi^{\Shapley}$ 
  maps from $P_{\calI}$ --- whose dimensions is close to $2^{2|V|}$ --- to a 
	$|V|$-dimensional vector.

To help understand these definitions, 
Figure~\ref{fig:example} provides a simple example of a
$3$-node graph in the IC model 
  with influence probabilities shown on the edges.
The associated table shows the result for Shapley and SNI centralities.
While SNI is straightforward in this case, 
  the Shapley centrality calculation already looks complex.
\OnlyInShort{Due to space constraint, we left readers to verify the computation.}\OnlyInFull{For example, 
	for node $u$, its second term in the Shapley computation, $\frac{1}{3}(1-p)\cdot 1.5 $,
	accounts for the case where $u$ is ordered in the second place (with probability $1/3$), in which
	case only when the first-place node (either $v$ or $w$) does not activate $u$ (with probability
	$1-p$), it could have marginal influence of $1$ in activating itself, and $0.5$ in activating
	the remaining node.
Similarly, the third term for the Shapley computation for node $v$ accounts for the case where 
	$v$ is ordered second and $w$ is ordered first (with probability $1/6$), 
	in which case if $w$ does not activate $u$ (with probability $1-p$), 
	$v$'s marginal influence spread is $1$ for itself and $p$ for activating $u$; while if $w$
	activates $u$ (with probability $p$), only when $u$ does not activate $v$ (with probability $0.5$), 
	$v$ has marginal influence of $1$ for itself.
The readers can verify the rest.}
Based on the result, 
  we find that for interval $p \in (1/2, 2/3)$, Shapley and SNI centralities
  do not align in ranking: 
  Shapley places $v,w$ higher than $u$ while SNI puts $u$ higher than $v,w$.
\begin{figure} [t]
	\centering
	{\includegraphics[width=1.0\linewidth]{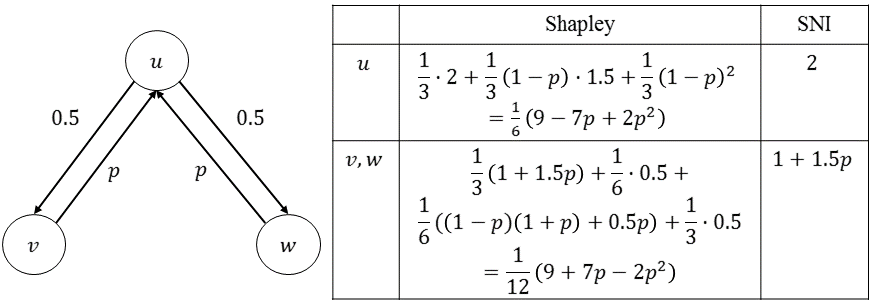}}
	\caption{Example on Shapley and SNI centrality.}
	\label{fig:example}
\end{figure}
This simple example already illustrates that 
  (a) computing Shapley centrality could be
  a nontrivial task; and (b) the relationship between Shapley 
  and SNI centralities could be 
  complicated.
Addressing both the computation and characterization questions
   are the subject of the remaining sections.


\OnlyInShort{\vspace{-2mm}}
\section{Axiomatic Characterization}\label{sec:Axioms}

In this section, we present two sets of axioms uniquely 
  characterizing Shapley and SNI centralities,
  respectively, based on which we 
  analyze their similarities and differences.

\subsection{Axioms for Shapley Centrality}

Our set of axioms for characterizing the Shapley centrality
  is adapted from the classical  Shapley's axioms \cite{Shapley53}.

The first axiom states that
labels on the nodes should have no effect on centrality measures.
This ubiquitous axiom is similar to the isomorphic axiom 
 in some other centrality characterizations, e.g.~\cite{Sabidussi66}.


\vspace{-1mm}
\begin{axiom}[Anonymity] \label{axiom:anonymity}
For any influence instance $\calI = (V,E,P_{\calI})$, and permutation 
$\pi\in \Pi$, $\psi_{v}(\calI) = \psi_{\pi(v)}({\pi(\calI)})$, $\forall v \in V$.
\end{axiom}
\vspace{-1mm}
In Axiom \ref{axiom:anonymity},
  $\pi(\calI) = (\pi(V),\pi(E), \pi(P_{\calI}))$
denotes the isomorphic instance:
(1) $\forall u,v\in V$,  $(\pi(u),\pi(v))\in \pi(E)$ iff  $(u,v)\in E$, and
(2) $\forall S,T\subseteq V$,
$P_{\calI}(S,T) = P_{\pi(\calI)}(\pi(S),\pi(T))$.

The second axiom states that the 
  centrality measure divides 
  the total share of influence $|V|$. 
In other words, the average centrality 
   is normalized 
	to $1$.


\vspace{-1mm}
\begin{axiom}[Normalization] \label{axiom:normalization}
For every influence instance $\calI = (V,E,P_{\calI})$, 
  $\sum_{v\in V} \psi_{v}(\calI)  = |V|$.
\end{axiom}

The next axiom characterizes the centrality 
  of a type of extreme nodes in  social influence.
\remove{\begin{definition}[Sink Node] \label{def:sinknode} In an influence instance}In instance $\calI = (V,E,P_{\calI})$,
  we say $v\in V$ is a {\em sink node}
  if $\forall S, T\subseteq V\setminus \{v\}$, 
$P_{\calI}(S\cup \{v\}, T\cup \{v\}) = P_{\calI}(S, T) + P_{\calI}(S, T\cup\{v\})$.
\remove{\end{definition}}In the extreme case  when $S = T=\emptyset$, 
$P_{\calI}(\{v\}, \{v\}) = 1$,
i.e., $v$ can only influence itself.
When  $v$ joins another $S$ to form a seed set, 
  the influence to a target $T\cup\{v\}$ can always be achieved by
  $S$ alone (except perhaps the influence to $v$ itself).
In the triggering model, a sink node is (indeed)
	a node without outgoing edges, matching the name ``sink''.

Because a sink node $v$ has no influence on other nodes, 
  we can ``remove'' it and obtain a projection of the influence 
  model on the \remove{social }network without $v$:\remove{, as the following:}
\remove{\begin{definition} [Influence Projection: Sink Nodes] \label{def:projection}
Suppose $v\in V$ is a sink node in  an \remove{influence }instance 
 $\calI = (V,E,P_{\calI})$.
}Let $\calI\setminus\{v\} = (V \setminus \{v\}, E\setminus\{v\}, P_{\calI\setminus\{v\}})$ denote the {\em projected}
	instance  over 
        $ V \setminus \{v\}$,
	where  $E\setminus\{v\} = \{(i,j) \in E: v\not\in \{i,j\} \}$
	and $P_{\calI \setminus\{v\}}$ is the influence model such 
	that for all $S,T\subseteq V \setminus \{v\}$:
	\[P_{\calI \setminus \{v\}}(S, T)
	= P_{\calI}(S, T) + P_{\calI}(S, T\cup\{v\}).
	\]
\remove{\end{definition}}
Intuitively, since sink node $v$ is removed, the previously distributed influence from $S$ to
	$T$ and $T\cup\{v\}$ is merged into the influence from $S$ to $T$ in the projected instance.
For the triggering model, influence projection is simply removing the
	sink node $v$ and its incident incoming edges without changing the triggering set distribution of
	any other nodes.

Axiom 3 below considers the simple case
  when the influence instance has two sink nodes $u, v\in V$.
In such a case, $u$ and $v$ have no influence to each other,
  and they influence no one else.
\remove{no other nodes.}Thus,\remove{intuitively,} 
  their centrality should be fully determined by $V \setminus \{u,v\}$:
Removing one sink node --- say  $v$ ---
  should not affect the centrality measure of another sink node $u$.
\remove{The next axiom makes this intuition explicit.}

\vspace{-1mm}
\begin{axiom}[Independence of Sink Nodes] \label{axiom:sink}
For any influence instance $\calI = (V,E,P_{\calI})$, 
  for any pair of sink nodes $u,v\in V$ in $\calI$, it should be the case:
$\psi_{u}(\calI) = \psi_{u}({\calI\setminus\{v\}})$.
\end{axiom}

\vspace{-1mm}
The next axiom 
	considers {\em Bayesian social influence} 
  through a given network:
\remove{\begin{definition}[Bayesian social influence Process] \label{def:bayesian}}Given a graph $G=(V,E)$, and $r$ influence instances on $G$:
$\calI^{\eta} = (V,E,P_{\calI^{\eta}})$ with $\eta \in [r]$.
Let $\lambda = (\lambda_1, \lambda_2, \ldots, \lambda_r)$ 
  be a prior distribution on $[r]$, 
  i.e. $\sum_{\eta=1}^{r} \lambda_\eta = 1$, and $\lambda_{\eta}\ge 0$, 
  $\forall \eta\in [r]$.
The {\em Bayesian influence instance}
  $\calI_{\cB(\{\calI^{\eta}\}, \lambda)}$ has the following
  influence process for a seed set $S\subseteq V$: 
  (1) Draw a random index $\boldeta \in [r]$ according to
  distribution $\lambda$ (denoted as $\boldeta \sim \lambda $).
 (2) Apply the influence process of $\calI^{\boldeta}$ 
     with seed set $S$ to obtain the activated set $T$.
Equivalently, we have for all $S, T\subseteq V$, 
	$P_{\calI_{\cB(\{\calI^{\eta}\}, \lambda)}}(S, T)
	= \sum_{\eta=1}^{r} \lambda_\eta P_{\calI^\eta}(S, T)$.
In the triggering model, we can view each live-edge graph and the deterministic diffusion on it via reachability
	as an influence instance, and the diffusion of the triggering model is by the Bayesian
	(or convex) combination of these live-edge instances.
%
%
%
%
The next axiom reflects the linearity-of-expectation principle:

\vspace{-1mm}
\begin{axiom}[Bayesian Influence] \label{axiom:bayesian}
For any network $G = (V,E)$ and Bayesian social-influence
  model  $\calI_{\cB(\{\calI^{\eta}\}, \lambda)}$:
\vspace{-2mm}
\begin{eqnarray*}
\psi_v(\calI_{\cB(\{\calI^{\eta}\}, \lambda)}) =
 \E_{\boldeta\sim \lambda} \left[ \psi_v(\calI^{\boldeta})\right]= \sum_{\eta=1}^r \lambda_\eta \cdot\psi_v(\calI^{\eta}), \forall v\in V.
\end{eqnarray*}
\end{axiom}

\vspace{-1mm}
The above axiom essentially says that the centrality of a Bayesian instance before realizing
	the actual model $\calI^\eta$ is the same as the expected centrality after realizing $\calI^\eta$.


The last axiom characterizes the centrality\remove{ measures}
  of a family of simple social-influence instances.
For any $\emptyset \subset R\subseteq U \subseteq V$, a {\em critical set instance} 
  $\calI_{R,U} = (V, E,   P_{\calI_{R, U}})$ 
  is such that: 
  (1) The network $G = (V,E)$ contains a complete directed
  bipartite sub-graph from $R$ to $U\setminus R$, together with 
  isolated nodes $V \setminus U$.
  (2) For all $S \supseteq R$, $P_{\calI_{R,U}}(S, U \cup S) = 1$, and
  (3) For all $S \not \supseteq R$, $P_{\calI_{R,U}}(S, S) = 1$.
  In $\calI_{R,U}$, $R$ is called the {\em critical set}, and
  $U$ is called the {\em target set}.
In other words, a seed set containing $R$ activates all nodes in $U$, but missing any node in $R$
	the seed set only activates itself.
We use $\calI_{R,v}$ to denote the special case of $U=R\cup \{v\}$ and $V= U$.
That is, only if all nodes in $R$ work together they can activate $v$.

\vspace{-1mm}
\begin{axiom}[Bargaining with Critical Sets] \label{axiom:critical}
In any critical set  instance
$\calI_{R,v}$, \remove{ = (R\cup\{v\}, \{(u,v): u\in R\}, P_{\calI_{R, v}})}
the centrality of $v$ 
	is $\frac{|R|}{|R|+1}$, i.e. $\psi_v(\calI_{R,v}) = \frac{|R|}{|R|+1}$. 
\end{axiom}

\vspace{-1mm}
Qualitatively, Axiom~\ref{axiom:critical} together with Normalization and Anonymity axioms implies
	that the relative importance of $v$ comparing to a node in the critial set $R$ increases when
	$|R|$ increases, which is reasonable because when the critical set $R$ grows, individuals in $R$ becomes weaker 
	and $v$ becomes relatively stronger.\OnlyInShort{ The actual quantity can be explained by Nash's solution 
	to the bargaining game~\cite{NashBargining} (see~\cite{CT16}).}
\OnlyInFull{This axiom can be 
  interpreted through  Nash's solution \cite{NashBargining}
  to the bargaining game 
between a player representing the critical set $R$ 
  and the sink node $v$.\remove{as the other player
  in a two-player bargaining game.}
Let $r = |R|$.
Player $R$ can influence all nodes by itself, achieving utility $r+1$, while
	player $v$ can only influence itself, with utility $1$.
The {\em threat point} of this bargaining game is $(r,0)$, 
  which reflects the credits that 
  each player agrees that the other player 
  should at least receive:
	Player $v$ agrees that player $R$'s 
        contribution is at least $r$, while player $R$
	thinks that player $v$ may not have any contribution
         because $R$ can activate  everyone.
The slack in this threat point is  $\Delta = r+1 - (r+0) = 1$.
However, in this case, player $R$ is actually 
  a coalition of $r$ nodes,  and these $r$ nodes have to
   cooperate in order to influence all $r+1$ nodes --- 
  missing any node in $R$ will not influence \remove{all nodes,   particularly will not influence not }$v$.
The need to  cooperative 
  in order to bargain with player $v$ weakens player $R$.
The ratio of $v$'s bargaining weight to that of $R$ is thus $1$ to $1/r$.
Nash's bargaining solution \cite{NashBargining}
  provides a fair division of this slack  between the two players:
\remove{Then, by Nash's bargaining solution the fair division point is:}
$$(x_1,x_2) \in \argmax_{x_1\ge r, x_2\ge0, x_1+x_2=r+1} (x_1-r)^{1/r} \cdot x_2.$$
The unique solution is $(x_1,x_2) = (r+\frac{1}{r+1}, \frac{r}{r+1})$.
Thus, node $v$ should receive a credit of $\frac{r}{r+1}$, 
  as stated in Axiom~\ref{axiom:critical}.
}

Our first axiomatic representation theorem can now be stated as the following:
 
\OnlyInShort{\vspace{-1mm}}
\begin{theorem} {\sc (Axiomatic Characterization of Shapley Centrality)} 
\label{thm:ShapleyCen}
The Shapley centrality $\psi^{\Shapley}$ is 
  the unique centrality measure that satisfies
Axioms \ref{axiom:anonymity}-\ref{axiom:critical}.
Moreover, every axiom in this set is independent of others.
\end{theorem}

The soundness of this representation theorem
  --- that the Shapley centrality satisfies all axioms ---
  is relatively simple.
However, because of the intrinsic complexity in influence models, 
  the uniqueness proof is in fact complex.
We give a high-level proof sketch here and the full proof is 
	in\OnlyInFull{ Appendix~\ref{app:thm2}.}\OnlyInShort{ \cite{CT16}.} We follow Myerson's proof strategy
      \cite{MyersonBook} of Shapley's theorem.
The probabilistic profile $P_{\calI}$ of influence instance 
	$\calI = (V,E,P_{\calI})$ is viewed as a vector in a large space $R^M$, where
	$M$ is the number of independent dimensions in $P_{\calI}$.
Bayesian Influence Axiom enforces that any conforming centrality measure
  is an affine mapping from $R^M$ to $\R^n$.
We then prove that the critical set instances $\calI_{R,U}$ form a 
  full-rank basis of the linear space  $R^M$.
Finally, we prove that any axiom-conforming centrality 
  measure over critical set instances (and the additional null instance in which every node is a sink node)
  must be unique.
The uniqueness of the critical set instances and the null instance, the linear independence
	of critical set instances in $R^M$, plus the affine
	mapping from $R^M$ to $\R^n$, together imply that the centrality measure of every influence instance is
	uniquely determined.
Our overall proof is more complex and --- to a certain degree ---
  more subtle than Myerson's proof, because
  our axiomatic framework is based on the influence model in a much larger dimensional space
  compared to the subset utility functions.
Finally, for independence, we need to show that for each axiom, we can construct an
	alternative centrality measure if the axiom is removed.
Except for Axiom~\ref{axiom:critical}, the constructions and the proofs for other axioms
	are nontrivial, and they shed lights on how related centrality measures could be formed
	when some conditions are relaxed.
  
\subsection{Axioms for SNI Centrality}

We first examine which of 
   Axioms \ref{axiom:anonymity}-\ref{axiom:critical} 
  are satisfied by SNI centrality.
It is easy to verify that  
  Anonymity and Bayesian Influence Axioms hold for SNI centrality.
For the Independence of Sink Node Axiom (Axiom~\ref{axiom:sink}), 
  since every sink node can only influence itself,
  its SNI centrality is $1$. Thus,
  Axiom~\ref{axiom:sink} is satisfied by SNI because of a stronger reason.

For the Normalization Axiom (Axiom~\ref{axiom:normalization}),
   the sum of single node influence
   is typically more than the total number of nodes (e.g., when the influence spread is submodular), 
   and thus Axiom~\ref{axiom:normalization} does not
		hold for SNI centrality. 
The Bargaining with Critical Sets Axiom (Axiom~\ref{axiom:critical}) does not hold either,
	since node $v$ in $\calI_{R,v}$ is a sink node and thus its SNI centrality is $1$.
	
We now present our axiomatic 
  characterization of 
  SNI centrality, which  will
  retain Bayesian Influence Axiom~\ref{axiom:bayesian},
  strengthen Independence of Sink Node Axiom~\ref{axiom:sink},
and recharacterize the centrality of a node in a critical set:

\begin{axiom}	[Uniform Sink Nodes] \label{axiom:uniformsink}
Every sink node has centrality~$1$.
\end{axiom}

\begin{axiom} [Critical Nodes] \label{axiom:criticalnodes}
In any critical set instance $\calI_{R,U}$, 
  the centrality of a node $w\in R$ is $1$ if $|R| > 1$, and
	is $|U|$ if $|R| = 1$.
\end{axiom}

These three axioms are sufficient to uniquely characterize
  SNI centrality, as they also imply Anonymity Axiom:
\OnlyInShort{\vspace{-1mm}}
\begin{theorem} {\sc (Axiomatic Characterization of SNI Centrality)} 
	\label{thm:SNICen}
	The SNI centrality $\psi^{\SNI}$ is 
	the unique centrality measure that satisfies
	Axioms \ref{axiom:bayesian}, \ref{axiom:uniformsink}, and \ref{axiom:criticalnodes}.
	Moreover, each of these axioms is independent of the others.
\end{theorem}

Theorems~\ref{thm:ShapleyCen} and~\ref{thm:SNICen} establish
the following appealing property:
Even though all our axioms are on probabilistic profiles $P_{\calI}$ of influence instances,
	the unique centrality measure satisfying these axioms is 
	in fact fully determined by the influence spread profile $\sigma_{\calI}$.
We find this amazing because 
the distribution profile $P_{\calI}$ has much higher dimensionality 
than its influence-spread profile $\sigma_{\calI}$.


\subsection{Shapley Centrality versus SNI Centrality}

We now  provide a comparative analysis
   between Shapley and SNI centralities based on their
  definitions, axiomatic characterizations, and 
  various other properties they satisfy.

%
\vspace{1mm}
\noindent{\bf Comparison by definition.\ } 
The definition of SNI centrality is more straightforward as 
  it uses individual node's influence spread as the centrality measure.
Shapley centrality is more sophisticatedly formulated,
   involving groups' influence spreads.
SNI centrality disregards the 
  influence profile of groups.
Thus, it may limit its usage in more complex
	situations where group influences should be considered.
Meanwhile, Shapley centrality considers group influence 
  in a particular way involving marginal influence of a node on 
	a given group randomly ordered before the node.
Thus, Shapley centrality is more suitable for 
  assessing marginal influence of a node in a group setting.

%
\vspace{1mm}
\noindent{\bf Comparison by axiomatic characterization.\ }
Both SNI and Shapley centralities
 satisfy Anonymity, Independence of Sink Nodes, 
 and Bayesian Influence axioms,
	which seem to be natural axioms 
	for desirable social-influence centrality measures. 
Their unique axioms characterize exactly their differences.
The first difference is on the Normalization Axiom, satisfied by Shapley but not SNI centrality.
This indicates that Shapley centrality aims at dividing the total share of possible influence spread $|V|$
	among all nodes, but SNI centrality does not enforce such share division among nodes.
If we artificially normalize the SNI centrality values of all nodes 
	to satisfy the Normalization Axiom, 
	the normalized SNI centrality would not satisfy
	the Bayesian Influence Axiom.
(In fact, it is not easy to find a new characterization for
	the normalized SNI centrality similar to Theorem~\ref{thm:SNICen}.)
We will see shortly that the Normalization Axiom would also cause a drastic difference
	between the two centrality measures for the symmetric IC influence
  model.

The second difference is on their treatment of sink nodes, exemplified by sink nodes in the critical set instances.
For SNI centrality, sink nodes are always treated with the same centrality of $1$ 
	(Axiom~\ref{axiom:uniformsink}).
But the Shapley centrality of a sink node may be affected by other nodes that influence the sink.
In particular, for the critical set instance $\calI_{R,v}$, $v$ has centrality $|R|/(|R|+1)$, 
	which increases with $R$.
As discussed earlier, larger $R$ indicates $v$ is getting stronger
	comparing to nodes in $R$.
In this aspect, 
 Shapley centrality assignment is sensible.
Overall, when considering $v$'s centrality, SNI centrality disregards other nodes' influence to $v$
	while Shapley centrality considers other nodes' influence to $v$.

The third difference is their treatment of critical nodes in the critical set instances.
For SNI centrality, in the critical set instance $\calI_{R,v}$, 
	Axiom~\ref{axiom:criticalnodes} 
obliviously assigns the same value $1$
	for nodes $u\in R$ whenever $|R|>1$, effectively equalizing
	the centrality of node $u\in R$ with $v$.
In contrast, Shapley centrality would assign $u\in R$ a value of $1+ \frac{1}{|R|(|R|+1)}$, decreasing with $R$ but
	is always larger than $v$'s centrality of $\frac{|R|}{|R|+1}$.
Thus Shapley centrality assigns more sensible values in this case, because $u\in R$ as part of a coalition
	should have larger centrality than $v$, who has no influence power at all.
%
We believe this shows the limitation of the SNI centrality --- it only considers individual influence and
	disregards group influence.
Since the critical set instances reflect the threshold behavior in influence propagation --- a node would be
	influenced only after the number of its influenced neighbors reach certain threshold --- this suggests that
	SNI centrality could be problematic in threshold-based influence models.

%
\vspace{1mm}
\noindent{\bf Comparison by additional properties.\ }
Finally, we compare additional properties they satisfy.
First, it is straightforward to verify that both centrality measures satisfy
	the {\em Independence of Irrelevant Alternatives (IIA))} property:
If an instance $\calI = (V,E,P_{\calI})$ 
is the {\em union of two disjoint and independent influence instances}, 
$\calI_1 = (V_1,E_1,P_{\calI_1})$ and $\calI_2 = (V_2,E_2,P_{\calI_2})$,
then for $k\in \{1,2\}$ and any $v \in V_k$:
$\psi_v({\calI}) = \psi_v(\calI_k).$

The IIA property together with the Normalization Axiom leads to a clear difference between SNI and Shapley
	centrality.
Consider an example of two undirected and connected graphs $G_1$ with $10$ nodes and $G_2$ with $3$ nodes, 
	and	the IC model on them with edge probability $1$.
Both SNI and Shapley centralities assign same values to nodes within each graph, but due to normalization,
	Shapley assigns $1$ to all nodes, while SNI assigns $10$ to nodes in $G_1$ and $3$ to nodes in $G_2$.
The IIA property ensures that the centrality does not change when we put $G_1$ and $G_2$ together.
That is, SNI considers nodes in $G_1$ more important while Shapley considers them the same.
While SNI centrality makes sense from individual influence point of view, the view of Shapley centrality 
	is that a node in $G_1$ is easily replaceable by any of the other $9$ nodes in $G_1$ but a node in $G_2$
	is only replaceable by two other nodes in $G_2$.
Shapley centrality uses marginal influence in randomly ordered groups to
	determine that the ``replaceability factor'' cancels out individual influence and
	assigns same centrality to all nodes.
	
The above example generalizes to the symmetric IC model where
	$p_{u,v} = p_{v,u}$, $\forall u,v \in V$: 
{\em Every node has Shapley centrality of $1$ in such models.}
The technical reason is that such models have an equivalent {\em undirected} live-edge graph representation, 
	containing a 
  number of connected components just like the above example.
The Shapley symmetry in the 
  symmetric IC model may sound counter-intuitive, since it 
  appears to be independent of network structures or edge probability values.
But we believe what it unveils is that symmetric IC model might be an unrealistic model in practice ---
	it is hard to imagine that between every pair of individuals the influence strength is symmetric.
For example, in a star graph, when we perceive that the node in the center has higher centrality, it is
	not just because of its center position, but also because that it typically exerts higher influence
	to its neighbors than the reverse direction.
This exactly reflects our original motivation that mere positions in a static network may not be an important
	factor in determining the node centrality, and 
	what important is the effect of individual nodes participating in the dynamic influence~process. 

From the above discussions, we clearly see that (a) SNI centrality focuses on individual influence in isolation, while
	(b) Shapley centrality focuses on marginal influence in group influence settings, and measures 
		the {\em irreplaceability} of the nodes in some sense.

\OnlyInShort{\vspace{-2mm}}
\section{Scalable~Algorithms}
\label{sec:ScalableAlgorithm}
In this section, we first give a sampling-based algorithm for approximating 
  the Shapley centrality 
 $\psi^{\Shapley}(\calI)$ of any influence instance in the triggering model.
We then give a slight adaptation to approximate SNI centrality.
In both cases, we characterize the performance 
  of our algorithms and prove that they are scalable
  for a large family of social-influence instances.
In next section, we empirically show that these algorithms are 
  efficient for real-world networks.
\subsection{Algorithm for Shapley Centrality}

In this subsection, 
 we use $\psi$ as a shorthand for 
 $\psi^{\Shapley}$. 
Let $n=|V|$ and $m=|E|$.
To precisely state our result, 
  we make the following general computational assumption, as in~\cite{tang14,tang15}: 
\vspace{-1mm}
\begin{assumption}
	\label{assump:ComputationalTriggerModel}
	The time to draw a random
	triggering set $\bT(v)$ is proportional to the in-degree of $v$.
\end{assumption}

\OnlyInShort{\vspace{-1mm}}

The key combinatorial structures that we use
  are the following random sets generated by the {\em reversed diffusion process} of the triggering model.
A {\em (random) reverse reachable (RR) set} 
  $\bR$ is generated as follows:
(0) Initially, $\bR = \emptyset$. 
(1) Select a node $\bv \sim V$ uniformly at random (called the {\em root} of $\bR$), 
	and add $\bv$ to $\bR$.
(2) Repeat the following process until every node in $\bR$ has a triggering
    set:  For every $\bu \in \bR$ not yet having a triggering set, 
   draw its random triggering set $\bT(\bu)$, and add $\bT(\bu)$ to $\bR$.
Suppose $\bv\sim V$ is selected in Step (1).
The reversed diffusion process uses $\bv$ as the seed, 
and 
  follows the incoming edges instead of the outgoing edges to iteratively
  ``influence'' triggering sets. 
Equivalently, an RR set $\bR$ is the set of nodes in a random live-edge graph $\bL$ that can reach
node $\bv$.

The following key lemma elegantly connects 
RR sets with Shapley centrality.
We will defer its intuitive explanation to the end of this section.
Let $\bpi$ be a random permutation on $V$.
Let $\I\{{\mathcal{E}}\}$ be the indicator function for event~$\mathcal{E}$.
\vspace{-2mm}
\begin{lemma}[Shapley Centrality Identity] \label{lem:spexpMainbody}
Let $\bR$ be a random RR set.
Then, $\forall u\in V$, $u$'s Shapley centrality is $\psi_u = n\cdot \E_{\bR}[\I\{u\in \bR\}/|\bR|]$.
\end{lemma}
\vspace{-1mm}
This lemma is instrumental to our scalable algorithm.
It guarantees that we can use random RR sets 
 to build {\em unbiased estimators} of Shapley centrality.
Our algorithm {\ASVRR}  (standing for ``Approximate Shapley Value by RR Set'')
  is presented in Algorithm~\ref{alg:rrshapleynew}.
It takes $\varepsilon$, $\ell$, and $k$ as input parameters, 
representing the relative error, the confidence of the error, and the number of nodes
with top Shapley values that achieve the error bound, respectively.
Their exact meaning will be made clear in Theorem~\ref{thm:ASVRR}.

\begin{algorithm}[t]
	\centering
	\caption{\ASVRR $(G, \bT, \varepsilon,\ell, k)$} \label{alg:rrshapleynew}
	\begin{algorithmic}[1]
		\INPUT{Network: $G=(V,E)$; Parameters:  random triggering set distribution $\{\bT(v)\}_{v\in V}$, 
			$\varepsilon > 0$, $\ell>0$, $k\in [n]$
}
		\OUTPUT{$\epsi_v$,  $\forall v\in V$: estimated centrality measure 
		}
		\STATE \{Phase 1. Estimate the number of RR sets needed \}
		\STATE $\LB = 1$; $\varepsilon' = \sqrt{2} \cdot \varepsilon$; $\theta_0 = 0$
		\STATE $\est_v = 0$ for every $v\in V$
		\FOR {$i=1$ to $ \lfloor \log_2 n \rfloor - 1$} \label{line:phase1forb}
		\STATE $x = n / 2^i$
		\STATE $\theta_i = \left\lceil \frac{ n \cdot 
			((\ell + 1)\ln n + \ln \log_2 n + \ln 2) \cdot (2+\frac{2}{3}\varepsilon')}
		{\varepsilon'^2 \cdot x} \right \rceil $ \label{line:setthetai}
		\FOR {$j = 1$ to $\theta_i - \theta_{i-1}$}
		\STATE generate a random RR set $\bR$
		\STATE for every $\bu\in \bR$, $\est_{\bu} = \est_{\bu} + 1/|\bR|$ \label{line:estnew1}
		\ENDFOR
		\STATE $\est^{(k)} = $ the $k$-th largest value in $\{\est_v\}_{v\in V}$ \label{line:kmaxest}
		\IF{ $ n\cdot \est^{(k)} / \theta_i \ge (1+\varepsilon') \cdot x$} \label{line:cond1}
		\STATE $ \LB = n\cdot \est^{(k)} / (\theta_i \cdot (1+\varepsilon'))$ \label{line:setLB}
		\STATE {\bf break}
		\ENDIF
		\ENDFOR \label{line:phase1fore}
		\STATE $\btheta = \left\lceil \frac{n ((\ell+1)\ln n + \ln 4)(2+ \frac{2}{3} \varepsilon) }{\varepsilon^2 \cdot \LB} \right\rceil$ \label{line:thetanew}
		\STATE \{Phase 2. Estimate Shapley value\}
		\STATE $\est_v = 0$ for every $v\in V$ \label{line:resetest}
		\FOR {$j = 1$ to $\btheta$ }
		\STATE generate a random RR set $\bR$
		\STATE for every $\bu\in \bR$, $\est_{\bu} = \est_{\bu} + 1/|\bR|$ \label{line:estnew2}
		\ENDFOR
		\STATE for every $v\in V$, $\epsi_v = n \cdot \est_v / \btheta$ 
		\label{line:adjustnew2}
		\STATE return $\epsi_v$, $v\in V$
	\end{algorithmic}
\end{algorithm}

{\ASVRR} follows the structure of the {\IMM} algorithm of \cite{tang15} but with some
	key differences.
In Phase 1, Algorithm~\ref{alg:rrshapleynew} estimates 
  the number of RR sets needed for the Shapley estimator.
For a given parameter $k$, we first estimate a lower bound $\LB$ of
	the $k$-th largest Shapley centrality $\psi^{(k)}$.
Following a similar structure as the sampling method in {\IMM} \cite{tang15},
	the search of the lower bound is carried out in at most $ \lfloor \log_2 n \rfloor - 1$ iterations,
	each of which halves the lower bound target $x = n/2^i$ and obtains the number of RR sets $\theta_i$
	needed in this iteration (line~\ref{line:setthetai}).
The key difference is that we do not need to store the RR sets and  compute a max cover.
Instead, for every RR set $\bR$, we only update the estimate $\est_{\bu}$ of each node $\bu \in \bR$ 
	with an additional $1/|\bR|$ (line~\ref{line:estnew1}), which is based on Lemma~\ref{lem:spexpMainbody}.
In each iteration, we select the $k$-th largest estimate (line~\ref{line:kmaxest}) and
	plug it into the condition in line~\ref{line:cond1}.
Once the condition holds, we calculate the lower bound $\LB$ in
	line~\ref{line:setLB} and break the loop.
Next we use this $\LB$ to obtain the number of RR sets $\btheta$ needed in Phase 2 (line~\ref{line:thetanew}).
In Phase 2, we first reset the estimates (line~\ref{line:resetest}),
	then generate $\btheta$ RR sets and again updating $\est_u$ with $1/|\bR|$ increment
	for each $\bu \in \bR$ (line~\ref{line:estnew2}).
Finally, these estimates are transformed into the Shapley estimation in line~\ref{line:adjustnew2}.

Unlike {\IMM}, we do not reuse the RR sets generated in Phase 1, because it would
	make the RR sets dependent and the resulting Shapley centrality estimates biased.
Moreover, our entire algorithm does not need to store any RR sets, and thus {\ASVRR} does not have the
	memory bottleneck encountered by {\IMM} when dealing with large networks.
The following theorem summarizes the performance of Algorithm~\ref{alg:rrshapleynew}, where
	$\psi$ and $\psi^{(k)}$ are Shapley centrality and $k$-th largest
	Shapley centrality value, respectively.

\vspace{-2mm}
\begin{theorem} 
	\label{thm:ASVRR}
For any $\epsilon > 0$, $\ell > 0$, and $k\in [n]$, 
Algorithm {\ASVRR} 
  returns an estimated Shapley value $\epsi_v$ that satisfies
(a) unbiasedness: $\E[\epsi_v] = \psi_v, \forall v\in V$;
(b) absolute normalization: $\sum_{v\in V} \epsi_v = n$ in every run; and
(c) robustness: under the condition that $\psi^{(k)}\ge 1$, with probability at least $1-\frac{1}{n^\ell}$:
\vspace{-1mm}
	\begin{equation} \label{eq:relativeerror}
	\left\{ 
	\begin{array}{lr}
	|\epsi_v - \psi_v| \le \varepsilon \psi_v & \forall v\in V \mbox{ with } \psi_v > \psi^{(k)},\\
	|\epsi_v - \psi_v| \le \varepsilon \psi^{(k)} & \forall v\in V \mbox{ with } \psi_v \le \psi^{(k)}.
	\end{array}
	\right.
	\end{equation}
Under Assumption \ref{assump:ComputationalTriggerModel} and the condition
  $\ell \ge (\log_2 k - \log_2 \log_2 n)/\log_2 n$,
	the expected running time 
	of {\ASVRR} is $O(\ell (m+n) \log n \cdot \E[\sigma(\tilde{\bv})]/ (\psi^{(k)} \varepsilon^2))$,
	where $\E[\sigma(\tilde{\bv})]$ is the expected influence spread of a random 
	node $\tilde{\bv}$ drawn from $V$ with probability proportional to the 
	in-degree of $\tilde{\bv}$.
\end{theorem}

\vspace{-1mm}
Eq.~\eqref{eq:relativeerror} above shows that for the top $k$ Shapley values,
	{\ASVRR} guarantees the multiplicative error of $\varepsilon$ relative to
	node's own Shapley value\OnlyInFull{ (with high probability)}, and
	for the rest Shapley value, the error is relative to the $k$-th largest Shapley value
	$\psi^{(k)}$.
This is reasonable since typically we only concern nodes with top Shapley values.
For time complexity, the condition $\ell \ge (\log_2 k - \log_2 \log_2 n)/\log_2 n$ always hold
	if $k \le \log_2 n$ or $\ell \ge 1$.
When fixing $\varepsilon$ as a constant, the running time depends almost linearly on the graph size ($m+n$)
	multiplied by a ratio $\E[\sigma(\tilde{\bv})]/\psi^{(k)}$.
This ratio is upper bounded by the ratio between the largest single node influence and the $k$-th
	largest Shapley value. 
	When these two quantities are about the same order, we have 
  a near-linear time, i.e., scalable \cite{TengScalable}, algorithm. 
Our experiments show that in most datasets tested the ratio $\E[\sigma(\tilde{\bv})]/\psi^{(k)}$ is 
	indeed less than $1$.
Moreover, if we could relax the robustness requirement in Eq.~\eqref{eq:relativeerror} to allow
	the error of $|\epsi_v - \psi_v| $ to be relative to the largest single node influence, 
	then we could indeed slightly modify the 
  algorithm to obtain a near-linear-time algorithm
	without the ratio $\E[\sigma(\tilde{\bv})]/\psi^{(k)}$  in the time complexity
	\OnlyInShort{(see~\cite{CT16}).}\OnlyInFull{(see Appendix~\ref{app:nearlt}).}


The accuracy of {\ASVRR} is based on Lemma~\ref{lem:spexpMainbody} while the time complexity analysis
	follows a similar structure as in~\cite{tang15}.
\OnlyInShort{Due to space limit, the proofs of Lemma \ref{lem:spexpMainbody} 
  and Theorems~\ref{thm:ASVRR} are presented in our full 
  report~\cite{CT16}.}\OnlyInFull{The proofs of Lemma \ref{lem:spexpMainbody} 
  and Theorem~\ref{thm:ASVRR} are presented in Appendix~\ref{app:thm1}.}
Here, we give a high-level explanation.
In the triggering model, 
 as for influence 
  maximization~\cite{BorgsBrautbarChayesLucier,tang14,tang15},
  a random RR set $\bR$ 
  can be equivalently obtained by first
  generating a random live-edge graph $\bL$,
  and then constructing $\bR$  as the set of nodes that can 
  reach a random $\bv \sim V$ in $\bL$.
The fundamental equation associated with this 
  live-edge graph process is:
\begin{equation} \label{eq:keyequation}
\sigma(S) = \sum_L \Pr_{\bL}(\bL = L) \Pr_{\bv}(\bv \in \Gamma(L, S)) \cdot n.
\end{equation}
Our Lemma \ref{lem:spexpMainbody}
 is the result of the following crucial observations:
First, the Shapley centrality $\psi_u$  of node $u\in V$
  can be equivalently formulated as 
  the expected Shapley centrality of $u$
  over all live-edge graphs and random choices of root $\bv$, from Eq.~\eqref{eq:keyequation}.
The chief advantage of this formulation is that
  it localizes the contribution of marginal influences:
On a fixed live-graph $L$ and root $v\in V$, 
  we only need to compute the marginal influence of $u$ 
  in terms of activating $v$ to obtain the Shapley contribution 
  of the pair.
We do not need to compute the marginal influences of $u$ for activating 
  other nodes.
Lemma \ref{lem:spexpMainbody} then follows from our 
  second crucial observation.
When $R$ is the fixed set that can reach $v$ in $L$, 
  the marginal influence of $u$ activating $v$ in a random order is $1$ 
  if and only if the following two conditions hold concurrently:
(a) $u$ is in $R$ --- so $u$ has chance to activate $v$, and
(b) $u$ is ordered before any other node in $R$ --- so $u$ can activate $v$ before other nodes
	in $R$ do so.
In addition, in a random permutation $\bpi \sim \Pi$ over $V$, 
  the probability that $u\in R$ is ordered first in $R$ is 
  exactly $1/|R|$.
This explains the contribution of $\I\{u\in \bR\}/|\bR|$ in 
  Lemma \ref{lem:spexpMainbody}, which is also 
 precisely what the updates 
  in lines~\ref{line:estnew1} and~\ref{line:estnew2} of Algorithm~\ref{alg:rrshapleynew} do.
The above two observations together establish Lemma \ref{lem:spexpMainbody},
  which is the basis for the unbiased estimator of $u$'s Shapley centrality.
Then, by a careful probabilistic analysis, 
  we can bound the number of random RR sets needed 
  to achieve approximation accuracy stated in Theorem \ref{thm:ASVRR} 
  and establish the scalability for Algorithm {\ASVRR}.

\subsection{Algorithm for SNI Centrality}

Algorithm~\ref{alg:rrshapleynew} relies on the key fact given in Lemma~\ref{lem:spexpMainbody} about the
	Shapley centrality:  $\psi^{\Shapley}_u = n\cdot \E_{\bR}[\I\{u\in \bR\}/|\bR|]$.
A similar fact holds for the SNI centrality: $\psi^{\SNI}_u = \sigma(\{u\}) = n\cdot \E_{\bR}[\I\{u\in \bR\}]$
	\cite{BorgsBrautbarChayesLucier,tang14,tang15}.
Therefore, it is not difficult to verify that we only need to replace $\est_{\bu} = \est_{\bu} + 1/|\bR|$
	in lines~\ref{line:estnew1} and~\ref{line:estnew2} with $\est_{\bu} = \est_{\bu} + 1$ to obtain
	an approximation algorithm for SNI centrality.
Let {\ASNIRR} denote the algorithm adapted from {\ASVRR} with the above change, and let
	$\psi_v$ below denote SNI centrality $\psi^{\SNI}_v$ and $\psi^{(k)}$ denote the $k$-th largest
	SNI value.
	
\OnlyInShort{\vspace{-1mm}}
\begin{theorem} 
	\label{thm:SNI}
	For any $\epsilon > 0$, $\ell > 0$, and $k\in \{1,2,\ldots, n \}$,
	Algorithm {\ASNIRR} 
	returns an estimated SNI centrality $\epsi_v$ that satisfies
	(a) unbiasedness: $\E[\epsi_v] = \psi_v, \forall v\in V$; and
	(b) robustness:  with probability at least $1-\frac{1}{n^\ell}$:
	\begin{equation} \label{eq:relativeerrorsni}
	\left\{ 
	\begin{array}{lr}
	|\epsi_v - \psi_v| \le \varepsilon \psi_v & \forall v\in V \mbox{ with } \psi_v > \psi^{(k)},\\
	|\epsi_v - \psi_v| \le \varepsilon \psi^{(k)} & \forall v\in V \mbox{ with } \psi_v \le \psi^{(k)}.
	\end{array}
	\right.
	\end{equation}
	Under Assumption \ref{assump:ComputationalTriggerModel} and the condition
	$\ell \ge (\log_2 k - \log_2 \log_2 n)/\log_2 n$,
	the expected running time 
	of {\ASNIRR} is $O(\ell (m+n) \log n \cdot \E[\sigma(\tilde{\bv})]/ (\psi^{(k)} \varepsilon^2))$,
	where $\E[\sigma(\tilde{\bv})]$ is the same as defined in Theorem~\ref{thm:ShapleyCen}.
\end{theorem}

\OnlyInShort{\vspace{-1mm}}
Together with Algorithm~{\ASVRR} and Theorem~\ref{thm:ASVRR}, we see that although Shapley and SNI centrality are quite
	different conceptually, surprisingly they share the same RR-set based scalable computation structure.
Comparing Theorem~\ref{thm:SNI} with Theorem~\ref{thm:ASVRR}, we can see that computing SNI 
	centrality should  be faster for small $k$ since the $k$-th largest SNI value is usually larger than
	the $k$-th largest Shapley value.

\section{Experiments}\label{sec:experiments}
We conduct experiments on a number of 
 real-world social networks to compare their Shapley and SNI
 centrality, and test the efficiency of
 our algorithms {\ASVRR} and {\ASNIRR}.
	
\subsection{Experiment Setup}
\begin{table}[t]
	{\centering \caption{\label{tab:stat} Datasets used in the experiments. 
			\OnlyInFull{``\# Edges''
			refers to the number of undirected edges for the first three datasets and
			the number of directed edges for the last dataset.}}
	\small
	\begin{tabular}{|l|r|r|c|}
		\hline 
		Dataset & \# Nodes & \# Edges & Weight Setting \\
		\hline
		\hline
		Data mining (DM) & 679 & 1687 & WC, PR, LN \\
		Flixster (FX)  & 29,357 & 212,614 & LN \\
\OnlyInFull{DBLP (DB)  & 654,628 & 1,990,159 & WC, PR \\
		}LiveJournal (LJ)  & 4,847,571 & 68,993,773 & WC \\
		\hline
	\end{tabular}}
\end{table}
The network datasets we used are summarized in Table~\ref{tab:stat}. 

The first dataset is a relatively small one used as a case study.
It is a collaboration network in the field of Data Mining (DM), 
  extracted from the ArnetMiner archive (arnetminer.org) \cite{ChiKDD09}: each node is an author
  and two authors are connected if they have coauthored a paper.
\OnlyInFull{The mapping from node ids to author names is available, 
 allowing us to gain some intuitive observations
 of the centrality measure. }We use \OnlyInFull{three}\OnlyInShort{two} large networks to 
  demonstrate the effectiveness of the Shapley and SNI centrality
  and the scalability of our algorithms.
Flixster (FX) \cite{barbieri2012topic} is a directed network extracted from movie rating site flixster.com.
The nodes are users and 
  a directed edge from $u$ to $v$ means that $v$ has rated some movie(s) 
  that $u$ rated earlier.
\OnlyInFull{Both network and the influence probability profile
   are obtained from the authors of \cite{barbieri2012topic},
   which shows how to learn topic-aware influence 
   probabilities. }We use influence probabilities on topic 1 
  in their provided data as an example.
\OnlyInFull{DBLP (DB) is another academic collaboration network extracted from online 
		archive DBLP (dblp.uni-trier.de)
	  and used for influence studies 
	  in \cite{wang2012scalable}. }Finally, LiveJournal (LJ) is the largest network we tested with.
It is a directed network of bloggers, 
  obtained from Stanford's SNAP project~\cite{SNAP}, and it was
	also used in \cite{tang14,tang15}.

	\begin{table*}[t] 
		\caption{Top 10 authors from DM dataset, ranked by Shapley, SNI, and degree centrality. }
		{\includegraphics[width=1\linewidth]{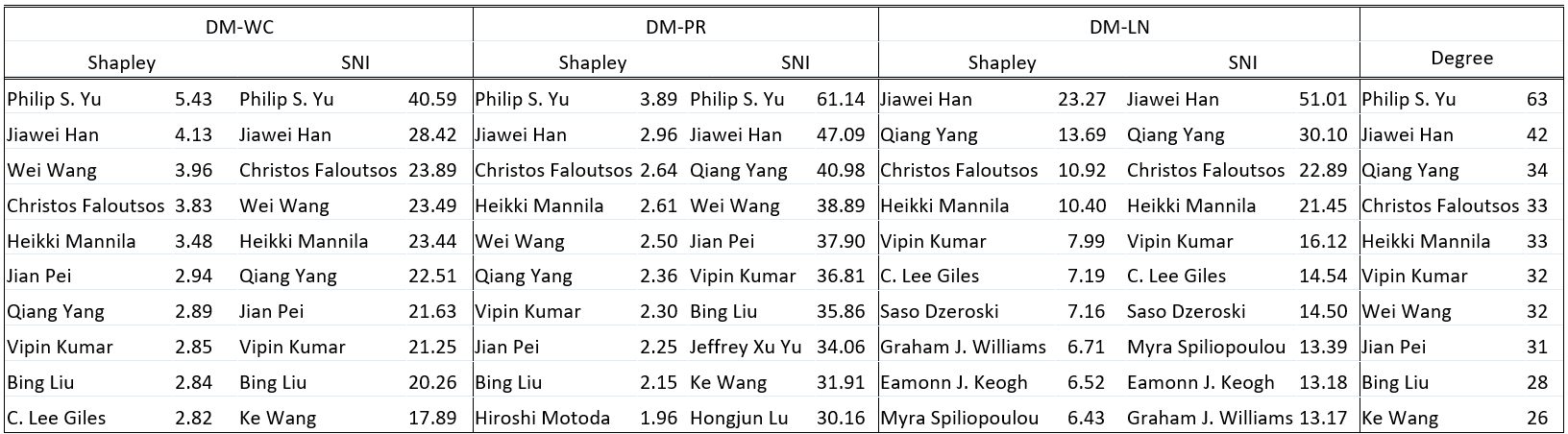}\label{tab:dmlist}}
	\end{table*}

We use the independent cascade (IC) model in our experiments.
The schemes for generating influence-probability 
  profiles are also shown in Table~\ref{tab:stat},
   where WC, PR,  and LN stand for {\em weighted cascade}, {\em PageRank-based}, and {\em learned from real data},
   respectively.
WC is a scheme of \cite{kempe03},
  which assigns $p_{u,v} = 1/d_v$ to edge $(u,v)\in E$,
  where $d_v$ is the in-degree of node $v$.
PR uses the nodes' PageRanks \cite{PageRank}
	instead of in-degrees:
We first compute the PageRank score $r(v)$ for every node $v \in V$
  in the unweighted network, using $0.15$ as the restart parameter.
\OnlyInFull{Note that in our influence network, edge $(u,v)$ means $u$ has influence to $v$; then
	when computing PageRank, we should reverse the edge direction to $(v,u)$ so that
	$v$ gives its PageRank vote to $u$, in order to be consistent
	on influence direction. }Then, for each original edge $(u,v)\in E$,  PR assigns an edge probability of $r(u)/(r(u)+r(v)) \cdot n/(2m^U)$,
	where $m^U$ is the number of undirected edges in the graph.
\OnlyInFull{The assignment achieves the effect that a higher PageRank node has larger influence to a
	lower PageRank nodes than the reverse direction (when both directions exist).
The scaling factor $n/(2m^U)$ is to normalize the total edge probabilities 
  to $\Theta(n)$, which is similar	to the setting of WC.
PR defines a PageRank-based asymmetric IC model. }LN applies to DM and FX datasets,
   where we obtain learned influence probability profiles from 
  the authors of the original studies.
For the DM dataset, the influence probabilities on edges are
	learned by the topic affinity algorithm TAP proposed in 
	\cite{ChiKDD09};
	for FX, the influence probabilities are learned using maximum likelihood from the 
	action trace data of user rating events.

We implement all algorithms 
  in Visual C++, compiled in Visual Studio 2013, and 
  run our tests on a server computer with 2.4GHz Intel(R) Xeon(R)
	E5530 CPU, 2 processors (16 cores), 48G memory, and Windows Server 2008 R2
	(64 bits).

\OnlyInShort{An additional dataset DBLP with WC and PR settings is 
		also tested and results are included in~\cite{CT16}.}
	

\OnlyInShort{\vspace{-2mm}}
\subsection{Experiment Results}

\noindent{\bf Case Study on DM.\ }
We set $\varepsilon = 0.01$, $\ell = 1$, and $k=50$ for both {\ASVRR} and 
and {\ASNIRR} algorithms.
For the three influence profiles: WC, PR, and LN, 
	Table~\ref{tab:dmlist} lists the 
  top 10 nodes in both Shapley and SNI ranking together with their numerical values. 
The names appeared in all ranking results are well-known 
  data mining researchers in the field,\OnlyInFull{ at the time of the data collection 2009, }but the ranking details have some difference.
  
We compare the Shapley ranking versus SNI ranking under the same probability profiles.
In general, the two top-10 ranking results 
align quite well with each other, showing that
in these influence instances, high individual influence 
usually translates into high marginal influence.
Some noticeable exception also exists.
For example, Christos Faloutsos is ranked No.3 
in the DM-PR Shapley centrality, but he is not in Top-10 based on DM-PR individual influence ranking.
Conceptually, this would mean that, in the DM-PR model,  Professor Faloutsos 
has better Shapley ranking  because he has more unique and marginal impact 
comparing to his individual influence.
\OnlyInFull{In terms of the numerical values, SNI values are larger than the Shapley values, which is expected
	due to the normalization factor in Shapley centrality.}

We next compare Shapley and SNI centrality with the structure-based degree centrality.
The results show that the Shapley and SNI rankings in DM-WC and DM-PR are similar to the degree centrality ranking, which
	is reasonable because DM-WC and DM-PR are all heavily derived from node degrees.
However, DM-LN differs from degree ranking a lot, since it is derived from topic modeling, not node degrees.
This implies that when the influence model parameters are learned from real-world data, it may contain further information such that
	its influence-based Shapley or SNI ranking may differ from structure-based ranking significantly.

\OnlyInFull{
When comparing the numerical values of the same centrality measure but across different influence models, 
	we see that Shapley values of top researchers in DM-LN are much higher
	than Shapley values of top researchers under DM-WC or DM-PR, which suggests that influence models
	learned from topic profiles differentiating nodes more than the synthetic WC or PR methods.
	
The above results differentiating DM-LN from DM-WC and DM-PR clearly demonstrate the interplay between social influence 
and network centrality: Different influence processes can lead to different centrality rankings,
	but when they share some aspects of common ``ground-truth'' influence, 
their induced rankings are more closely correlated.}

\OnlyInFull{
\subsubsection*{\sc Tuning Parameter $\varepsilon$}

\begin{figure} [t]
	\centering
	\captionsetup[subfigure]{font=scriptsize,oneside,margin={0.6cm,0.0cm}}
	\setcounter{subfigure}{0}%
	\subfloat[\textit{Shapley computation}]
	{\includegraphics[width=0.50\linewidth]{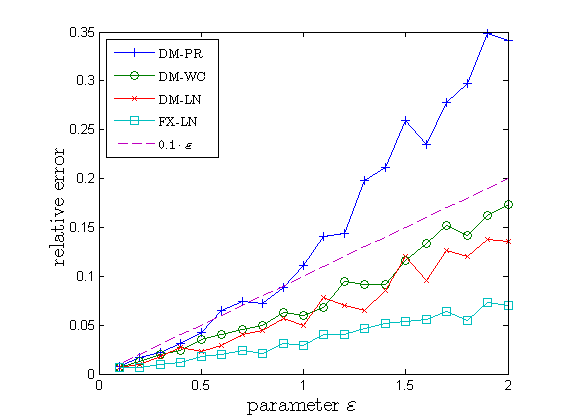}}
	\subfloat[\textit{SNI computation}]
	{\includegraphics[width=0.50\linewidth]{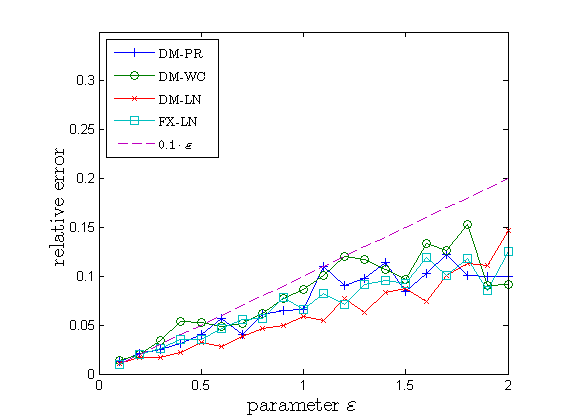}}
	\caption{Relative error of centrality computation when $\varepsilon$ setting increases.}
	\label{fig:epsilon}
\end{figure}

We now investigate the impact of our {\ASVRR}/{\ASNIRR}
  parameters, to be applied to our tests on large datasets.
Parameter $\ell$ is a simple parameter controlling 
  the probability, $1-\frac{1}{n^\ell}$, that the accuracy guarantee holds.
We set it to $1$, which is the same as in \cite{tang14,tang15}.
For parameter $\varepsilon$, a smaller value improves accuracy
   at the cost of higher running time.
Thus, we want to set $\varepsilon$ at a proper level
   to balance accuracy and efficiency.

We test different $\varepsilon$ values 
  from $0.1$ to $2$, on both DM and FX datasets, for both algorithms. 
To evaluate the accuracy, we use the results 
  from $\varepsilon^* = 0.01$ as the benchmark:
For $v\in V$, suppose $s^*_v$ and $s_v$ 
  are the Shapley values computed for $\varepsilon^* = 0.01$ 
  and a larger $\varepsilon$ value, respectively. 
Then, we compute $|s_v - s^*_v|/ s^*_v$ and use it as the relative error at $v$.
Since the top rankers' relative errors are more important, 
  we take top 50 nodes from the two ranking results (using $\varepsilon^*$ and $\varepsilon$ respectively), and 
  compute the average relative error over the union of these two sets 
  of top 50 nodes.
Accordingly, we set parameter $k=50$.
We also apply the same relative error computation to SNI centrality.

Figure~\ref{fig:epsilon} reports our results on the three DM options 
  and the FX dataset, for both Shapley and SNI computations.
We can see clearly that
	when $\varepsilon \le 0.5$, the relative errors of all datasets are within
	$0.05$.
In general, the actual relative error is below one tenth of $\varepsilon$ in most cases, except for
	DM-PR dataset with $\varepsilon \ge 1$.
Hence, for the tests on large datasets, we use $\varepsilon=0.5$ to provide reasonable accuracy for
	top values.
Comparing to $\varepsilon=0.01$, this reduces the running time $2500$ fold, because
	the running time is proportional to $1/\varepsilon^2$.
}

\vspace{2mm}
\noindent{\bf Results on Large Networks.\ }
\OnlyInShort{We set $\varepsilon = 0.5$, $\ell = 1$, and $k=50$ for this test, where
		$\varepsilon=0.5$ is obtained with an omitted tuning step~\cite{CT16}.
Due to the lack of user profiles, we assess the effectiveness of Shapley and SNI centralities
	through lens of influence maximization (IM), by comparing the IM performance of their top-ranked nodes with
	nodes selected by the {\IMM} algorithm \cite{tang15}.
A simple heuristic {\Degree}  based on degree centrality is also compared as a baseline.

Figure~\ref{fig:infmax} shows the influence spread results and Table~\ref{tab:time} shows
	the running time.
We see that both Shapley and SNI have similar IM performance and is reasonable close to
	the specialized {\IMM} algorithm, but Shapley is noticeably better 
	(average $8.3\%$ improvement) than SNI in Flixster-LN test.
This is perhaps due to that Shapley centrality accounts for more marginal influence, which is closer to what
is needed for influence maximization.
In FX-LN, both {\ASNIRR} and {\ASVRR} performs significantly better than {\Degree}, again indicating that
	influence learned from the real-world data may contain significantly more information than the graph
	structure, in which case degree centrality is not a good index for node importance.
Both {\ASNIRR} and {\ASVRR} can scale to the large LiveJournal group with $69M$ edges, and
	{\ASNIRR} scales better as predicted by Theorems~\ref{thm:ASVRR} and~\ref{thm:SNI}.
}\OnlyInFull{
We conduct experiments to evaluate both the effectiveness 
  and the efficiency of our {\ASVRR} algorithm on large networks.
For large networks, it is no longer easy to inspect rankings manually, 
  especially when these datasets lack user profiles.
For the effectiveness, we 
  assess the effectiveness of Shapley and SNI centrality 
  rankings through the lens of influence maximization.
In particular,  we use top rankers of Shapley/SNI
   centrality as seeds and measure their effectiveness for influence maximization.
We compare the quality and performance of our algorithm 
  with  the state-of-the-art scalable algorithm {\IMM} proposed in \cite{tang15} 
  for influence maximization.
Note that the {\IMM} algorithm is based on the RR set approach.
For {\IMM}, we set its parameters as $\varepsilon = 0.5$, $\ell = 1$, and $k=50$,
   matching the parameter settings we use for	{\ASVRR}/{\ASNIRR}.
We also choose a baseline algorithm {\Degree}, 
  which is based on degree centrality to select top degree nodes as seeds for influence maximization.
}

\OnlyInShort{
\begin{figure}[t]
	\centering
	\captionsetup[subfigure]{font=scriptsize,oneside,margin={0.6cm,0.0cm}}
	\setcounter{subfigure}{0}%
	\subfloat[\textit{Flixster-LN}]
	{\includegraphics[width=0.5\linewidth]{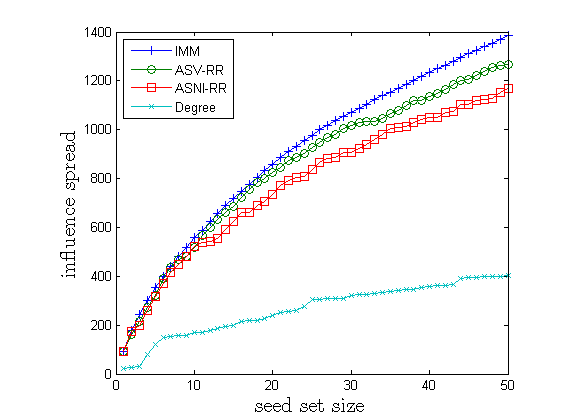}\label{fig:fxinfmax}}
	\subfloat[\textit{LiveJournal-WC}]
	{\includegraphics[width=0.5\linewidth]{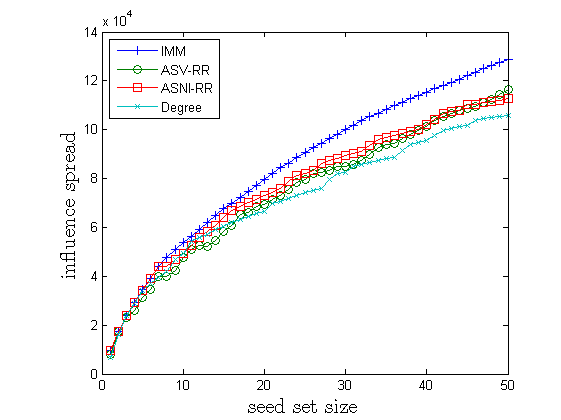}\label{fig:ljwcinfmax}}
	\caption{Influence maximization test.}
	\label{fig:infmax}
\end{figure}
}

\OnlyInFull{
\begin{figure*}[t]
	\centering
	\captionsetup[subfigure]{font=scriptsize,oneside,margin={0.6cm,0.0cm}}
	\setcounter{subfigure}{0}%
	\subfloat[\textit{Flixster-LN}]
	{\includegraphics[width=0.25\linewidth]{./fig/flixster_t01_infmax.png}\label{fig:fxinfmax}}
	\subfloat[\textit{DBLP-WC}]
	{\includegraphics[width=0.25\linewidth]{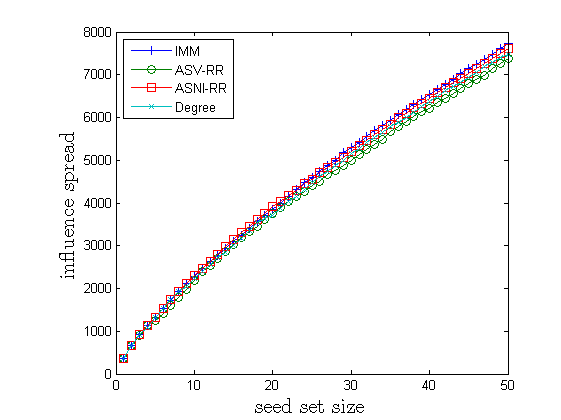}\label{fig:dblpwcinfmax}}
	\subfloat[\textit{DBLP-PR}]
	{\includegraphics[width=0.25\linewidth]{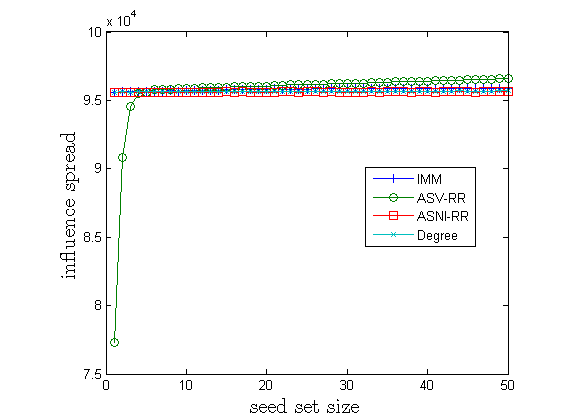}\label{fig:dblpprinfmax}}
	\subfloat[\textit{LiveJournal-WC}]
	{\includegraphics[width=0.25\linewidth]{./fig/livejournal-wc_infmax.png}\label{fig:ljwcinfmax}}
	\caption{Influence maximization test on {\IMM}, {\ASVRR}, and {\Degree}.}
	\label{fig:infmax}
\end{figure*}
}

\OnlyInFull{
We run {\ASVRR}, {\ASNIRR}, {\IMM}, and {\Degree}
   on four influence instances: 
  (1) the Flixster network with learned probability,
  (2) the DBLP network with WC parameters, 
  (3) the DBLP network with PR parameters, and
  (4)  the LiveJournal network with WC	parameters.
Figure~\ref{fig:infmax} shows the results of these 
  four tests whose objectives are to identify 50  influential seed nodes.
The influence spread in each case 
  is obtained by running 10K  Monte Carlo simulations 
  and taking the average value.
The results on all datasets in general show that
	both Shapley and SNI centrality performs reasonably well for the influence maximization task, but in some
	cases {\IMM} is still noticeably better.
This is because {\IMM} is specially designed for the influence maximization task while Shapley and SNI are
	two centrality measures related to influence but not specialized for the influence maximization task.
For the FX-LN dataset, Shapley top rankers performs noticeably better than SNI top rankers (average
	$8.3\%$ improvement).
This is perhaps due to that Shapley centrality accounts for more marginal influence, which is closer to what
	is needed for influence maximization.
This is also the test where they both significantly outperform the baseline {\Degree} heuristic, 
	again indicating that
	influence learned from the real-world data may contain significantly more information than the graph
	structure, in which case degree centrality is not a good index for node importance.

The behavior of DBLP-PR needs a bit more attention. 
For {\ASNIRR} (as well as {\IMM} and {\Degree}), the first seed selected already generates influence spread of
	$95K$, but subsequent seeds only have very small marginal contribution to the influence spread.
On the contrary, the first seed selected by {\ASVRR} only has influence spread of $77K$, and the spread reaches
	the level of {\ASNIRR} at the fourth seed.
Looking more closely, the first seed selected by {\ASVRR} has Shapley centrality of $10.3$ but its influence spread
	of $77K$ is only ranked at around $68K$ on SNI ranking, while
	the first seed of {\ASNIRR} has Shapley centrality of $3.15$,
	with Shapley ranking beyond $2100$.
This shows that when a large portion of nodes have high individual but overlapping influence 
	(due to the emergence of
	the giant component in live-edge graphs), they all become more or less replaceable, 
	and thus Shapley ranking, which focuses on marginal influence in a random order, 
	would differs from SNI ranking significantly.

Finally, we evaluate the scalability of {\ASVRR} and {\ASNIRR}, and use
	{\IMM} as a reference point, even though {\IMM} is designed for a different task.
We use the same setting of $\varepsilon=0.5$, $\ell=1$, and $k=50$.
Table~\ref{tab:time} reports the running time of 
  the three algorithms on four large influence instances.
For FX-LN, DB-WC, and LJ-WC,
	the general trend is that {\IMM} is the fastest, followed by {\ASNIRR}, and then {\ASVRR}.
This is expected, because the theoretical running time of {\IMM} is
	$\Theta((k+\ell) (m+n) \log n \cdot \E[\sigma(\tilde{\bv})]/ (OPT_k \cdot \varepsilon^2))$,
	where $OPT_k$ is maximum influence spread with $k$ seeds.
Thus comparing to the running time results in Theorems~\ref{thm:ASVRR} and~\ref{thm:SNI}, 
	typically $OPT_k$ is much larger than the $k$-th largest SNI centrality, which in turn is much larger than
	the $k$-th largest Shapley centrality, which leads to the observed running time result.
Nevertheless, both {\ASNIRR} and {\ASVRR} could be considered efficient in these cases and they can scale to
	large graphs with tens of millions of nodes and edges.

DB-PR again is an out-lier, with {\ASNIRR} faster than {\IMM}, and {\ASVRR} being too slow and inefficient.
This is because a large portion of nodes have large individual but overlapping influence, 
	so that $OPT_{50} = 95.9K$ is almost
	the same as the $50$-th largest SNI value ($94.2K$), in which case the $(k+\ell)$ factor in the running
	time of {\IMM} dominates and makes {\IMM} slower than {\ASNIRR}.
As for {\ASVRR}, due to the severe overlapping influence, the $50$-th largest Shapley value ($5.10$)
	is much smaller than the $50$-th largest SNI value or $OPT_{50}$, resulting in much slower running time
	for {\ASVRR}.
}

\OnlyInFull{
\begin{table}[t]
	\centering \caption{\label{tab:time} Running time (in seconds). 
		}
	\begin{tabular}{|l|r|r|r|r|}
		\hline 
		Algorithm & FX-LN & DB-WC & DB-PR & LJ-WC \\
		\hline
		\hline
		{\ASVRR} 		& 24.83 & 838.27 	& 594752 & 8295.57\\
		{\ASNIRR} 		& 1.36  & 61.41 	& 28.42  & 267.50\\
		{\IMM} 			& 0.62 	& 18.08 	& 336.63 & 54.88 \\
		\hline
	\end{tabular}
\end{table}
}

\OnlyInShort{
\begin{table}[t]
	\centering \caption{\label{tab:time} Running time (in seconds). 
	}
	\begin{tabular}{|l|r|r|r|}
		\hline 
					 & {\ASVRR}   &  {\ASNIRR} & {\IMM} \\
		\hline
		\hline
		FX-LN  		&  	24.83	& 	1.36	&  0.62 \\
		LJ-WC  		&   8295.57	& 	267.50	& 54.88 \\
		\hline
	\end{tabular}
\end{table}
	}

\OnlyInFull{
In summary, 
  our experimental results on small and large datasets 
  demonstrate that 
  (a) Shapley and SNI centrality behaves similarly in these networks, but with noticeable differences;
  (b) for the influence maximization task, they perform close to the specially designed IM algorithm, with 
	  Shapley centrality noticeably better than SNI in some case; and
  (c) both can scale to large graphs with tens of millions of edges, with {\ASNIRR} having better scalability.
	  except that {\ASVRR} would not be efficient for graphs with a huge gap between individual
	  influence and marginal influence.
Finally, we remark that {\ASVRR} and {\ASNIRR} do not need to store RR sets, 
	which eliminates a memory bottleneck that could be encountered by {\IMM} on large datasets.
}

\OnlyInShort{\vspace{-2mm}}
\section{Conclusion and Future Work}

Through an integrated mathematical, algorithmic, and empirical study
   of Shapley and SNI centralities
	in the context of network influence, we have shown that
	(a) both enjoy concise
	axiomatic characterizations, which precisely capture their
	similarity and differences;
	(b) both centrality measures can be efficiently 
  approximated with guarantees under the same algorithmic 
		structure, for a large class of influence models; and
	(c) Shapley centrality focuses 
  on nodes' marginal influence and their irreplaceability in group 
		influence settings, 
		while SNI centrality focuses on individual influence in isolation, and is not suitable in assessing
		nodes' ability in group influence setting, such as threshold-based models.

There are several directions 
  to extend this work and further explore 
  the interplay between social influence
	and network centrality.
One important direction is to formulate
   centrality measures that combine the advantages of Shapley and SNI 
	centralities, 
  by viewing Shapley and SNI centralities 
  as two extremes in a centrality spectrum, one
	focusing on individual influence 
   while the other focusing on marginal influence in groups of all sizes.
Then, would there be some intermediate centrality measure that provides a better balance?
Another direction is to incorporate other classical centralities into influence-based
	centralities.
For example, SNI centrality may be viewed as a generalized version of degree centrality, because
	when we restrict the influence model to deterministic activation
  of  only immediate neighbors, 
	SNI centrality essentially becomes degree centrality.
What about the general forms of closeness, betweenness, PageRank in the influence model?
Algorithmically, efficient algorithms for other influence models such as general threshold models~\cite{kempe03}
	is also interesting.
In summary, this paper lays a foundation for the further development 
of the axiomatic and algorithmic theory for influence-based network centralities, 
which we hope will provide us with deeper insights into network
structures and  influence dynamics.

\section*{Acknowledgment}

We thank Tian Lin for sharing his IMM implementation code and helping on
	data preparation.
We thank David Kempe and four anonymous reviewers for their valuable feedback.
Wei Chen is partially supported by 
  the National Natural Science Foundation of China (Grant No. 61433014).
Shang-Hua Teng is supported in part by 
  a Simons Investigator Award from the 
  Simons Foundation and by NSF grant CCF-1111270.

\OnlyInShort{\clearpage}



{\raggedright

\bibliographystyle{abbrv}
\bibliography{singlebib}

\begin{thebibliography}{10}

\bibitem{SNAP}
Stanford network analysis project.
\newblock https://snap.stanford.edu/data/.

\bibitem{PageRankAxioms}
A.~Altman and M.~Tennenholtz.
\newblock Ranking systems: The pagerank axioms.
\newblock In {\em ACM}, EC '05, pages 1--8, 2005.

\bibitem{ArrowBook}
K.~J. Arrow.
\newblock {\em Social Choice and Individual Values}.
\newblock Wiley, New York, 2nd edition, 1963.

\bibitem{barbieri2012topic}
N.~Barbieri, F.~Bonchi, and G.~Manco.
\newblock Topic-aware social influence propagation models.
\newblock In {\em ICDM}, 2012.

\bibitem{BV14}
P.~Boldi and S.~Vigna.
\newblock Axioms for centrality.
\newblock {\em Internet Mathematics}, 10:222--262, 2014.

\bibitem{Bonacich1972}
P.~Bonacich.
\newblock Factoring and weighting approaches to status scores and clique
  identification,.
\newblock {\em Journal of Mathematical Sociology}, 2:113--120, 1972.

\bibitem{Bonacich1987power}
P.~Bonacich.
\newblock Power and centrality: A family of measures.
\newblock {\em American Journal of Sociology}, 92(5):1170--1182, 1987.

\bibitem{BorgattiCentrality}
S.~P. Borgatti.
\newblock Centrality and network flow.
\newblock {\em Social Networks}, 27(1):55--71, 2005.

\bibitem{borgatti06}
S.~P. Borgatti.
\newblock Identifying sets of key players in a social network.
\newblock {\em Computational and Mathematical Organizational Theory},
  12:21--34, 2006.

\bibitem{BorgsBrautbarChayesLucier}
C.~Borgs, M.~Brautbar, J.~Chayes, and B.~Lucier.
\newblock Maximizing social influence in nearly optimal time.
\newblock In {\em ACM-SIAM}, SODA '14, pages 946--957, 2014.

\bibitem{PageRank}
S.~Brin and L.~Page.
\newblock The anatomy of a large-scale hypertextual web search engine.
\newblock {\em Computer Networks}, 30(1-7):107--117, 1998.

\bibitem{ChenYuanZhang}
W.~Chen, Y.~Yuan, and L.~Zhang.
\newblock Scalable influence maximization in social networks under the linear
  threshold model.
\newblock In {\em IEEE}, ICDM '10, pages 88--97, 2010.

\bibitem{CL06}
F.~R.~K. Chung and L.~Lu.
\newblock Concentration inequalities and martingale inequalities: {A} survey.
\newblock {\em Internet Mathematics}, 3(1):79--127, 2006.

\bibitem{DomingosRichardson}
P.~Domingos and M.~Richardson.
\newblock Mining the network value of customers.
\newblock In {\em ACM}, KDD '01, pages 57--66, 2001.

\bibitem{GL14}
R.~Ghosh and K.~Lerman.
\newblock Rethinking centrality: The role of dynamical processes in social
  network analysis.
\newblock {\em Discrete and Continuous Dynamical Systems Series B}, pages
  1355--1372, 2014.

\bibitem{GhoshInterplay}
R.~Ghosh, S.-H. Teng, K.~Lerman, and X.~Yan.
\newblock The interplay between dynamics and networks: centrality, communities,
  and cheeger inequality.
\newblock In {\em ACM}, KDD '14, pages 1406--1415, 2014.

\bibitem{Katz}
L.~Katz.
\newblock A new status index derived from sociometric analysis.
\newblock {\em Psychometrika}, 18(1):39--43, March 1953.

\bibitem{kempe03}
D.~Kempe, J.~M. Kleinberg, and {\'E}.~Tardos.
\newblock Maximizing the spread of influence through a social network.
\newblock In {\em KDD}, pages 137--146, 2003.

\bibitem{ShapleyValueForCentrality1}
T.~P. Michalak, K.~V. Aadithya, P.~L. Szczepanski, B.~Ravindran, and N.~R.
  Jennings.
\newblock Efficient computation of the shapley value for game-theoretic network
  centrality.
\newblock {\em J. Artif. Int. Res.}, 46(1):607--650, Jan. 2013.

\bibitem{MU05}
M.~Mitzenmacher and E.~Upfal.
\newblock {\em Probability and Computing}.
\newblock Cambridge University Press, 2005.

\bibitem{MyersonBook}
R.~B. Myerson.
\newblock {\em Game Theory : {A}nalysis of Conflict}.
\newblock Harvard University Press, 1997.

\bibitem{NashBargining}
J.~Nash.
\newblock The bargaining problem.
\newblock {\em Econometrica}, 18(2):155--162, April 1950.

\bibitem{NewmanBook}
M.~Newman.
\newblock {\em Networks: An Introduction}.
\newblock Oxford University Press, Inc., New York, NY, USA, 2010.

\bibitem{Nieminen73}
U.~Nieminen.
\newblock On the centrality in a directed graph.
\newblock {\em Social Science Research}, 2(4):371--378, 1973.

\bibitem{PageRankBMW98}
L.~Page, S.~Brin, R.~Motwani, and T.~Winograd.
\newblock The pagerank citation ranking: Bringing order to the web.
\newblock In {\em Proceedings of the 7th International World Wide Web
  Conference}, pages 161--172, 1998.

\bibitem{IntellectualInfluence}
I.~Palacios-Huerta and O.~Volij.
\newblock The measurement of intellectual influence.
\newblock {\em Econometrica}, 72:963--977, 2004.

\bibitem{PercolationCentrality}
M.~Piraveenan, M.~Prokopenko, and L.~Hossain.
\newblock Percolation centrality: Quantifying graph-theoretic impact of nodes
  during percolation in networks.
\newblock {\em PLoS ONE}, 8(1), 2013.

\bibitem{RichardsonDomingos}
M.~Richardson and P.~Domingos.
\newblock Mining knowledge-sharing sites for viral marketing.
\newblock In {\em ACM}, KDD '02, pages 61--70, 2002.

\bibitem{Sabidussi66}
G.~Sabidussi.
\newblock The centrality index of a graph.
\newblock {\em Psychometrika}, 31(4):581--603, 1966.

\bibitem{SB16}
D.~Schoch and U.~Brandes.
\newblock Re-conceptualizing centrality in social networks.
\newblock {\em European Journal of Applied Mathematics}, 27:971--985, 2016.

\bibitem{Shapley53}
L.~S. Shapley.
\newblock A value for $n$-person games.
\newblock In H.~Kuhn and A.~Tucker, editors, {\em Contributions to the Theory
  of Games, Volume II}, pages 307--317. Princeton University Press, 1953.

\bibitem{ChiKDD09}
J.~Tang, J.~Sun, C.~Wang, and Z.~Yang.
\newblock Social influence analysis in large-scale networks.
\newblock In {\em KDD}, 2009.

\bibitem{tang15}
Y.~Tang, Y.~Shi, and X.~Xiao.
\newblock Influence maximization in near-linear time: a martingale approach.
\newblock In {\em SIGMOD}, pages 1539--1554, 2015.

\bibitem{tang14}
Y.~Tang, X.~Xiao, and Y.~Shi.
\newblock Influence maximization: near-optimal time complexity meets practical
  efficiency.
\newblock In {\em SIGMOD}, pages 75--86, 2014.

\bibitem{TengScalable}
S.-H. Teng.
\newblock Scalable algorithms for data and network analysis.
\newblock {\em Foundations and Trends in Theoretical Computer Science},
  12(1-2):1--261, 2016.

\bibitem{NetworkEssence}
S.-H. Teng.
\newblock Network essence: Pagerank completion and centrality-conforming markov
  chains.
\newblock In J.~N. Martin~Loebl and R.~Thomas, editors, {\em A Journey through
  Discrete Mathematics. A Tribute to Ji\v{r}\'{i} Matou\v{s}ek}. Springer
  Berlin / Heidelberg, 2017.

\bibitem{wang2012scalable}
C.~Wang, W.~Chen, and Y.~Wang.
\newblock Scalable influence maximization for independent cascade model in
  large-scale social networks.
\newblock {\em DMKD}, 25(3):545--576, 2012.

\end{thebibliography}
}

\OnlyInFull{
\clearpage
\appendix



\section{Proofs on Axiomatic Characterization}

\subsection{Proof of Theorem~\ref{thm:ShapleyCen}} \label{app:thm2}

We use $\cA$ to denote the set of Axioms \ref{axiom:anonymity}-\ref{axiom:critical}.                             

\subsubsection*{{\sc Analysis of Sink Nodes}}
We first prove that the involvement of sink nodes in the influence
  process is what we have expected:
(1) The marginal contribution of a sink node $v$ 
     is equal to the probability that $v$ is not influenced by the seed set.
(2) For any other node $u\in V$, 
  $u$'s activation probability is the same
  whether or not $v$ is in the seed set.

\begin{lemma} \label{lem:sinkmargin}
Suppose $v$ is a sink node in $\calI = (V,E,P_{\calI})$. Then,
(a) for any $S\subseteq V\setminus \{v\}$:
	$$\sigma_\calI(S\cup \{v\}) - \sigma_\calI(S) = \Pr(v\not\in \bI_\calI(S)).$$
(b) for any $u\ne v$ and any $S\subseteq V\setminus \{u,v\}$:
	$$\Pr(u\not\in \bI_\calI(S\cup\{v\})) = \Pr(u\not\in \bI_\calI(S)).$$
\end{lemma}
\begin{proof}
For (a), by the definitions of $\sigma_\calI$ and sink nodes:
	\begin{align*}
	& \sigma_\calI(S\cup \{v\}) \\
	& = \sum_{T\supseteq S\cup \{v\}} P_{\calI}(S\cup\{v\}, T) \cdot |T| \\
	& =  \sum_{T\supseteq S\cup \{v\}} \left(P_{\calI}(S, T\setminus \{v\}) + P_{\calI}(S, T)\right) \cdot |T|  \\
	& = \sum_{T\supseteq S, T\subseteq V\setminus \{v\}} P_{\calI}(S, T) (|T|+1) + 
	\sum_{T\supseteq S\cup \{v\}} P_{\calI}(S, T) \cdot |T| \\
	& = \sum_{T\supseteq S} P_{\calI}(S, T) \cdot |T| + \sum_{T\supseteq S, T\subseteq V\setminus \{v\}} P_{\calI}(S, T) \\
	& = \sigma_\calI(S)+\Pr(v\not\in \bI_\calI(S)).
	\end{align*}
	For (b),
	\begin{align*}
	&\Pr(u\not\in \bI_\calI(S\cup\{v\}))  \\
	& =   \sum_{T\supseteq S\cup\{v\}, T\subseteq V\setminus \{u\}} P_{\calI}(S\cup \{v\}, T) \\
	& =  \sum_{T\supseteq S\cup\{v\}, T\subseteq V\setminus \{u\}} \left(P_{\calI}(S, T\setminus \{v\}) + P_{\calI}(S, T)  \right)  \\
	& =  \sum_{T\supseteq S, T\subseteq V\setminus \{u\}} P_{\calI}(S, T) = \Pr(u\not\in \bI_\calI(S)).
	\end{align*}
\end{proof}

Lemma \ref{lem:sinkmargin} immediately
  implies that for any two sink nodes $u$ and $v$, 
$u$'s marginal contribution to any $S\subseteq V\setminus \{u,v\}$
  is the same as its marginal contribution to $S\cup \{v\}$:
\begin{lemma}[Independence between Sink Nodes] \label{lem:sinkind}
	If $u$ and $v$ are two sink nodes in $\calI$,  
then for any $S\subseteq V\setminus \{u,v\}$, 
	$\sigma_\calI(S\cup \{v,u\}) - \sigma_\calI(S\cup \{v\}) = \sigma_\calI(S\cup \{u\}) - \sigma_\calI(S)$.
\end{lemma}
\begin{proof}
By Lemma~\ref{lem:sinkmargin} (a) and (b), both sides 
 are  equal to $\Pr(u\not\in \bI_\calI(S))$.
\end{proof}

The next two lemmas  connect the influence spreads
  in the original and projected instances. 
\begin{lemma} \label{lem:projectsigma}
If $v$ is a sink in $\calI$, then for any $S\subseteq V\setminus \{v\}$:
$$\sigma_{\calI\setminus \{v\}}(S) = \sigma_\calI(S)- \Pr(v\in \bI_\calI(S)).$$
\end{lemma}
\begin{proof}
By the definition of influence projection: 
	\begin{align*}
	& \sigma_{\calI\setminus \{v\}}(S) \\
	& = \sum_{T \supseteq S, T\subseteq V\setminus \{v\}} P_{\calI\setminus \{v\}}(S, T) \cdot |T| \\
	& = \sum_{T \supseteq S, T\subseteq V\setminus \{v\}} \left( P_{\calI}(S, T) + P_{\calI}(S, T\cup \{v\})\right) \cdot |T| \\
	& = \sum_{T \supseteq S, T\subseteq V\setminus \{v\}}  P_{\calI}(S, T) \cdot |T| +
	\sum_{T \supseteq S\cup \{v\}}  P_{\calI}(S, T) \cdot (|T| -1) \\
	& = \sum_{T \supseteq S} P_{\calI}(S, T) \cdot |T| - \sum_{T \supseteq S\cup \{v\}}  P_{\calI}(S, T) \\
	& = \sigma_\calI(S)- \Pr(v\in \bI_\calI(S)).
	\end{align*}
\end{proof}

\begin{lemma} \label{lem:projectmargin}
For any two sink nodes $u$ and $v$ in $\calI$:
$$\sigma_{\calI\setminus \{v\}}(S \cup \{u\}) - \sigma_{\calI\setminus \{v\}}(S) 
	= \sigma_{\calI}(S \cup \{u\}) - \sigma_{\calI}(S).$$
\end{lemma}
\begin{proof}
	By Lemmas~\ref{lem:projectsigma} and~\ref{lem:sinkmargin} (b), we have
	\begin{align*}
	&\sigma_{\calI\setminus \{v\}}(S \cup \{u\}) - \sigma_{\calI\setminus \{v\}}(S) \\
	& = \sigma_\calI(S\cup \{u\})- \Pr(v\in \bI_\calI(S\cup \{u\})) \\
	& \quad 	\quad 	- \left( \sigma_\calI(S)- \Pr(v\in \bI_\calI(S))  \right) \\
	& = \sigma_\calI(S\cup \{u\})- \Pr(v\in \bI_\calI(S)) 
	- \left( \sigma_\calI(S)- \Pr(v\in \bI_\calI(S))  \right) \\
	& = \sigma_{\calI}(S \cup \{u\}) - \sigma_{\calI}(S).
	\end{align*}
\end{proof}

\subsubsection*{{\sc Soundness}}

\begin{lemma} \label{lem:shapleycen}
The Shapley centrality satisfies all	Axioms \ref{axiom:anonymity}-\ref{axiom:critical}.
\end{lemma}
\begin{proof}
	Axioms~\ref{axiom:anonymity}, \ref{axiom:normalization}, and
	\ref{axiom:bayesian} are trivially satisfied by $\psi^{\Shapley}$, or are direct implications
	from the original Shapley axiom set.
	
	Next, we show that $\psi^{\Shapley}$ satisfies Axiom~\ref{axiom:sink}, the Axiom of 
	Independence of Sink Nodes.
	Let $u$ and $v$ be two sink nodes.
	Let $\bpi$ be a random permutation on $V$.
	Let $\bpi'$ be the random permutation on $V\setminus \{v\}$ derived from $\bpi$ by removing $v$ from the random order.
	Let $\{u \prec_{\bpi} v\}$ be the event that $u$ is ordered before $v$ in the permutation $\bpi$.
	Then we have
	\begin{align}
	& \psi^{\Shapley}_u(\calI) =  \E_{\bpi}[\sigma_\calI(S_{\bpi,u} \cup \{u\}) - \sigma_\calI(S_{\bpi,u})] \nonumber \\
	& = \Pr(u \prec_{\bpi} v)\E_{\bpi}[\sigma_\calI(S_{\bpi,u} \cup \{u\}) - \sigma_\calI(S_{\bpi,u}) \mid u \prec_{\bpi} v] + \nonumber \\
	& \quad \Pr(v \prec_{\bpi} u)\E_{\bpi}[\sigma_\calI(S_{\bpi,u} \cup \{u\}) - \sigma_\calI(S_{\bpi,u}) \mid v \prec_{\bpi} u] \nonumber \\
	& = \Pr(u \prec_{\bpi} v)\E_{\bpi'}[\sigma_\calI(S_{\bpi',u} \cup \{u\}) - \sigma_\calI(S_{\bpi',u})] + \nonumber \\
	& \quad \Pr(v \prec_{\bpi} u)\E_{\bpi}[\sigma_\calI(S_{\bpi,u} \cup \{u\}) - \sigma_\calI(S_{\bpi,u}) \mid v \prec_{\bpi} u] \nonumber \\
	& = \Pr(u \prec_{\bpi} v)\E_{\bpi'}[\sigma_\calI(S_{\bpi',u} \cup \{u\}) - \sigma_\calI(S_{\bpi',u})] + \nonumber \\
	& \quad \Pr(v \prec_{\bpi} u) \cdot \nonumber \\
	& \quad \E_{\bpi}[\sigma_\calI(S_{\bpi,u} \setminus \{v\} \cup \{u\}) - \sigma_\calI(S_{\bpi,u}\setminus \{v\} ) \mid v \prec_{\bpi} u]
	\label{eq:usesinkind} \\
	& = \Pr(u \prec_{\bpi} v)\E_{\bpi'}[\sigma_\calI(S_{\bpi',u} \cup \{u\}) - \sigma_\calI(S_{\bpi',u})] + \nonumber \\
	& \quad \Pr(v \prec_{\bpi} u) \E_{\bpi'}[\sigma_\calI(S_{\bpi',u} \cup \{u\}) - \sigma_\calI(S_{\bpi',u})] \nonumber \\
	& = \E_{\bpi'}[\sigma_\calI(S_{\bpi',u} \cup \{u\}) - \sigma_\calI(S_{\bpi',u})] \nonumber \\
	& = \E_{\bpi'}[\sigma_{\calI\setminus \{v\}}(S_{\bpi',u} \cup \{u\}) - \sigma_{\calI\setminus \{v\}}(S_{\bpi',u})]
	\label{eq:useprojectmar} \\
	& = \psi^{\Shapley}_u({\calI\setminus \{v\}}). \nonumber
	\end{align}
	Eq.\eqref{eq:usesinkind} above uses Lemma~\ref{lem:sinkind}, while Eq.\eqref{eq:useprojectmar}
	uses Lemma~\ref{lem:projectmargin}.
	
	Finally, we show that $\psi^{\Shapley}$ satisfies Axiom~\ref{axiom:critical}, the Critical Set Axiom.
	By the definition of the critical set instance, 
	we know that if influence instance $\calI$ has
	critical set $R$, then $\sigma_\calI(S) = |V|$ if $S \supseteq R$, and
	$\sigma_\calI(S) = |S|$ if $S \not\supseteq R$.
	Then for $v \not\in R$, for any $S\subseteq V\setminus \{v\}$,  
	$\sigma_\calI(S \cup \{v\}) - \sigma_\calI(S) = 0$ if $S \supseteq R$, and
	$\sigma_\calI(S \cup \{v\}) - \sigma_\calI(S) = 1$ if $S \not\supseteq R$.
	For a random permutation $\bpi$, the event $R\subseteq S_{\bpi,v}$
	is the event that all nodes in $R$ are ordered before $v$
	in $\bpi$, which has probability $1/(|R|+1)$.
	Then we have that for $v \not\in R$, 
	\begin{align*}
	& \psi^{\Shapley}_v(\calI) = \E_{\bpi}[\sigma_\calI(S_{\bpi,v} \cup \{v\}) - \sigma_\calI(S_{\bpi,v})] \\
	& =  \Pr(R\subseteq S_{\bpi,v})\E_{\bpi}[\sigma_\calI(S_{\bpi,v} \cup \{v\}) - 
	\sigma_\calI(S_{\bpi,v}) \mid R\subseteq S_{\bpi,v}] + \\
	& \quad \Pr(R \not\subseteq S_{\bpi,v})\E_{\bpi}[\sigma_\calI(S_{\bpi,v} \cup \{v\}) - 
	\sigma_\calI(S_{\bpi,v}) \mid R \not\subseteq S_{\bpi,v}] \\
	& =   \Pr(R \not\subseteq S_{\bpi,v}) = \frac{|R|}{|R|+1}.
	\end{align*}
	Therefore, Shapley centrality $\psi^{\Shapley}$ is a solution consistent with
        Axioms \ref{axiom:anonymity}-\ref{axiom:critical}.
\end{proof}

\subsubsection*{{\sc Completeness (or Uniqueness)}}

We now prove the uniqueness of axiom set $\cA$.
Fix a set $V$. 
For any $R, U \subseteq V$ with $R\ne \emptyset$ and $R\subseteq U$,
we define the critical set instance $\calI_{R,U}$, an extension to
	the critical set instance $\calI_{R,v}$ defined for
	Axiom~\ref{axiom:critical}.

\begin{definition}[General Critical Set Instances] \label{def:critical}
For any $R, U \subseteq V$ with $R\ne \emptyset$ and $R\subseteq U$,
the critical set instance $\calI_{R,U} = (V, E, P_{\calI_{R,U}})$
is the following influence instance: 
(1) The network $G = (V,E)$ contains a complete directed
    bipartite sub-graph from $R$ to $U\setminus R$, together with 
    isolated nodes $V \setminus U$.
(2) For all $S \supseteq R$, $P_{\calI_{R,U}}(S, U \cup S) = 1$, and
(3) For all $S \not \supseteq R$, $P_{\calI_{R,U}}(S, S) = 1$.
For this instance, $R$ is called the {\em critical set}, and
	$U$ is called the {\em target set}.
\end{definition}

Intuitively, in the critical set instance $\calI_{R,U}$, once
	the seed set contains the critical set $R$, it guarantees to activate
	target set $U$ together with other nodes in $S$; but as long as some
	nodes in $R$ is not included in the seed set $S$, only nodes in $S$
	can be activated.
These critical set instances play an important role in the uniqueness proof.
Thus, we first study their properties.

To study the properties of the critical set instances, it is helpful for us to introduce
	a special type of sink nodes called {\em isolated nodes}.
\remove{\begin{definition}[Isolated Nodes] \label{def:isolated}
}We say $v\in V$ is an {\em isolated node} in 
  \remove{influence instance }$\calI = (V,E,P_{\calI})$, if
  $\forall S, T \subseteq V\setminus \{v\}$ 
  with $S \subseteq T$, $P_{\calI}(S\cup\{v\}, T\cup \{v\}) = P_{\calI}(S,T)$.
In the extreme case, 
  $P_{\calI}(\{v\}, \{v\}) = P_{\calI}(\emptyset, \emptyset) = 1$, 
  meaning that $v$ only activates itself,
No seed set can influence $v$ unless it contains $v$:
For any $S, T \subseteq V\setminus \{v\}$ with $S \subseteq T$, 
	$P_{\calI}(S, T\cup \{v\}) \le 1 - \sum_{T'\supseteq S, T' \subseteq V\setminus \{v\}} P_{\calI}(S, T')
		= 1 - \sum_{T'\supseteq S, T' \subseteq V\setminus \{v\}} P_{\calI}(S\cup\{v\}, T' \cup \{v\})=0$.
\remove{\end{definition}In $\calI$, 
Intuitively, $v\in V$ is an isolated node if it can neither influence others nor be influenced by others:
In other words,}The role of $v$ in any seed set 
  is just to activate itself: The probability of activating other nodes
  is unchanged if $v$ is removed from the seed set.
It is easy to see that by definition an isolated node is a sink node.

\begin{lemma}[Sinks and Isolated Nodes] \label{lem:isosink}
In the critical set instance $\calI_{R,U}$, 
every node in $V\setminus U$ is an isolated node,
and every node in $V \setminus R$ is a sink node.
\end{lemma}
\begin{proof}
We first prove that every node $v \in V\setminus U$ is an isolated node.
Consider any two subsets $S,T\subseteq V\setminus \{v\}$ with $S\subseteq T$.
We first analyze the case when $S \supseteq R$.
By Definition~\ref{def:critical}, 
  $P_{\calI}(S\cup \{v\}, T\cup \{v\}) = 1$ 
  iff  $T \cup \{v\} = U \cup S \cup\{v\}$, which
	is equivalent to $T = U\cup S$ since $v \not \in U $.
	This implies that $P_{\calI}(S\cup \{v\}, T\cup \{v\}) = P_{\calI}(S, T) $.
We now analyze the case when $S\not\supseteq R$.
By Definition~\ref{def:critical},
	$P_{\calI}(S\cup \{v\}, T\cup \{v\}) = 1$ iff  $T \cup \{v\} = S \cup\{v\}$, which is
	equivalent to $T = S$.
This again implies that $P_{\calI}(S\cup \{v\}, T\cup \{v\}) = P_{\calI}(S, T) $.
Therefore, 
  $v$ is an isolated node.
	
Next we show that every node $v\not\in R$ is a sink node.
Consider any two subsets $S,T\subseteq T\setminus \{v\}$ with $S\subseteq T$.
In the case when $S \supseteq R$,
 $P_{\calI}(S\cup \{v\}, T\cup \{v\}) = 1$ iff $T \cup \{v\} = U \cup S \cup \{v\}$, 
 which is equivalent to $T = U \cup S \setminus \{v\}$.
Depending on whether $v \in U$,	$T = U \cup S \setminus \{v\}$ is equivalent to 
	exactly one of  $T = U \cup S$ or $T \cup \{v\} = U\cup S$ being true.
This implies that
	$P_{\calI}(S\cup \{v\}, T\cup \{v\}) = P_{\calI}(S, T) + P_{\calI}(S, T\cup\{v\})$.
In the case when $S \not \supseteq R$,
	$P_{\calI}(S\cup \{v\}, T\cup \{v\}) = 1$ iff
	$T\cup \{v\} = S \cup \{v\}$, which is
	equivalent to $T = S$.
This also implies that 
	$P_{\calI}(S\cup \{v\}, T\cup \{v\} = P_{\calI}(S, T) + P_{\calI}(S, T\cup\{v\})$.
Therefore, $v$ is a sink node by definition.	
\end{proof}

\begin{lemma}[Projection] \label{lem:projcritical}
	In the critical set instance $\calI_{R,U}$, 
	for any node $v\in V\setminus U$, the projected influence instance
	of $\calI_{R,U}$ on $V\setminus \{v\}$, $\calI_{R,U}\setminus \{v\}$, is
	a critical set instance with critical set $R$ and target $U$, in the
	projected graph $G\setminus \{v\} = (V\setminus \{v\}, E\setminus \{v\})$.
	For any node $v\in U\setminus R$, 	the projected influence instance
	of $\calI_{R,U}$ on $V\setminus \{v\}$, $\calI_{R,U}\setminus \{v\}$,  is
	a critical set instance with critical set $R$ and target $U\setminus \{v\}$, in the
	projected graph $G\setminus \{v\} = (V\setminus \{v\}, E\setminus \{v\})$.
\end{lemma}
\begin{proof}
	First let  $v\in V\setminus U$ and consider the projected instance 
	$\calI_{R,U}\setminus \{v\}$.
	If $S\subseteq V\setminus \{v\}$ is a subset with $S \supseteq R$, then
	by the definition 
         of projection and critical sets: 
	\begin{align*}
	& P_{\calI_{R,U} \setminus \{v\}}(S, S\cup U)  \\
	& = P_{\calI_{R,U}}(S, S\cup U) + 
	P_{\calI_{R,U}}(S, S\cup U\cup\{v\}) \\
	& = 1 + 0 = 1. 
	\end{align*}
	If $S\subseteq V\setminus \{v\}$ is a subset with $S \not\supseteq R$, similarly, we have:
	\begin{align*}
	& P_{\calI_{R,U} \setminus \{v\}}(S, S)  \\
	& = P_{\calI_{R,U}}(S, S) + 
	P_{\calI_{R,U}}(S, S\cup\{v\}) = 1 + 0 = 1. 
	\end{align*}
	Thus by Definition~\ref{def:critical},  $\calI_{R,U} \setminus \{v\}$
	is still a critical set instance with $R$ as the critical set and $U$ as
	the target set.	
	
	Next let $v\in U\setminus R$ and consider the projected instance 
	$\calI_{R,U}\setminus \{v\}$.
	If $S\subseteq V\setminus \{v\}$ is a subset with $S \supseteq R$, then
	by the definition 
	of projection and critical sets: 
	\begin{align*}
	& P_{\calI_{R,U} \setminus \{v\}}(S, S\cup (U \setminus \{v\}))  \\
	& = P_{\calI_{R,U}}(S, S\cup (U \setminus \{v\})) + 
	P_{\calI_{R,U}}(S, S\cup (U \setminus \{v\}) \cup\{v\}) \\
	& = 0 + 1 = 1. 
	\end{align*}
	If $S\subseteq V\setminus \{v\}$ is a subset with $S \not\supseteq R$, similarly, we have:
	\begin{align*}
	& P_{\calI_{R,U} \setminus \{v\}}(S, S)  \\
	& = P_{\calI_{R,U}}(S, S) + 
	P_{\calI_{R,U}}(S, S\cup\{v\}) = 1 + 0 = 1. 
	\end{align*}
	Thus by Definition~\ref{def:critical},  $\calI_{R,U} \setminus \{v\}$
	is still a critical set instance with $R$ as the critical set and $U \setminus \{v\}$ as
	the target set.	
\end{proof}

\begin{lemma}[Uniqueness in Critical Set Instances] \label{lem:uniquecritical}
Fix a set $V$.
Let $\psi$ be a centrality measure that satisfies axiom set $\cA$.
For any $R, U \subseteq V$ with $R\ne \emptyset$ and $R\subseteq U$,
 the centrality	$\psi(\calI_{R,U})$ of the critical set instance $\calI_{R,U}$ 
  must be unique.
\end{lemma}
\begin{proof}
	Consider the critical set instance $\calI_{R,U}$.
	First, it is easy to check that all nodes in $R$ are symmetric to one another, all nodes in
	$U\setminus R$ are symmetric to one another, and all nodes in $V\setminus U$ are symmetric
	to one another.
	Thus, by the Anonymity Axiom (Axiom~\ref{axiom:anonymity}), all nodes in $R$ have the same
	centrality measure, say $a_{R,U}$, all nodes in $U\setminus R$ have the same centrality measure, say $b_{R,U}$,
	and all nodes in $V\setminus U$ have the same centrality measure, say $c_{R,U}$.
	By the Normalization Axiom (Axiom~\ref{axiom:normalization}), we have
	\begin{equation} \label{eq:threeparameters}
	a_{R,U} \cdot |R| + b_{R,U} \cdot (|U|-|R|) + c_{R,U} \cdot (|V| - |U|) = |V|.
	\end{equation}	
	
	Second, we consider any node $v \in V\setminus U$.
	By Lemma~\ref{lem:isosink}, $v$ is an isolated node, which is also a sink node. 
	By Lemma~\ref{lem:projcritical}, we can iteratively remove all sink nodes in $U\setminus R$,
	which would not change the centrality measure of $v$ by the Independence of Sink Nodes Axiom
	(Axiom~\ref{axiom:sink}).
	Moreover, after removing all nodes in $U\setminus R$, the projected instance $\calI'$
	is on set $R\cup (V\setminus U)$, with $R$ as both the critical set and the target set.
	In this projected instance $\calI'$, it is straightforward to check that for every 
	$S\subseteq R\cup (V\setminus U)$, $P_{\calI'}(S,S)=1$, which implies that every node
	in $R\cup (V\setminus U)$ is an isolated node.
	Then we can apply the Anonymity Axiom to know that every node in $\calI'$ has the same centrality,
	and together with the Normalization Axiom, we know that every node in $\calI'$ has centrality $1$.
	Since by the Independence of Sink Nodes Axiom removing nodes in $U\setminus R$ does not change
	the centrality of nodes in $V\setminus U$, we know that $c_{R,U} = 1$.
	
	
	Third, if $U = R$, then we do not have parameter $b_{R,U}$ and $a_{R,U}$ is determined
	by Eq.~\eqref{eq:threeparameters}.
	If $ U \ne R$, then by Lemma~\ref{lem:isosink}, any node $v \in V\setminus R$ is a sink node.
	Then we can apply the Sink Node Axiom (Axiom~\ref{axiom:sink}) to iteratively
	remove all nodes in $V\setminus (R \cup \{v\})$ (which are all sink nodes),
	such that the centrality measure of $v$ does not change
	after the removal.
	By Lemma~\ref{lem:projcritical}, the remaining instance with node set $R\cup\{v\}$ is 
	still a critical set instance with critical set $R$ and target set $R\cup \{v\}$.
	Thus we can apply the Critical Set Axiom (Axiom~\ref{axiom:critical}) to this remaining
	influence instance, and know that the centrality measure
	of $v$ is $|R|/(|R|+1)$, that is, $b_{R,U} = |R|/(|R|+1)$.
	Therefore, $a_{R,U}$ is also uniquely determined, which means that the centrality measure
	$\psi(\calI_{R,U})$ for instance $\calI_{R,U}$ is unique, 
	for every nonempty subset $R$ and its superset $U$.
\end{proof}

The influence probability profile, $(P_\calI(S,T))_{S\subseteq T\subseteq V}$,
 of each social-influence instance $\calI$
  can be viewed as a high-dimensional vector. 
Note that in the boundary cases:
(1) when $S=\emptyset$, we have $P_\calI(S,T) = 1$ iff $T = \emptyset$; and
(2)  when $S = V$, $P_\calI(S,T) = 1$ iff $T  = V$.
Thus, {\em the influence-profile vector} does not need to include $S=\emptyset$ and $S=V$.
Moreover, for any $S$, $\sum_{T\supseteq S} P_\calI(S,T) = 1$.
Thus, we can omit the entry associated with one $T \supseteq S$ from influence-profile vector.
In our proof, we canonically remove the entry associated with $T = S$ from the vector.
With a bit of overloading on the notation, we also use
$P_\calI$ to denote this influence-profile vector for $\calI$, and thus
$P_\calI(S,T)$ is the value of the specific dimension of the vector corresponding to $S,T$.
We let $M$ denote the dimension of space
  of the influence-profile vectors.
$M$ is 
  equal to the number of pairs $(S,T)$  satisfying (1) $S \subset T\subseteq V$, and (2)
   $S \not\in \{\emptyset, V\}$.
$S\subset T$ means $S\subseteq T$ but $S \ne T$.
We stress that when we use $P_\calI$ as a vector and use linear combinations of such vectors,
the vectors have no dimension corresponding to $(S,T)$ with $S \in \{\emptyset, V\}$
or $S=T$.

For each $R$ and $U$ with $R \subset U$ and $R \not\in \{\emptyset, V\}$, 
we consider the critical set instance $\calI_{R,U}$ and its
corresponding vector $P_{\calI_{R,U}}$.
Let $\cV$ be the set of these vectors.
\begin{lemma}[Linear Independence] \label{lem:linearind}
	Vectors in $\cV$ are linearly independent in the space $\R^M$.
\end{lemma}
\begin{proof}
	
	Suppose, for a contradiction, that vectors in $\cV$ are not linearly independent.
	Then for each such $R$ and $U$, we have a number $\alpha_{R,U} \in \R$, such that
	$\sum_{R\not\in\{\emptyset,V\}, R\subset U} \alpha_{R,U} \cdot P_{\calI_{R,U}} = \vec{0}$,
	and at least some $\alpha_{R,U}\ne 0$.
	Let $S$ be the smallest set with $\alpha_{S,U} \ne 0$ for some $U \supset S$, and
	let $T$ be any superset of $S$ with $\alpha_{S,T} \ne 0$.
	By the critical set instance definition, we have $P_{\calI_{S,T}}(S,T) = 1$.
	Also since the vector does not contain any dimension corresponding to
	$P_\calI(S,S)$, we know that $T \supset S$.
	Then by the minimality of $S$, we have 
	\begin{align}
	& 0 = \sum_{R, U:R\not\in\{\emptyset,V\}, R\subset U} \alpha_{R,U} \cdot P_{\calI_{R,U}}(S,T) 
	\nonumber \\
	& = \alpha_{S,T} \cdot P_{\calI_{S,T}}(S,T) + 
	\sum_{U: U \supset S,  U \ne T} \alpha_{S,U} \cdot P_{\calI_{S,U}}(S,T) + \nonumber \\
	& \quad \quad \sum_{R, U:|R| \ge |S|, R\ne S, U \supset R} \alpha_{R,U} \cdot P_{\calI_{R,U}}(S,T) 
	\nonumber 	\\
	& = \alpha_{S,T}  + 
	\sum_{U:U \supset S,  U \ne T} \alpha_{S,U} \cdot P_{\calI_{S,U}}(S,T) + \nonumber \\
	& \quad \quad \sum_{R, U:|R| \ge |S|, R\ne S, U \supset R} \alpha_{R,U} \cdot P_{\calI_{R,U}}(S,T).
	\label{eq:linearind} 	
	\end{align}
	For the third term in Eq.\eqref{eq:linearind}, 
	consider any set $R$ with $|R| \ge |S|$ and $R\ne S$.
	We have that $S \not\supseteq R$, 
	and thus by the critical set instance definition, for any $U \supset R$, 
	$P_{\calI_{R,U}}(S,S) = 1$.
	Since $T \supset S$, we have $T \ne S$, and thus $P_{\calI_{R,U}}(S,T) = 0$.
	This means that the third term in Eq.\eqref{eq:linearind} is $0$.
	
	For the second term in Eq.\eqref{eq:linearind}, consider any 
	$U \supset S$ with $ U \ne T$.
	By the critical set instance definition, we have
	$P_{\calI_{S,U}}(S,U) = 1$ (since $S$ is the critical set and $U$ is the
	target set).
	Then $P_{\calI_{S,U}}(S,T) = 0$ since $T \ne U$.
	This means that the second term in Eq.\eqref{eq:linearind} is also $0$.
	
	Then we conclude that $\alpha_{S,T} = 0$, which is a contradiction.
	Therefore, vectors in $\cV$ are linearly independent.
\end{proof}


The following basic  lemma is useful for our uniqueness proof.
\begin{lemma} \label{lem:linearmap}
Let $\psi$ be a mapping from a convex set $D \subseteq \R^M$ to $\R^n$ satisfying
	that for any vectors $\vv_1, \vv_2, \ldots, \vv_s \in D$, 
	for any $\alpha_1, \alpha_2, \ldots, \alpha_s \ge 0$ and $\sum_{i=1}^s \alpha_i = 1$,
	$\psi(\sum_{i=1}^{s} \alpha_i \cdot \vv_i )=  \sum_{i=1}^{s} \alpha_i \cdot \psi(\vv_i)$.
	Suppose that $D$ contains a set of linearly independent basis vectors of $\R^M$,
	$\{\vb_1, \vb_2, \ldots, \vb_M\}$ and also vector $\vzero$.
	Then for any $\vv\in D$, which can be represented as $\vv = \sum_{i=1}^M \lambda_i \cdot \vb_i$
	for some $\lambda_1, \lambda_2, \ldots, \lambda_M \in \R$,
	we have 
	\[
	\psi(\vv) = \psi\left(\sum_{i=1}^M \lambda_i \cdot \vb_i\right) 
	= \sum_{i=1}^M \lambda_i \cdot \psi(\vb_i) + \left(1-\sum_{i=1}^M \lambda_i \right) \cdot \psi(\vzero).
	\]
\end{lemma}
\begin{proof}
	We consider the convex hull formed by $\{\vb_1, \vb_2, \ldots, \vb_M\}$ together 
	with $\vzero$.
	Let $\vv^{(0)} = \frac{1}{M+1} (\sum_{i=1}^M \vb_i + \vzero)$, 
        which is an interior point in the convex hull.
	For any $\vv \in D$, since $\{\vb_1, \vb_2, \ldots, \vb_M\}$ is a set of basis, 
	we have  $\vv = \sum_{i=1}^M \lambda_i \cdot \vb_i$
	for some $\lambda_1, \lambda_2, \ldots, \lambda_M \in \R$.
	Let $\vv^{(1)} = \rho \vv^{(0)} + (1-\rho) \vv$ with $\rho \in (0,1)$ be a convex 
	combination of $\vv^{(0)}$ and $\vv$.
	Then we have $\psi(\vv^{(1)}) = \rho \psi(\vv^{(0)})  + (1-\rho) \psi(\vv)$, or equivalently
	\begin{equation} \label{eq:convertlin}
	\psi(\vv) = \frac{1}{1-\rho}\psi(\vv^{(1)}) - \frac{\rho}{1-\rho}\psi(\vv^{(0)}).
	\end{equation}
	
	We select a $\rho$ close enough to $1$ such that for all $i\in [M]$, 
	$\frac{\rho}{M+1} + (1-\rho)\lambda_i \ge 0$, and
	$\frac{\rho}{M+1} + (1-\rho)(1 - \sum_{i=1}^M\lambda_i) \ge 0$.
	Then $\vv^{(1)} = \sum_{i=1}^M (\frac{\rho}{M+1} + (1-\rho)\lambda_i) \vb_i + 
	(\frac{\rho}{M+1} + (1-\rho)(1 - \sum_{i=1}^M\lambda_i)) \vzero$ 
  is in the convex hull of $\{\vb_1, \vb_2, \ldots, \vb_M, \vzero\}$.  
	Then from Eq.\eqref{eq:convertlin}, we have
	\begin{align*}
	\psi(\vv) = & \psi\left(\sum_{i=1}^M \lambda_i \cdot \vb_i\right) \\
	= & \frac{1}{1-\rho}\psi\left(\sum_{i=1}^M 
	\left(\frac{\rho}{M+1} + (1-\rho)\lambda_i \right) \vb_i + \right.\\
	& \quad  \left. \left(\frac{\rho}{M+1} + (1-\rho)\left(1 - \sum_{i=1}^M\lambda_i \right) \right) \vzero\right) - \\
	& \quad \frac{\rho}{1-\rho}\psi\left(\frac{1}{M+1} \left(\sum_{i=1}^M \vb_i +\vzero \right) \right) \\
	= & \frac{1}{1-\rho}\left(\sum_{i=1}^M 
	\left(\frac{\rho}{M+1} + (1-\rho)\lambda_i \right) \psi(\vb_i) + \right. \\
	& \quad  \left. \left(\frac{\rho}{M+1} + (1-\rho)\left(1 - \sum_{i=1}^M\lambda_i \right) \right) \psi(\vzero)\right) - \\
	& \quad \frac{\rho}{1-\rho}\left(\frac{1}{M+1} \left(\sum_{i=1}^M \psi(\vb_i) + \psi(\vzero) \right) \right) \\
	= & \sum_{i=1}^M \lambda_i \psi(\vb_i)+ \left(1-\sum_{i=1}^M \lambda_i \right) \cdot \psi(\vzero).
	\end{align*}
\end{proof}


\begin{lemma}[Completeness] \label{lem:unique}
	The centrality measure satisfying axiom set $\cA$ is unique.
\end{lemma}
\begin{proof}
	Let $\psi$ be a centrality measure that satisfies axiom set $\cA$.

       	Fix a set $V$. 	Let the {\em null influence instance} $\calI^N$ to be the instance in which no seed set
	has any influence except to itself, that is, For any $S\subseteq V$, 
	$P_{\calI^N}(S,S) = 1$.
	It is straightforward to check that every node is an isolated node in the null instance,
	and thus by the Anonymity Axiom (Axiom~\ref{axiom:anonymity}) and the Normalization Axiom
		(Axiom~\ref{axiom:normalization}), 
	we have $\psi_v(\calI^N) = 1$ for all $v\in V$.
	That is, $\psi_v(\calI^N)$ is uniquely determined.
Note that, by our canonical convention of influence-profile vector space, 
        $P_{\calI^N}(S,S)$ is not in the vector representation of  $P_{\calI^N}$.
Thus vector $P_{\calI^N}$ is the all-$0$ vector in $\R^M$.
	By Lemma~\ref{lem:linearind}, we know that $\cV$ is a set of basis for $\R^M$.
	Then for any influence instance $\calI$, 
	\[
	P_\calI = \sum_{R\not\in\{\emptyset,V\}, R\subset U} \lambda_{R,U} \cdot P_{\calI_{R,U}},
	\]
	where parameters $\lambda_{R,U}\in \R$.
	Because of the Bayesian Influence Axiom (Axiom~\ref{axiom:bayesian}), and the fact that
	the all-$0$ vector in $\R^M$ is the influence instance $\calI^N$, we can apply
	Lemma~\ref{lem:linearmap} and obtain:
	\begin{align} \label{eq:linearinf}
	& \psi(P_\calI) = \sum_{R\not\in\{\emptyset,V\}, R\subset U} \lambda_{R,U} \cdot \psi(P_{\calI_{R,U}}) \nonumber \\
	& \quad \quad + \left(1 - \sum_{R\not\in\{\emptyset,V\}, R\subset U} \lambda_{R,U} \right) \psi(P_{\calI^N}),
	\end{align}
	where the notation $\psi(P_\calI) $ is the same as $\psi(\calI)$.
	By Lemma~\ref{lem:uniquecritical} we know that all $\psi(P_{\calI_{R,U}})$'s are uniquely
	determined.
	By the argument above, we also know that $\psi(P_{\calI^N})$ is uniquely determined.
	Therefore, $\psi(P_\calI)$ must be unique.
\end{proof}

\subsubsection*{{\sc Independence}}

An axiom is {\em independent} if it cannot be implied by other axioms in the axiom set.
Thus, if an axiom is not independent, 
	the centrality measure satisfying the rest axioms should still be unique
	by Lemma~\ref{lem:unique}. 
Therefore, to show the independence of an axiom,  it is sufficient to show that there is a 
centrality measure different from the Shapley centrality that satisfies the rest axioms.
We will show the independence of each axiom in $\cA$ in the next series of lemmas.

\begin{lemma} \label{lem:indaxiom1}
The Anonymity Axiom (Axiom~\ref{axiom:anonymity}) is independent.
\end{lemma}
\begin{proof}
We consider a centrality measure $\psi^{(1)}$ defined as follows.
Let $\Pi'$ be a nonuniform distribution on all permutations over set $V$, such that for any
node $v\in V$, the probability that $v$ is ordered at the last position in 
a random permutation $\pi$ drawn from $\Pi'$ is $1/|V|$, but the probabilities of $v$ in
other positions may not be uniform.
Such a nonuniform distribution can be achieved by uniformly pick $v \in V$ and put $v$
in the last position, and then apply an arbitrary nonuniform distribution for the 
rest $|V| -1$ positions.
We then define $\psi^{(1)}$ as:
\begin{equation*} 
\psi^{(1)}_v(\calI) = \E_{\bpi\sim \Pi'}[\sigma_{\calI}(S_{\bpi,v} \cup \{v\}) - \sigma_{\calI}(S_{\bpi,v})].
\end{equation*}
Since $\Pi'$ is nonuniform, the above defined $\psi^{(1)}$ is not Shapley centrality, although it has the same
form.

We now verify that $\psi^{(1)}$ satisfies Axioms~\ref{axiom:normalization}--\ref{axiom:critical}.
Actually, since $\psi^{(1)}$ follows the same form as $\psi^{\Shapley}$, one can easily check that
it also satisfies Axioms~\ref{axiom:normalization}, \ref{axiom:sink} and \ref{axiom:bayesian}.
In particular, for Axiom~\ref{axiom:sink}, one can check the proof of Lemma~\ref{lem:shapleycen} and see that
the proof for Shapley centrality satisfying Axiom~\ref{axiom:sink} does not rely on whether 
random permutation $\bpi$ is drawn from a uniform or nonuniform distribution of permutations.
Thus the same proof works for the current $\psi^{(1)}$.
For Axiom~\ref{axiom:critical}, following the same proof as in the proof of Lemma~\ref{lem:shapleycen}, we have
$\psi_v^{(1)}(\calI_{R,v}) = \Pr_{\bpi \sim \Pi'}(R \not\subseteq S_{\bpi, v}) = 1 -\Pr_{\bpi \sim \Pi'}(R = S_{\bpi, v}) $.
As we know, for distribution $\Pi'$, node $v$ appearing as the last node in a random permutation $\bpi$ drawn
from $\Pi'$ is $1/|V| = 1/(|R| +1 )$, which is exactly $\Pr_{\bpi \sim \Pi'}(R = S_{\bpi, v})$.
Therefore, we have $\psi_v^{(1)}(\calI_{R,v})  = |R|/(|R|+1)$.
Axiom~\ref{axiom:critical} also holds.

As a remark, the Anonymity Axiom~\ref{axiom:anonymity} does not hold for $\psi^{(1)}$:
Consider the influence instance $\calI$
where every subset deterministically influences all nodes.
In this case, for any permutation $\pi$, $\pi(\calI)$ is the same as $\calI$, because every node is symmetric.
Axiom~\ref{axiom:anonymity} says in this case all nodes should have the same centrality.
Notice that by our definition of $\psi^{(1)}$, $\psi_v^{(1)}$ is exactly $|V|$ times the probability of $v$ being ranked first
in a random permutation $\bpi$ drawn from $\Pi'$.
But since $\Pi'$ is nonuniform, some node $u$ would have higher probability to be ranked first
than some other node $v$, and thus $\psi_u^{(1)}(\calI) > \psi_v^{(1)}(\calI)$, and Axiom~\ref{axiom:anonymity} does not hold.
\end{proof}

\begin{lemma} \label{lem:indaxiom2}
	The Normalization Axiom (Axiom~\ref{axiom:normalization}) is independent.
\end{lemma}
\begin{proof}
For this axiom, we define $\psi^{(2)}$ first on the critical set instances in $\cV$, 
	and then use their linear independence to define $\psi^{(2)}$ on all instances.

For every instance $\calI_{R,U} \in \cV$, we define 
\begin{equation} \label{eq:defpsi2critiacal}
\psi^{(2)}_v(\calI_{R,U}) = 
\left\{
\begin{array}{lr}
a_{|R|,|U|, |V|} 		& v \in R, \\
\frac{|R|}{|R|+1}  	& v \in U \setminus R, \\
c 					& v \in V \setminus U.
\end{array}
\right.
\end{equation}
We can show that by Axioms~\ref{axiom:anonymity}, \ref{axiom:sink} and~\ref{axiom:critical},
	the above centrality assignments are the only possible assignments.
In fact, for every $v \in V\setminus U$, we can repeatedly apply Axiom~\ref{axiom:sink} to remove
	nodes in $U\setminus R$ first, and then remove all but $v$ to get a single node instance,
	which must have one centrality value, and we denote it $c$.
For every $v\in U \setminus R$, we apply Axiom~\ref{axiom:sink} again to remove all nodes in
	$V \setminus (R\cup \{v\})$. and then apply Axiom~\ref{axiom:critical} to show that $v$
	must have centrality $|R|/(|R|+1)$. 
For every node $v \in R$, by Anonymity Axiom, they must have the same centrality within the 
	same instance $\calI_{R,U}$, and then further apply Anonymity Axiom between two instances
	with the same size of $|V|$, $|R|$ and $|U|$, we know that they all have the same value, and thus
	we can use $a_{|R|,|U|, |V|}$ to denote it.
Thus, for a fixed $|V|$, totally the degree of freedom for $\psi^{(2)}$ is $(|V|-2)(|V|-1)/2 + 1$.

For the null instance $\calI^N$ defined in the proof of Lemma~\ref{lem:unique}
	(where every node is an isolated node), applying Axiom~\ref{axiom:sink}
	repeatedly we know that the centrality of every node in $\calI^N$ must be $c$.

For an arbitrary instance $\calI$, by Lemma~\ref{lem:linearind} we have
\begin{equation} \label{eq:linearprofile}
P_\calI = \sum_{R\not\in\{\emptyset,V\}, R\subset U} \lambda^{\calI}_{R,U} \cdot P_{\calI_{R,U}},
\end{equation}
where $\lambda^{\calI}_{R,U} \in \R$.
Now we define $\psi^{(2)}(\calI)$ as
\begin{align}
	& \psi^{(2)}(\calI) = \sum_{R\not\in\{\emptyset,V\}, R\subset U} \lambda^{\calI}_{R,U} \cdot \psi^{(2)}({\calI_{R,U}}) \nonumber \\
	& \quad \quad + \left(1 - \sum_{R\not\in\{\emptyset,V\}, R\subset U} \lambda^{\calI}_{R,U} \right) \psi^{(2)}(\calI^N).  \label{eq:psi2linear}
\end{align}

It is obvious that the definition of $\psi^{(2)}$ does not depend on node labeling, and thus
	it satisfies Axiom~\ref{axiom:anonymity}.
Since $\psi^{(2)}_v(\calI_{R,v}) = |R|/(|R|+1)$, it satisfies Axiom~\ref{axiom:critical}.
Since all construction is linear, it is not hard to see that it satisfies Axiom~\ref{axiom:bayesian}, 
	and we provide the complete derivation below.
For any Bayesian instance $\calI_{\cB(\{\calI^{\eta}\}, \lambda)}$, by definition 
we have $P_{\calI_{\cB(\{\calI^{\eta}\}, \lambda)}}(S, T)
= \sum_{\eta=1}^{r} \lambda_\eta P_{\calI^\eta}(S, T)$.
Note that here $\lambda_\eta$'s and $\lambda^{\calI}_{R,U} $'s are different sets of parameters.
Thus we have
\begin{align*}
& P_{\calI_{\cB(\{\calI^{\eta}\}, \lambda)}} = 
\sum_{\eta=1}^{r} \lambda_\eta P_{\calI^\eta} \\
& = \sum_{\eta=1}^{r} \lambda_\eta \sum_{R\not\in\{\emptyset,V\}, R\subset U} \lambda^{\calI^\eta}_{R,U} \cdot P_{\calI_{R,U}} \\
& = \sum_{R\not\in\{\emptyset,V\}, R\subset U} 
\left( \sum_{\eta=1}^{r} \lambda_\eta \lambda^{\calI^\eta}_{R,U} \right) \cdot P_{\calI_{R,U}}.
\end{align*}
Then by Eq.~\eqref{eq:psi2linear},
	\begin{align*} 
	& \psi^{(2)}(\calI_{\cB(\{\calI^{\eta}\}, \lambda)}) = \sum_{R\not\in\{\emptyset,V\}, R\subset U} \sum_{\eta=1}^{r} \lambda_\eta \lambda^{\calI^\eta}_{R,U}  \cdot \psi^{(2)}({\calI_{R,U}})  \\
	& \quad \quad + \left(1 - \sum_{R\not\in\{\emptyset,V\}, R\subset U} \sum_{\eta=1}^{r} \lambda_\eta\lambda^{\calI^\eta}_{R,U} \right) \psi^{(2)}(\calI^N) \\
	& =  \sum_{\eta=1}^{r} \lambda_\eta \cdot \left( \sum_{R\not\in\{\emptyset,V\}, R\subset U} \lambda^{\calI^\eta}_{R,U}  \cdot \psi^{(2)}({\calI_{R,U}}) \right.\\
	& \quad \quad + \left. \left(1 - \sum_{R\not\in\{\emptyset,V\}, R\subset U} \lambda^{\calI^\eta}_{R,U} \right) \psi^{(2)}(\calI^N) \right) \\
	& = \sum_{\eta=1}^{r} \lambda_\eta \psi^{(2)}(\calI^\eta),
	\end{align*}
where the second equality uses the fact that $\sum_{\eta=1}^{r} \lambda_\eta = 1$.
Therefore, the Bayesian Axiom (Axiom~\ref{axiom:bayesian}) holds.

Finally, we verify that the Independence of Sink Node Axiom (Axiom~\ref{axiom:sink}) holds.
Suppose $u$ and $v$ are two sink nodes of an instance $\calI=(V,E, P_{\calI})$.
Since we need to work on projection, we clarify the notation and use ${\calI}_{R,U}^V$ 
	and $\calI^{V,N}$ to
	represent the critical instance and the null instance, respectively, in set $V$.
By the definition of projection and Eq.~\eqref{eq:linearprofile}, we have for any 
	$\emptyset \subset S \subset T \subseteq V \setminus \{v\}$, 
\begin{align}
& P_{\calI\setminus \{v\}}(S,T) = P_{\calI}(S,T) + P_{\calI}(S, T\cup \{v\}) \nonumber \\
& = \sum_{\emptyset\subset R \subset U \subseteq V} \lambda^{\calI}_{R,U} \cdot
	\left( P_{\calI^V_{R,U}}(S,T)  +P_{\calI^V_{R,U}}(S,T\cup \{v\}) \right).
	\label{eq:psi2projection}
\end{align}
For $v\in R$, by the definition of $\calI^V_{R,U}$, we know that for 
	$S \subset T \subseteq V\setminus \{v\}$, 
	$P_{\calI^V_{R,U}}(S,T)  = P_{\calI^V_{R,U}}(S,T\cup \{v\}) = 0$.
For $v\not\in R$, $v$ is a sink node in $\calI^V_{R,U}$, and thus
	$P_{\calI^V_{R,U}}(S,T)  +P_{\calI^V_{R,U}}(S,T\cup \{v\}) = 
		P_{\calI^V_{R,U}\setminus \{v\}}(S,T)$.
From Lemma~\ref{lem:projcritical}, 	we know that the projection
	$\calI^V_{R,U}\setminus \{v\} = \calI^{V\setminus \{v\}}_{R,U}$ when $v \in V\setminus U$, and
	$\calI^V_{R,U}\setminus \{v\} = \calI^{V\setminus \{v\}}_{R,U\setminus \{v\}}$ when
	if $v\in U\setminus R$.
In particular, if $U = R\cup\{v\}$, then $\calI^{V\setminus \{v\}}_{R,U\setminus \{v\}}$ is the
	null instance where every node is an isolated node, in which case
	$P_{\calI^{V\setminus \{v\}}_{R,U\setminus \{v\}}}(S,T) = 0$ for any $S\subset T$.
Combining the above, we continue Eq.~\eqref{eq:psi2projection} to have
\begin{align}
& P_{\calI\setminus \{v\}}(S,T) \nonumber \\
& = \sum_{\emptyset\subset R \subset U \subseteq V, v\not\in U} \lambda^{\calI}_{R,U} \cdot
	P_{\calI^{V\setminus \{v\}}_{R,U}}(S,T) + \nonumber\\
& \quad \sum_{\emptyset\subset R \subset U \subseteq V, v\in U\setminus R} \lambda^{\calI}_{R,U} \cdot
	P_{\calI^{V\setminus \{v\}}_{R,U\setminus \{v\}}}(S,T) \nonumber \\
& = \sum_{\emptyset\subset R \subset U \subseteq V\setminus \{v\}} \lambda^{\calI}_{R,U} \cdot
P_{\calI^{V\setminus \{v\}}_{R,U}}(S,T) + \nonumber\\
& \quad \sum_{\emptyset\subset R \subset U \subseteq V\setminus \{v\}} \lambda^{\calI}_{R,U\cup \{v\}} \cdot
P_{\calI^{V\setminus \{v\}}_{R,U}}(S,T) \nonumber \\
& = \sum_{\emptyset\subset R \subset U \subseteq V\setminus \{v\}} 
	\left(\lambda^{\calI}_{R,U} + \lambda^{\calI}_{R,U\cup \{v\}} \right) \cdot
P_{\calI^{V\setminus \{v\}}_{R,U}}(S,T). \nonumber
\end{align}
Since $P_{\calI\setminus \{v\}}$ has unique linear representation from 
	$P_{\calI^{V\setminus \{v\}}_{R,U}}$'s, for all $\emptyset\subset R \subset U \subseteq V\setminus \{v\}$ we have
\begin{equation} \label{eq:lambdaprojection}
\lambda^{\calI\setminus \{v\}}_{R,U}  = \lambda^{\calI}_{R,U} + \lambda^{\calI}_{R,U\cup \{v\}}.
\end{equation}

We now derive the relation between $\psi^{(2)}_u(\calI\setminus \{v\})$ and
	$\psi^{(2)}_u(\calI)$.
By Eqs.~\eqref{eq:defpsi2critiacal}, \eqref{eq:psi2linear} and \eqref{eq:lambdaprojection},
\begin{align}
&\psi^{(2)}_u(\calI\setminus \{v\}) = \sum_{\emptyset\subset R \subset U \subseteq V\setminus \{v\}} 
	\lambda^{\calI\setminus \{v\}}_{R,U} \psi^{(2)}_u(\calI^{V\setminus\{v\}}_{R,U}) 
		\nonumber \\
& \quad + \left(1 -   \sum_{\emptyset\subset R \subset U \subseteq V\setminus \{v\}} 
\lambda^{\calI\setminus \{v\}}_{R,U} \right) \psi^{(2)}_u(\calI^{V\setminus \{v\},N}) \nonumber \\
& = \sum_{\emptyset\subset R \subset U \subseteq V\setminus \{v\}, u\not\in U} 
 \left(\lambda^{\calI}_{R,U} + \lambda^{\calI}_{R,U\cup \{v\}} \right) \cdot c
 \nonumber \\
& \quad +  \sum_{\emptyset\subset R \subset U \subseteq V\setminus \{v\}, u\in U\setminus R} 
\left(\lambda^{\calI}_{R,U} + \lambda^{\calI}_{R,U\cup \{v\}} \right) \cdot \frac{|R|}{|R|+1}
\nonumber \\
& \quad +  \sum_{\emptyset\subset R \subset U \subseteq V\setminus \{v\}, u\in R} 
\left(\lambda^{\calI}_{R,U} + \lambda^{\calI}_{R,U\cup \{v\}} \right) \cdot a_{|R|,|U|, |V|-1}
\nonumber \\
& \quad + \left(1 -   \sum_{\emptyset\subset R \subset U \subseteq V\setminus \{v\}} 
\left(\lambda^{\calI}_{R,U} + \lambda^{\calI}_{R,U\cup \{v\}} \right) \right) \cdot c \nonumber \\
& =  \sum_{\emptyset\subset R \subset U \subseteq V\setminus \{v\}, u\in U\setminus R} 
\left(\lambda^{\calI}_{R,U} + \lambda^{\calI}_{R,U\cup \{v\}} \right) \cdot \frac{|R|}{|R|+1}
\nonumber \\
& \quad +  \sum_{\emptyset\subset R \subset U \subseteq V\setminus \{v\}, u\in R} 
\left(\lambda^{\calI}_{R,U} + \lambda^{\calI}_{R,U\cup \{v\}} \right) \cdot a_{|R|,|U|, |V|-1}
\nonumber \\
& \quad + \left(1 -   \sum_{\emptyset\subset R \subset U \subseteq V\setminus \{v\}, u\in U} 
\left(\lambda^{\calI}_{R,U} + \lambda^{\calI}_{R,U\cup \{v\}} \right) \right) \cdot c \label{eq:psi2projected}
\end{align}
Note that
\begin{align*}
& \sum_{\emptyset\subset R \subset U \subseteq V\setminus \{v\}, u\in U} 
\left(\lambda^{\calI}_{R,U} + \lambda^{\calI}_{R,U\cup \{v\}} \right) \\
& = \sum_{\emptyset\subset R \subset U \subseteq V, u\in U, v\not\in U} \lambda^{\calI}_{R,U}
	+  \sum_{\emptyset\subset R \subset U \subseteq V, u\in U, v\in U\setminus R} \lambda^{\calI}_{R,U} \\
& = \sum_{\emptyset\subset R \subset U \subseteq V, u\in U, v\not\in R} \lambda^{\calI}_{R,U}.
\end{align*}
\begin{align*}
& \sum_{\emptyset\subset R \subset U \subseteq V\setminus \{v\}, u\in U\setminus R} 
\left(\lambda^{\calI}_{R,U} + \lambda^{\calI}_{R,U\cup \{v\}} \right) \cdot \frac{|R|}{|R|+1} \\
& = \sum_{\emptyset\subset R \subset U \subseteq V, u\in U \setminus R, v\not\in U } \lambda^{\calI}_{R,U}
	\cdot \frac{|R|}{|R|+1} \\
& \quad +  \sum_{\emptyset\subset R \subset U \subseteq V, u\in U \setminus R, v\in U\setminus R} \lambda^{\calI}_{R,U} 
	\cdot \frac{|R|}{|R|+1} \\
& = \sum_{\emptyset\subset R \subset U \subseteq V, u\in U\setminus R, v\not\in R} \lambda^{\calI}_{R,U}
	\cdot \frac{|R|}{|R|+1}.
\end{align*}
\begin{align*}
& \sum_{\emptyset\subset R \subset U \subseteq V\setminus \{v\}, u\in R} 
\left(\lambda^{\calI}_{R,U} + \lambda^{\calI}_{R,U\cup \{v\}} \right) \cdot a_{|R|,|U|, |V|-1} \\
& = \sum_{\emptyset\subset R \subset U \subseteq V, u\in R, v\not\in U} \lambda^{\calI}_{R,U} \cdot a_{|R|,|U|, |V|-1} \\
& \quad 
+  \sum_{\emptyset\subset R \subset U \subseteq V, u\in R, v\in U\setminus R} \lambda^{\calI}_{R,U} 
	\cdot a_{|R|,|U|-1, |V|-1}
\end{align*}
Plugging the above three results into Eq.~\eqref{eq:psi2projected}, we have
\begin{align}
&\psi^{(2)}_u(\calI\setminus \{v\}) = 
\sum_{\emptyset\subset R \subset U \subseteq V, u\in U\setminus R, v\not\in R} \lambda^{\calI}_{R,U}
\cdot \frac{|R|}{|R|+1} \nonumber \\
& \quad  + \sum_{\emptyset\subset R \subset U \subseteq V, u\in R, v\not\in U} \lambda^{\calI}_{R,U} \cdot a_{|R|,|U|, |V|-1} 
	\nonumber \\
& \quad + \sum_{\emptyset\subset R \subset U \subseteq V, u\in R, v\in U\setminus R} \lambda^{\calI}_{R,U} 
	\cdot a_{|R|,|U|-1, |V|-1}
	\nonumber \\
& \quad + \left( 1 - \sum_{\emptyset\subset R \subset U \subseteq V, u\in U, v\not\in R} \lambda^{\calI}_{R,U} \right) 
	\cdot c \label{eq:psi2projclean}
\end{align}

Similarly, we expand $\psi^{(2)}_u(\calI)$:
\begin{align}
&\psi^{(2)}_u(\calI) = \sum_{\emptyset\subset R \subset U \subseteq V} 
\lambda^{\calI}_{R,U} \psi^{(2)}_u(\calI^{V}_{R,U}) 
\nonumber \\
& \quad + \left(1 -   \sum_{\emptyset\subset R \subset U \subseteq V} 
\lambda^{\calI}_{R,U} \right) \psi^{(2)}_u(\calI^{V,N}) \nonumber \\
& = \sum_{\emptyset\subset R \subset U \subseteq V, u\not\in U} 
\lambda^{\calI}_{R,U}  \cdot c
\nonumber  +  \sum_{\emptyset\subset R \subset U \subseteq V, u\in U\setminus R} 
\lambda^{\calI}_{R,U} \cdot \frac{|R|}{|R|+1}
\nonumber \\
& \quad +  \sum_{\emptyset\subset R \subset U \subseteq V, u\in R} 
\lambda^{\calI}_{R,U} \cdot a_{|R|,|U|, |V|}
\nonumber  + \left(1 -   \sum_{\emptyset\subset R \subset U \subseteq V} 
\lambda^{\calI}_{R,U}  \right) \cdot c \nonumber \\
& = \sum_{\emptyset\subset R \subset U \subseteq V, u\in U\setminus R} 
\lambda^{\calI}_{R,U} \cdot \frac{|R|}{|R|+1} \nonumber \\
& \quad  +  \sum_{\emptyset\subset R \subset U \subseteq V, u\in R} 
\lambda^{\calI}_{R,U} \cdot a_{|R|,|U|, |V|} \nonumber  \\
& \quad + \left(1 -   \sum_{\emptyset\subset R \subset U \subseteq V, u\in U} 
\lambda^{\calI}_{R,U}  \right) \cdot c \label{eq:psi2nonprojected}
\end{align}
Subtracing Eq.~\eqref{eq:psi2projclean} from Eq.~\eqref{eq:psi2nonprojected}, we have
\begin{align}
&\psi^{(2)}_u(\calI)  - \psi^{(2)}_u(\calI\setminus \{v\}) =
	\sum_{\emptyset\subset R \subset U \subseteq V, u\in U\setminus R, v\in R} \lambda^{\calI}_{R,U}
	\cdot \frac{|R|}{|R|+1} \nonumber \\
& \quad  + \sum_{\emptyset\subset R \subset U \subseteq V, u\in R} \lambda^{\calI}_{R,U} \cdot a_{|R|,|U|, |V|} 
	\nonumber \\
& \quad - \sum_{\emptyset\subset R \subset U \subseteq V, u\in R, v \not \in U } \lambda^{\calI}_{R,U} 
\cdot a_{|R|,|U|, |V|-1}
\nonumber \\
& \quad - \sum_{\emptyset\subset R \subset U \subseteq V, u\in R, v\in U \setminus R} \lambda^{\calI}_{R,U} 
	\cdot a_{|R|,|U|-1, |V|-1}
	\nonumber \\
& \quad  - \sum_{\emptyset\subset R \subset U \subseteq V, u\in U, v\in R} \lambda^{\calI}_{R,U}  
	\cdot c. \label{eq:psi2diff}
\end{align}

We want the above difference to be zero, but so far we have not used the property that
	both $u$ and $v$ are sink nodes in $\calI$ yet, except that the project
	$\calI\setminus \{v\}$ is defined when $v$ is a sink node.
Next, suppose that $v$ is a sink node in $\calI$ and we would derive some properties 
	on $\lambda^{\calI}_{R,U}$ based on this fact.
By the definition of sink nodes, we have  for
all $\emptyset\subseteq S\subset T \subseteq V \setminus \{v\}$, 
\begin{align*}
& \sum_{\emptyset\subset R \subset U \subseteq V} \lambda^{\calI}_{R,U} 
P_{\calI^V_{R,U}}(S \cup\{v\},T \cup \{v\}) \\
& = P_{\calI}(S\cup \{v\}, T\cup \{v\}) \\
& = P_{\calI}(S,T) + P_{\calI}(S, T\cup\{v\})\\
& = \sum_{\emptyset\subset R \subset U \subseteq V} \lambda^{\calI}_{R,U} 
	\left( P_{\calI^V_{R,U}}(S,T) + P_{\calI^V_{R,U}}(S,T \cup \{v\}) \right). \\
\end{align*}
Note that when $v\not \in R$, $v$ is a sink node in $\calI^V_{R,U}$, and so
	we have $P_{\calI^V_{R,U}}(S,T) + P_{\calI^V_{R,U}}(S,T \cup \{v\}) = P_{\calI^V_{R,U}}(S \cup\{v\},T \cup \{v\})$.
When $v \in R$, since $v\not\in S$ and $S \subset T$, we have
	$P_{\calI^V_{R,U}}(S,T) + P_{\calI^V_{R,U}}(S,T \cup \{v\}) = 0$.
Thus the above implies that
\[
\sum_{\emptyset\subset R \subset U \subseteq V, v\in R} \lambda^{\calI}_{R,U} 
P_{\calI^V_{R,U}}(S \cup\{v\},T \cup \{v\}) = 0.
\]
Note that if $R \not \subseteq S \cup\{v\}$, then
	$P_{\calI^V_{R,U}}(S \cup\{v\},T \cup \{v\}) = 0$.
When  $R \subseteq S \cup\{v\}$, $P_{\calI^V_{R,U}}(S \cup\{v\},T \cup \{v\}) = 1$
	if and only if $T \cup \{v\} = S \cup\{v\} \cup U$; otherwise it is $0$.
Thus the above is equivalent to
\begin{equation}
\sum_{R,U: \emptyset \subset R\subset U, v\in R, R \subseteq S \cup\{v\}, T \cup \{v\} = S \cup\{v\} \cup U} \lambda^{\calI}_{R,U}  = 0. \label{eq:lambdasinkcond}
\end{equation}
Set $S = \emptyset$ first.
Then $R$ must be $\{v\}$.
For any $T$ such that $S \subset T \subseteq V\setminus \{v\}$, 
	we see that $U$ must be $T \cup \{v\}$ to satisfy the constraint in the above summation.
Thus we have $\lambda^{\calI}_{\{v\},  U} = 0$ for any $\{v\} \subset U \subseteq V$.
Now, let $|S| = 1$.
In this case $R = \{v\}$ or $R = S\cup \{v\}$.
We already know from above that if $R= \{v\}$, $\lambda^{\calI}_{\{v\},  U} = 0$.
Thus in Eq.~\eqref{eq:lambdasinkcond}, what are left are the terms with $R= S \cup\{v\}$.
When $R= S \cup\{v\}$, we can see that $U$ must be $T\cup \{v\}$.
Then we obtain that $\lambda^{\calI}_{R,U} = 0$ for every $|R|=2$, $v\in R$, and $R \subset U\subseteq V$.
Repeating the above argument for $|R|=3, 4, \ldots$ (or $|S|=2, 3, \ldots,$), we eventually conclude
	that for every $R$ and $U$ such that $v\in R$ and $R\subset U \subseteq V$, 
	$\lambda^{\calI}_{R,U} = 0$.
	
With the above important property, we look back at Eq.~\eqref{eq:psi2diff}.
Notice that in all the five summation terms, we have either $v\in R$ or $u \in R$, and both
	$u$ and $v$ are sink nodes.
Therefore, all these five summation terms are $0$, and finally we conclude that
	$\psi^{(2)}_u(\calI)  = \psi^{(2)}_u(\calI\setminus \{v\}) $ for an arbitrary instance $\calI$.
This means that the Independence of Sink Nodes Axiom (Axiom~\ref{axiom:sink}) always holds for
	the definition of $\psi^{(2)}$ (Eq.~\eqref{eq:defpsi2critiacal}) with any possible parameters
	$a_{|R|,|U|,|V|}$'s and $c$.
Hence, we have many degree of freedom to chose these parameters other than the ones determined
	by the Shapley centrality, and thus the Normalization Axiom (Axiom~\ref{axiom:normalization})
	is independent.
\end{proof}

\begin{lemma} \label{lem:indaxiom3}
	The Independence of Sink Nodes Axiom (Axiom~\ref{axiom:sink}) is independent.
\end{lemma}
\begin{proof}
Similar to the proof of Lemma~\ref{lem:indaxiom2}, 
	we first define $\psi^{(3)}$ on critical set instances in $\cV$, and then use their linear
	independence to extend the definition to all instances.

For every instance $\calI_{R,U} \in \cV$, we know that $R \subset U$ and $R \not \in \{\emptyset,V\}$.
When $|R| = |V|-1$, we define $\psi^{(3)}_v(\calI_{R,U}) = |R|/(|R|+1)$ for the
	unique $v\in V\setminus R$, and 
	$\psi^{(3)}_u(\calI_{R,U}) = (|V| - |R|/(|R|+1))/|R|$ for every $u \in R$.
When $|R| \ne |V| -1 $, we simply define $\psi^{(3)}_v(\calI_{R,U}) = 1$ for all $v\in V$.
It is straightforward to see that for every $\calI_{R,U} \in \cV$, $\psi^{(3)}$ is anonymous
	(not depend on the label of a node), and normalized (centrality summed up to $|V|$), and
	for the critical set instance $\calI_{R,v}$, it satisfies the requirement of Axiom~\ref{axiom:critical}.
For the null instance $\calI^N$ defined in the proof of Lemma~\ref{lem:unique}, 
	we define $\psi^{(3)}_v(\calI^N) = 1$ for every $v\in V$.
Thus for $\calI^N$ $\psi^{(3)}$ is also anonymous and normalized.

Now for an arbitrary influence instance $\calI$, by Lemma~\ref{lem:linearind} we have
	\[
	P_\calI = \sum_{R\not\in\{\emptyset,V\}, R\subset U} \lambda^{\calI}_{R,U} \cdot P_{\calI_{R,U}},
	\]
	where $\lambda^{\calI}_{R,U} \in \R$.
Then we define $\psi^{(3)}(\calI)$ below patterned by Eq.~\eqref{eq:linearinf}:
	\begin{align} 
	& \psi^{(3)}(\calI) = \sum_{R\not\in\{\emptyset,V\}, R\subset U} \lambda^{\calI}_{R,U} \cdot \psi^{(3)}({\calI_{R,U}}) \nonumber \\
	& \quad \quad + \left(1 - \sum_{R\not\in\{\emptyset,V\}, R\subset U} \lambda^{\calI}_{R,U} \right) \psi^{(3)}({\calI^N}). \label{eq:psi3linear}
	\end{align}
It is straightforward to verify that when all $\calI_{R,U} \in \cV$ and $\calI^N$ are anonymous 
	and normalized, $\calI$ is also anonymous and normalized, and thus $\psi^{(3)}$ satisfies
	Axioms~\ref{axiom:anonymity} and~\ref{axiom:normalization}.
Moreover, $\psi^{(3)}$ also satisfies the Bayesian Axiom (Axiom~\ref{axiom:bayesian}), with the
	same proof as the one in the proof of Lemma~\ref{lem:indaxiom2}.
By its definition, we already know that $\psi^{(3)}$ satisfies Axiom~\ref{axiom:critical}.
Obviously, $\psi^{(3)}$ is different from $\psi^{\Shapley}$ and it does not satisfies the
	Independence of Sink Nodes Axiom (Axiom~\ref{axiom:sink}).
Hence, Axiom~\ref{axiom:sink} is independent.
\end{proof}

\begin{lemma} \label{lem:indaxiom4}
	The Bayesian Axiom (Axiom~\ref{axiom:bayesian}) is independent.
\end{lemma}
\begin{proof}
We define a centrality measure $\psi^{(4)}$ as follows.
Given an influence instance $\calI = (V, E, P_{\calI})$, 
for every sink node $v$ in $\calI$, if there is a set $R \subseteq V \setminus \{v\}$, such that
$\sum_{T: R\cup \{v\} \subseteq T} P_{\calI}(R, T) = 1$, 
we let $R_v$ be the smallest such set (tie is broken with some arbitrary deterministic rule); if such $R$ does not exists,
then we let $R_v = \emptyset$.
Intuitively, $R_v$ is the smallest set that can influence $v$ with probability $1$.
Then, for every sink node $v$, we define $\psi^{(4)}_v(\calI) = |R_v| / (|R_v| + 1)$; for non-sink nodes,
we let them equally devide the rest share so that the total centrality is $|V| $.

The definition of $\psi^{(4)}$ does not depend on node labeling, so Axiom~\ref{axiom:anonymity} is clearly satisfied.
The definition enforces that the sum of all centralities is $|V|$, so Axiom~\ref{axiom:normalization} is satisfied.	
For the critical set infance $\calI_{R,v}$, $v$ is a sink node, and by definition $R$ is the smallest one such that 
$P_{\calI_{R,v}}(R, R\cup\{v\}) = 1$, so by the definition of $\psi^{(4)}$, we have $\psi^{(4)}_v(\calI_{R,v}) = |R|/(|R|+1)$,
thus satisfying Axiom~\ref{axiom:critical}.

For the Independence of Sink Nodes Axiom (Axiom~\ref{axiom:sink}), consider an influence instance $\calI$ and
its two sink nodes $u$ and $v$.
We claim that the projection $\calI\setminus \{v\}$ does not change set $R_u$.
In fact, suppose first that in $\calI$ there is a set $R \subseteq V \setminus \{u\}$ 
$\sum_{T: R\cup \{u\} \subseteq T} P_{\calI}(R, T) = 1$.
If $v\not\in R$, then by the definition of projection, 
\begin{align*}
& \sum_{T: R\cup \{u\} \subseteq T \subseteq V\setminus \{v\}} P_{\calI\setminus \{v\}}(R, T) \\
& = \sum_{T: R\cup \{u\} \subseteq T \subseteq V\setminus \{v\}} P_{\calI}(R, T) + P_{\calI}(R, T \cup\{v\}) \\
& = \sum_{T: R\cup \{u\} \subseteq T \subseteq V} P_{\calI}(R, T) = 1.\\
\end{align*}
Thus $R$ is still a set in $\calI\setminus \{v\}$ that influence $v$ with probability $1$.
If $v \in R$, since $v$ is a sink node, we have
\begin{align*}
& \sum_{T: R\setminus \{v\} \cup \{u\} \subseteq T \subseteq V} P_{\calI}(R \setminus \{v\}, T) \\
& =  \sum_{T: R \cup \{u\} \subseteq T \subseteq V} 
P_{\calI}(R\setminus \{v\}, T\setminus \{v\}) + P_{\calI}(R\setminus \{v\}, T) \\
& =  \sum_{T: R \cup \{u\} \subseteq T \subseteq V} 
P_{\calI}(R, T) = 1,
\end{align*}
where the second to last equality is by the definition of sink node.
Thus, the above equation implies that $R\setminus \{v\}$ is a smaller set
	that influences $u$ with probability $1$.
The cases of $v \in R$ and $v\not\in R$ together imply that $R_u$ for instance $\calI$ still works for
instance $\calI \setminus \{v\}$.

Conversely, suppose $R$ is a set influencing $u$ with probability $1$ in $\calI\setminus \{v\}$.
Then we have 
\begin{align*}
& \sum_{T: R\cup \{u\} \subseteq T \subseteq V} P_{\calI}(R, T) \\
& = \sum_{T: R\cup \{u\} \subseteq T \subseteq V\setminus \{v\}} 
	P_{\calI}(R, T) + P_{\calI}(R, T \cup \{v\}) \\
& = \sum_{T: R\cup \{u\} \subseteq T \subseteq V\setminus \{v\}} P_{\calI \setminus \{v\}}(R, T) = 1.
\end{align*}
Therefore, $R$ is still a set influencing $u$ with probability $1$ in $\calI$.

Hence, from the above argument from both sides, we know that, either there does not exists
	a set $R$ that influences $u$ with probability $1$ in $\calI$ or $\calI\setminus \{v\}$,
	or the smallest such sets in $\calI$ and $\calI\setminus \{v\}$ are the same.
Therefore, we have $\psi^{(4)}_u(\calI) = \psi^{(4)}_u(\calI \setminus \{v\})
	= |R_v|/(|R_v|+1)$, where $R_v$ is the smallest such set or $\emptyset$.
This means, $\psi^{(4)}$ satisfies the Independence of Sink Nodes Axiom (Axiom~\ref{axiom:sink}).
Since $\psi^{(4)}$ is clearly different from $\psi^{\Shapley}$, we know that
	Axiom~\ref{axiom:sink} is independent.
\end{proof}

\begin{lemma} \label{lem:indaxiom5}
	The Bargaining with Critical Sets Axiom (Axiom~\ref{axiom:critical}) is independent.
\end{lemma}
\begin{proof}
We construct $\psi^{(5)}$ by trivially assigning every node with centrality $1$.
It is obvious that this constant $\psi^{(5)}$ satisfies Axioms~\ref{axiom:anonymity}--\ref{axiom:bayesian},
	and it is different from $\psi^{\Shapley}$.
Thus Axiom~\ref{axiom:critical} is independent.
\end{proof}

Lemmas~\ref{lem:indaxiom1}--\ref{lem:indaxiom5} together implies the following:
\begin{lemma}[Independence] \label{lem:independence}
	All axioms in the axiom set $\cA$ are independent.
\end{lemma}

\begin{proof}[of Theorem~\ref{thm:ShapleyCen}]
	The theorem is proved by combining 
	Lemmas~\ref{lem:shapleycen}, \ref{lem:unique} and~\ref{lem:independence}.
\end{proof}

\subsection{On Theorem~\ref{thm:SNICen}}

The proof of Theorem~\ref{thm:SNICen}, the axiomatic characterization of SNI centrality, follows the same
	structure as the proof of Theorem~\ref{thm:ShapleyCen}.
For soundness, it is easy to verify that SNI centrality satisfies Axioms \ref{axiom:bayesian}, \ref{axiom:uniformsink}, and \ref{axiom:criticalnodes}.
In particular, for the Bayesian Influence Axiom (Axiom~\ref{axiom:bayesian}), we can verify that 
\begin{align*}
&\sigma_{\calI_{\cB(\{\calI^{\eta}\}, \lambda)}}(\{v\})  = 
\sum_{T\subseteq V, v\in T}
P_{\calI_{\cB(\{\calI^{\eta}\}, \lambda)}}(\{v\}, T) \cdot |T| \\
& = \sum_{T\subseteq V, v\in T} \sum_{\eta=1}^{r} \lambda_\eta P_{\calI^\eta}(\{v\}, T) 
= \sum_{\eta=1}^{r} \lambda_\eta  \sum_{T\subseteq V, v\in T} P_{\calI^\eta}(\{v\}, T) \\
& = \sum_{\eta=1}^{r} \lambda_\eta  \sigma_{\calI^\eta}(\{v\}),
\end{align*}
and thus Bayesian Influence Axiom also holds for SNI centrality.

For completeness, since SNI centrality also satisfies the Bayesian Influence Axiom, we
	following the same proof structure as Lemma~\ref{lem:unique}, which
	utilizes the linear mapping lemma~\ref{lem:linearmap}.
All we need to show is that for all critical set instances $\calI_{R,U}$ as well as the null
	instance $\calI^N$, Axioms \ref{axiom:uniformsink}, and \ref{axiom:criticalnodes} dictate
	that their centrality measure is unique.
For the null instance, by the Uniform Sink Node Axiom (Axiom~\ref{axiom:uniformsink}), we know that
	all nodes have centrality value of $1$.
For the critical set instance $\calI_{R,U}$, again by the Uniform Sink Node Axiom, we know that
	all nodes in $V \setminus R$ are sink nodes and thus have centrality value of $1$.
Finally, Axiom~\ref{axiom:criticalnodes} uniquely determines the centrality value of all critical nodes
	in $R$.
Therefore,  Axioms \ref{axiom:bayesian}, \ref{axiom:uniformsink}, and \ref{axiom:criticalnodes}
	uniquely determines the centrality measure, which is SNI centrality.
This proves Theorem~\ref{thm:SNICen}.

For independence, the independence of Axiom~\ref{axiom:criticalnodes} can be shown by considering
	the uniform centrality measure where every node is assigned centrality of $1$.
We can see that the uniform centrality satisfies Axiom~\ref{axiom:bayesian} and~\ref{axiom:uniformsink}
	but not~\ref{axiom:criticalnodes}.
For the independence of Axiom~\ref{axiom:uniformsink}, we can see that Axiom~\ref{axiom:criticalnodes}
	only restricts the nodes in $R$ in the critical set instances $\calI_{R,U}$.
Then we can assign arbitrary values, for example $0$,  to nodes not in $R$ in $\calI_{R,U}$, and thus
	obtaining a centrality measure defined on the critical set instances that is consistent
	with Axiom~\ref{axiom:criticalnodes} but different from $\psi^{\SNI}$.
Next, we use the linearity (Eq.~\eqref{eq:linearinf}) to extend the centrality measure to
	arbitrary instances.
Bayesian Axiom (Axiom~\ref{axiom:bayesian}) holds because our way of linear extension.
Finally, for the independence of Axiom~\ref{axiom:bayesian}, we notice that with Axioms~\ref{axiom:uniformsink} and~\ref{axiom:criticalnodes}, the centrality for all critical
	set instances are uniquely determined, but we have the freedom to define other instances,
	as long as the sink nodes always have centrality of $1$.
This means we can easily find a centrality measure that is different from $\psi^{\SNI}$ but satisfies
	Axioms~\ref{axiom:uniformsink} and~\ref{axiom:criticalnodes}.
Therefore, Axioms~\ref{axiom:bayesian}, ~\ref{axiom:uniformsink} and~\ref{axiom:criticalnodes}
	are all independent of one another.

\section{Shapley Symmetry of Symmetric IC Models} \label{app:symmetricIC}

In this appendix section, 
  we formally prove the Shapley symmetry of 
  the symmetric IC model stated in Section  \ref{sec:Axioms}.
We restate it in the following theorem.
	
\begin{theorem}[Shapley Symmetry of Symmetric IC] \label{thm:symmetricIC}
In any symmetric IC model, the Shapley centrality of every node is the
	same.
\end{theorem}

We first prove the following basic lemma.
\begin{lemma}[Deterministic Undirected Influence] \label{lem:undirected}
Consider an undirected graph $G=(V,E)$, and the IC instance 
  $\calI$ on $G$ in which
	for every undirected edge $(u,v)\in E$, $p_{u,v} = p_{v,u} =1$.
Then, $\psi^{\Shapley}_v({\calI}) = 1$, $\forall v \in V$,
\end{lemma}
\begin{proof}
Let $C$ be the connected component containing node $v$.
For any fixed permutation $\pi$ of $V$, 
  if some other node $u\in C$ appears before $v$ in $\pi$
 --- i.e. $u \in S_{\pi,v}$ --- then because
	all edges have influence probability $1$ in both directions,
	$u$ influences every node in $C$. 
For this permutation,  $v$ has no marginal influence: 
	$\sigma_{\calI}(S_{\pi,v}\cup \{v\}) - \sigma_{\calI}(S_{\pi,v}) = 0$.
If $v$ is the first node in $C$ that appears in $\pi$, then
	$v$ activates every node in $C$, and its marginal spread 
	is $|C|$.
The probability that $v$ appears first among all nodes in $C$ in a
	random permutation $\bpi$ is exactly $1/|C|$.
Therefore:
$$\psi^{\Shapley}_v({\calI}) = 
\E_{\bpi}[\sigma_{\calI}(S_{\bpi,v}\cup \{v\}) - \sigma_{\calI}(S_{\bpi,v})]
= 1/ |C| \cdot |C| = 1.$$
\end{proof}

\begin{proof}[of Theorem~\ref{thm:symmetricIC}]
We will use the following well-known but 
  important observation about symmetric IC models:
We can	use the following {\em undirected} 
  live-edge graph model to represent its influence spread.
For every edge $(u,v) \in E$, since we have $p_{u,v} = p_{v,u}$, 
	we sample an {\em undirected} edge $(u,v)$ with success probability
	$p_{u,v}$.
The resulting {\em undirected} random live-edge graph is denoted as $\bar{\bL}$.
For any seed set $S$, the propagation from the seed set
   can only pass through each edge $(u,v)$ at most once, 
   either from $u$ to $v$ or from $v$ to $u$, but never in both directions.
Therefore, we can apply the {\em Principle of Deferred Decision}
   and only decide the direction of the live edge $(u,v)$
   when the influence process does need to pass the edge.
Hence, the set of nodes reachable from $S$ in the undirected graph $\bar{\bL}$,
	namely $\Gamma(\bar{\bL}, S)$, is the set of activated nodes.
Thus, $\sigma_{\calI}(S) = \E_{\bar{\bL}}[|\Gamma(\bar{\bL}, S)|]$.

For each ``deferred'' realization $\bar{L}$ 
  of $\bar{\bL}$, the propagation on $\bar{L}$
  is the same as treating every edge in $\bar{L}$ having influence probability 
 $1$ in	both directions.
Then, by Lemma~\ref{lem:undirected}, 
      the Shapley centrality of every node on the fixed $\bar{L}$ is the same.
Finally, by taking expectation over the distribution of  $\bar{\bL}$, we have:
\begin{align*}
\psi^{\Shapley}_v({\calI}) 
 = & \E_{\bpi}[\sigma_{\calI}(S_{\bpi,v}\cup \{v\}) - \sigma_{\calI}(S_{\bpi,v})] \\
 = & \E_{\bpi}[\E_{\bar{\bL}}[|\Gamma(\bar{\bL}, S_{\bpi,v}\cup \{v\})| - 
	 |\Gamma(\bar{\bL}, S_{\bpi,v})|]] \\
 = & \E_{\bar{\bL}}[\E_{\bpi}[|\Gamma(\bar{\bL}, S_{\bpi,v}\cup \{v\})| - 
 |\Gamma(\bar{\bL}, S_{\bpi,v})|]] \\
 = & \E_{\bar{\bL}}[1] = 1.
\end{align*}
\end{proof}


\section{Analysis of {\ASVRR}} \label{app:thm1}

%
In this appendix, we provide a complete proof of Theorem~\ref{thm:ASVRR}, and briefly
extend the discussion to the proof of Theorem~\ref{thm:ASVRRsigma}.
In the discussion below, we will use $\bv \sim V$
to denote that $\bv$ is drawn uniformly at random from $V$.
We will use $\bpi \sim \Pi(V)$
to denote that $\bpi$ is a uniform random permutation
of $V$.
Let $\I\{\mathcal{E}\}$ be the indicator function for event $\mathcal{E}$.
Let $m =|E|$ and $n = |V|$. 

\subsection{Unbiasedness and Absolute Normalization of the Shapley Estimator of {\ASVRR}}

We first build connections between random RR sets and the Shapley value
computation.
The following is a straightforward proposition to verify:
\begin{proposition} \label{lem:orderfirst}
	Fix a subset $R \subseteq V$. 
	For any $v\in R$, $\Pr(R \cap S_{\bpi,v} = \emptyset ) = 1/|R|$, where $\bpi\sim \Pi(V)$
	and $S_{\bpi,v}$ is the subset of nodes preceding $v$ 	in 
	$\bpi$.
\end{proposition}
\begin{proof}
	The event $R \cap S_{\bpi,v} = \emptyset $ 
	is equivalent to 
	$\bpi$ placing $v$ ahead of other nodes in $R$.
	Because $\bpi \sim \Pi(V)$, 
	this event happens with probability exactly $1/|R|$.
\end{proof}

\begin{proposition}\label{RRSetLiveGraph}
	A random RR set $\bR$ is equivalently generated by first
	(a) generating a random live-edge graph  $\bL$, and
	(b) selecting  $\bv\sim V$.
	Then, $\bR$ is the set of nodes that can reach $\bv$ in $\bL$.
\end{proposition}


\begin{lemma}[Marginal Contribution] \label{lem:margin}
	Let $\bR$ be a random RR set.
	For any $S\subseteq V$ and $v\in V\setminus S$:
	\begin{eqnarray}
	\sigma(S) &  = &  n\cdot \Pr(S \cap \bR \ne \emptyset),  \label{eq:rrsetinf}\\
	\sigma(S\cup \{v\}) - \sigma(S)&  = & n \cdot \Pr(v\in \bR \wedge S\cap \bR = \emptyset).
	\end{eqnarray}
\end{lemma}
\begin{proof}
	Let $\bL$ be a random live-edge graph generated by the triggering model (see Section~\ref{sec:infmodel}).
	Recall that $\Gamma(L, S)$ denote the set of nodes in graph $L$ 
	reachable from set $S$.
	Then: 
	\begin{align*}
	\sigma(S) = & \E_{\bL}[|\Gamma(\bL, S)|] \\
	= & \E_{\bL} \left[\sum_{u\in V} \I\{u \in \Gamma(\bL, S) \}\right] \\
	= &  n \cdot \E_{\bL} \left[ \sum_{u\in V} \frac{1}{n} \cdot \I\{u \in \Gamma(\bL, S) \}\right] \\
	= &   n \cdot \E_{\bL}\left[ \E_{\bu\sim V} [\I\{\bu \in \Gamma(\bL, S) \} ]   \right] \\
	= &  n \cdot \Pr_{\bL,\bu\sim V}\{\bu \in \Gamma(\bL, S) \},
	\end{align*}
	Note that for any function $f$, and  random variables $\bx, \by$:
	$$\E_{\bx}\left[ \E_{\by} [f(\bx, \by) ]   \right] = \E \left[ \E [f(\bx, \by) \mid \bx = x]   \right].$$
	In other words, we can evaluate the expectation as the following:
	(1) fix the value of random variable $\bx$ to $x$ first, then 
	(2) take the conditional expectation of $f(\bx, \by)$ conditioned upon $\bx = x$,  and finally
	(3) take the expectation according to $\bx$'s distribution.
	
	By Proposition~\ref{RRSetLiveGraph}, event $\bu \in \Gamma(\bL, S) $ is the same as the event $S \cap \bR \ne \emptyset$.
	Hence we have $\sigma(S) = n \cdot \Pr(S\cap \bR \ne \emptyset)$.
	
	Similarly,
	\begin{align*}
	& \sigma(S\cup \{v\}) - \sigma(S) \\
	& \quad =  \E_{\bL}[|\Gamma(\bL,S\cup\{v\})\setminus \Gamma(\bL, S)|] \\
	& \quad = \E_{\bL} \left[\sum_{u\in V} \I\{u \in \Gamma(\bL,S\cup\{v\})\setminus \Gamma(\bL, S) \}\right] \\
	& \quad =  n \cdot \E_{\bL} \left[ \sum_{u\in V} \frac{1}{n} \cdot \I\{u \in \Gamma(\bL,S\cup\{v\})\setminus \Gamma(\bL, S) \}\right] \\
	& \quad =  n \cdot \E_{\bL}\left[ \E_{\bu\sim V} [\I\{\bu \in \Gamma(\bL,S\cup\{v\})\setminus \Gamma(\bL, S) \} ]   \right] \\
	& \quad = n \cdot \Pr_{\bL,\bu\sim V}\{\bu \in \Gamma(\bL,S\cup\{v\})\setminus \Gamma(\bL, S) \}.
	\end{align*}
	By a similar argument, 
	event $\bu \in \Gamma(\bL,S\cup\{v\})\setminus \Gamma(\bL, S) $ is the same as the event $v \in \bR \wedge S \cap \bR = \emptyset$.
	Hence we have $\sigma(S\cup \{v\}) - \sigma(S) = n \cdot \Pr(v\in \bR \wedge S\cap \bR = \emptyset)$.
\end{proof}

For a fixed subset $R\subseteq V$ and a node $v\in V$, define:
\[
X_R(v) = \left\{ \begin{array}{lr}
0 & \mbox{if $v\not\in R$;}\\
\frac{1}{|R|} & \mbox{if $v\in R$.}
\end{array}
\right.
\]
If $\bR$ is a random RR set, then $X_{\bR}(v)$ is a random variable.
The following is a restatement of Lemma~\ref{lem:spexpMainbody} using
the $X_{\bR}(v)$ random variable.
\begin{lemma}[Shapley Value Identity] \label{lem:spexp}
	Let $\bR$ be a random RR set.
	Then, for all $v\in V$, the Shapley centrality of $v$ is
	$\psi_v = n\cdot \E_{\bR}[X_{\bR}(v)]$.
\end{lemma}
\begin{proof}
	Let $\bR$ be a random RR set.
	We have
	\begin{align*}
	\psi_v = & \E_{\bpi}[\sigma(S_{\bpi,v} \cup \{v\}) - \sigma(S_{\bpi,v})] & \mbox{\{by Eq.~\eqref{eq:sprandom}\} }\\
	= & \E_{\bpi}[n \cdot \Pr(v\in \bR \wedge S_{\bpi,v}\cap \bR = \emptyset)] & \mbox{\{by Lemma~\ref{lem:margin}\} } \\
	= & n \cdot \E_{\bpi}[\E_{\bR} [\I\{v\in \bR \wedge S_{\bpi,v}\cap \bR = \emptyset\} ] ] \\
	= & n \cdot \E_{\bR}[\E_{\bpi} [\I\{v\in \bR \wedge S_{\bpi,v}\cap \bR = \emptyset\} ] ].
	\end{align*}
	By Proposition~\ref{lem:orderfirst}, for any realization of $\bR$:
	\[
	\E_{\bpi\sim \Pi(V)} [\I\{v\in \bR \wedge S_{\bpi,v}\cap \bR = \emptyset\} ] = 
	\left\{ 
	\begin{array}{lr}
	0 & \mbox{if $v\not\in \bR$, }\\
	\frac{1}{|\bR|} & \mbox{if $v\in \bR$.}
	\end{array}
	\right.
	\]
	This means that $\E_{\bpi\sim\Pi(V)} [\I\{v\in \bR \wedge S_{\bpi,v}\cap \bR = \emptyset\} ]
	= X_{\bR}(v)$.
	Therefore, $\psi_v = n\cdot \E_{\bR}[X_{\bR}(v)]$.
\end{proof}

After the above preparation, we are ready to show the unbiasedness of
our Shapley estimator.

\begin{lemma}[Unbiased Estimator] \label{lem:unbiassp}
	For any $v\in V$, the estimated value $\epsi_v$ returned by 
	Algorithm~\ref{alg:rrshapleynew} satisfies $\E[\epsi_v] = \psi_v$, where
	the expectation is taken over all randomness used in Algorithm {\ASVRR}.
\end{lemma}
%
%
%
\begin{proof}
	In Phase 2 of Algorithm {\ASVRR}, when $\btheta$ is fixed to 
	$\theta$,
	the algorithm generates $\theta$ independent random RR sets $\bR_1, \ldots, \bR_{\theta}$.
	Let $\est^\theta_v$ be the value of $\est_v$ at the end of the for-loop in Phase 2,
	when $\btheta = \theta$.
	It is straightforward to see that $\est^\theta_v = \sum_{i=1}^\theta X_{\bR_i}(v)$.
	Therefore,  by Lemma~\ref{lem:spexp}:
	$$\E[\epsi_v \mid \btheta = \theta] = \E[n\cdot \est^\theta_v / \theta]
	= \E[n\cdot \sum_{i=1}^\theta X_{\bR_i}(v) / \theta]
	= \psi_v.$$
	Since this is true for every fixed $\theta$, we have $\E[\epsi_v] = \psi_v$.
\end{proof}

\begin{lemma}[Absolute Normalization] \label{lem:absnorm}
	In every run of {\ASVRR}, we have $\sum_{v\in V} \epsi_v = n$.
\end{lemma}
\begin{proof}
	According to line~\ref{line:estnew2} of the algorithm, for every RR set $\bR$ generated in Phase 2,
	each node $\bu \in \bR$ increases its estimate $\est_{\bu}$	by $1/|\bR|$ and no other nodes
	increase their estimates.
	Thus the total increase in the estimates of all nodes for each $\bR$ is exactly $1$.
	Then after generating $\btheta$ RR sets, the sum of estimates is $\btheta$.
	According to line~\ref{line:adjustnew2}, we conclude that  $\sum_{v\in V} \epsi_v = n$.
\end{proof}

\subsection{Robustness of the Shapley Estimator of {\ASVRR}} \label{app:robustness}

The analysis on the robustness and time complexity is similar to that of {\IMM} in
\cite{tang15}, but since we are working on Shapley values while {\IMM} is for influence
maximization, there are also a number of differences.
In what follows, we provide an indepdent and complete proof for our algorithm, borrowing some
ideas from \cite{tang15}.

We will use the following basic Chernoff bounds~\cite{MU05,CL06} in our analysis.

\begin{fact}[Chernoff Bounds] \label{fact:chernoff}
	Let $\bY$ be the sum of $t$ i.i.d. random variables with mean $\mu$ and value range
	$[0,1]$. For any $\delta > 0$, we have:
	\begin{align*}
	\Pr\{\bY - t \mu \ge \delta \cdot t\mu \} &\le \exp\left( - \frac{\delta^2}{2+\frac{2}{3}\delta}t\mu \right).
	\end{align*}
	For any $0 < \delta < 1$, we have
	\begin{align*}
	\Pr\{\bY - t \mu \le - \delta \cdot t\mu \} &\le \exp\left( - \frac{\delta^2}{2}t\mu  \right).
	\end{align*}
\end{fact}

Let $\psi^{(k)}$ be the $k$-th largest value among all
shapley values in $\{\psi_v\}_{v\in V}$,
as defined in Theorem~\ref{thm:ASVRR}.
The following lemma provides a condition for robust Shapley value estimation.

\begin{lemma} \label{lem:suftheta}
	At the end of  Phase 2 of Algorithm {\ASVRR},
	\[
	\left\{ 
	\begin{array}{lr}
	|\epsi_v - \psi_v| \le \varepsilon \psi_v & \forall v\in V \mbox{ with } \psi_v > \psi^{(k)},\\
	|\epsi_v - \psi_v| \le \varepsilon \psi^{(k)} & \forall v\in V \mbox{ with } \psi_v \le \psi^{(k)}.
	\end{array}
	\right.
	\]
	holds with probability at least $1 - \frac{1}{2n^\ell}$, 
	provided that the realization  $\theta$ of $\btheta$ satisfies:
	\begin{equation} \label{eq:thetabound}
	\theta \ge  \frac{n ((\ell+1)\ln n + \ln 4)(2+\frac{2}{3}\varepsilon) }{\varepsilon^2 \psi^{(k)} }.
	\end{equation}
\end{lemma}
\begin{proof}
	Let $\bR_1, \bR_2, \ldots, \bR_{\theta}$ be the $\theta$ independent and
	random RR sets generated in Phase 2.
	Let $\est^\theta_v$ be the value of $\est_v$ at the end of the for-loop in Phase 2,
	when $\btheta = \theta$.
	Then,  $\est^\theta_v = \sum_{i=1}^\theta X_{\bR_i}(v)$, $\forall v\in V$.
	By Lemma~\ref{lem:spexp}, $\E[X_{\bR_i}(v)] = \psi_v / n$.
	
	For every $v\in V$ with $\psi_v > \psi^{(k)}$, 
	we apply the Chernoff bounds (Fact~\ref{fact:chernoff}) and have:
	\begin{align*}
	&\Pr\{|\epsi_v - \psi_v| \ge \varepsilon \psi_v \} \\
	& \quad = \Pr\{|n \cdot \est^\theta_v /\theta - \psi_v| \ge \varepsilon \psi_v \} \\
	& \quad = \Pr\{|\est^\theta_v  - \theta \cdot \psi_v / n| \ge  \varepsilon  \cdot (\theta\cdot \psi_v/n)\} \\
	& \quad \le  2 \exp\left( - \frac{\varepsilon^2}{2+ \frac{2}{3}\varepsilon}\cdot \theta\cdot \psi_v/n\right) \\
	& \quad \le  2 \exp\left( - \frac{\varepsilon^2 \cdot \psi_v}{(2+ \frac{2}{3}\varepsilon) \cdot n} 
	\cdot \frac{n ((\ell+1)\ln n + \ln 4)(2+\frac{2}{3}\varepsilon) }{\varepsilon^2 \psi^{(k)}} \right) \\
	& \quad \le 2 \exp\left( - (\ell+1)\ln n + \ln 4 \right) 
	\qquad \qquad \quad \mbox{\{since $\psi_v > \psi^{(k)} $\}} \\
	& \quad \le \frac{1}{2n^{\ell+1}}. 
	\end{align*}
	For every $v\in V$ with $\psi_v \le \psi^{(k)}$, 
	we also apply the Chernoff bound and have:
	\begin{align*}
	&\Pr\{|\epsi_v - \psi_v| \ge \varepsilon \psi^{(k)} \} \\
	& \quad = \Pr\{|n \cdot \est^\theta_v /\theta - \psi_v| \ge \varepsilon \psi^{(k)} \} \\
	& \quad = \Pr\{|\est^\theta_v  - \theta \cdot \psi_v / n| \ge  (\varepsilon \psi^{(k)}  / \psi_v) \cdot (\theta\cdot \psi_v/n)\} \\
	& \quad \le  2 \exp\left( - \frac{(\varepsilon \psi^{(k)}  / \psi_v)^2}{2+ \frac{2}{3}(\varepsilon \psi^{(k)}  / \psi_v)}\cdot \theta\cdot \psi_v/n\right) \\
	& \quad =  2 \exp\left( - \frac{\varepsilon^2 (\psi^{(k)} )^2}{n (2\psi_v+\frac{2}{3}\varepsilon \psi^{(k)} )}\cdot \theta\right) \\
	& \quad \le 2 \exp\left( - \frac{\varepsilon^2 \psi^{(k)} }{n (2+\frac{2}{3}\varepsilon)}\cdot \theta\right)
	\qquad \qquad \quad \mbox{\{since $\psi_v \le \psi^{(k)} $\}} \\
	& \quad \le \frac{1}{2n^{\ell+1}}. \qquad  \qquad \qquad\qquad \qquad \qquad \ \mbox{\{use Eq.~\eqref{eq:thetabound}\}}
	\end{align*}
	Finally, we take the union bound among all $n$ nodes in $V$ to obtain the result.
\end{proof}

For Phase 1, we need to show that with high probability $\LB \le \psi^{(k)}$, and thus
Eq.\eqref{eq:thetabound} hold for the random $\btheta$ set in line~\ref{line:thetanew}
of the algorithm.
The structure of the Phase 1 of {\ASVRR} follows the Sampling() algorithm 
in~\cite{tang15} (Algorithm 2, lines 1-13), with the difference that 
our Phase 1 is to estimate a lower bound for $\psi^{(k)}$, while
their purpose is to estimate a lower bound for $OPT_k$, the maximum influence spread
of any $k$ seed nodes.
The probabilistic analysis follows the same approach, and for completeness, we provide
an independent analysis for our algorithm.

Let $\btheta'$ be the number of RR sets generated in Phase 1, and 
$\bR^{(1)}_1, \bR^{(1)}_2, \ldots, \bR^{(1)}_{\btheta'}$ be these RR sets.
Note that these random RR sets are not mutually independent, because earlier generated RR sets
are used to determine if more RR sets need to be generated
(condition in line~\ref{line:cond1}).
However, once RR sets $\bR^{(1)}_1, \ldots, \bR^{(1)}_{i-1}$ are generated, the generation of
RR set $\bR^{(1)}_i$ follows the same random behavior for each $i$, which means we could
use martingale approach \cite{MU05} to analyze these RR sets and Phase 1 of Algorithm {\ASVRR}.

\begin{definition}[Martingale]
	A sequence of random variables $\bY_1, \bY_2, \bY_3, \ldots$ is a martingale, if and only
	if $\E[|\bY_i|] < +\infty$ and $\E[\bY_{i+1} \mid \bY_1, \bY_2, \ldots, \bY_{i}] = \bY_{i}$ for any
	$i \ge 1$.
\end{definition}

In our case, let $\bY_i(v) = \sum_{j=1}^i (X_{\bR^{(1)}_j}(v) - \psi_v / n)$, for any
$v\in V$ and any $i$.
\begin{lemma} \label{lem:martingale}
	For every $v\in V$ and every $i\ge 1$, $\E[X_{\bR^{(1)}_{i}}(v) \mid 
	X_{\bR^{(1)}_{1}}(v), \ldots, X_{\bR^{(1)}_{i-1}}(v)] = \psi_v / n$.
	As a consequence, for every $v\in V$, the sequence of random variables $\{\bY_i(v), i\ge 1\}$
	is a martingale.
\end{lemma}
\begin{proof}
	Consider a node $v$ and an index $i \ge 1$.
	Note that RR sets $\bR^{(1)}_1, \ldots, \bR^{(1)}_i$ determines whether $\bR^{(1)}_{i+1}$ should be generated,
	but the actual random generation process of $\bR^{(1)}_{i+1}$, i.e. selecting the random root 
	and the random live edge graph, is independent of $\bR^{(1)}_1, \ldots, \bR^{(1)}_i$.
	Therefore, by Lemma~\ref{lem:spexp} we have
	\begin{equation} \label{eq:martingaleconditional}
	\E\left[X_{\bR^{(1)}_{i}}(v) \left\lvert  
	X_{\bR^{(1)}_{1}}(v), \ldots, X_{\bR^{(1)}_{i-1}}(v) \right. \right] = \psi_v / n. 
	\end{equation}
	
	From the definition of $\bY_i(v)$, it is straightforward to see that the value range
	of $\bY_i(v)$ is $[-i, i]$, and thus $\E(|\bY_i(v)|] < + \infty$.
	Second, by definition $\bY_{i+1}(v) = X_{\bR^{(1)}_{i+1}}(v) - \psi_v / n + \bY_{i}(v) $.
	With the similar argument as for Eq.~\eqref{eq:martingaleconditional}, we have
	\begin{align*}
	&\E[\bY_{i+1}(v) \mid \bY_1(v), \ldots, \bY_i(v)] \\
	& =  \E[X_{\bR^{(1)}_{i+1}}(v) - \psi_v / n \mid \bY_1(v), \ldots, \bY_i(v)]  \\
	&  \quad + \E[\bY_{i}(v) \mid  \bY_1(v), \ldots, \bY_i(v)]  \\
	&  = 0 + \bY_{i}(v) = \bY_{i}(v). 
	\end{align*}
	Therefore, $\{\bY_i(v), i\ge 1\}$ is a martingale.
\end{proof}

Martingales have similar tail bounds as the Chernoff bound given in Fact~\ref{fact:chernoff},
as we give below.
For convenience, we did not explicitly refer to the sequence below as a martingale, but
notice that if we define $\bY_i = \sum_{j=1}^i (\bX_i - \mu)$, then
$\{\bY_1, \bY_2, \ldots, \bY_t \}$ is indeed a martingale.
\begin{fact}[Martingale Tail Bounds] \label{fact:martingalechernoff}
	Let $\bX_1, \bX_2, \ldots, \bX_t$ be random variables with range $[0,1]$, and
	for some $\mu \in [0,1]$, $\E[\bX_i \mid \bX_1, \bX_2, \ldots, \bX_{i-1}] = \mu$ for every
	$i\in [t]$.
	Let $\bY = \sum_{i=1}^t \bX_i$.
	For any $\delta > 0$, we have:
	\begin{align*}
	\Pr\{\bY - t \mu \ge \delta \cdot t\mu \} &\le \exp\left( - \frac{\delta^2}{2+\frac{2}{3}\delta}t\mu \right).
	\end{align*}
	For any $0< \delta < 1$, we have
	\begin{align*}
	\Pr\{\bY - t \mu \le - \delta \cdot t\mu \} &\le \exp\left( - \frac{\delta^2}{2}t\mu  \right).
	\end{align*}
\end{fact}
Since there are numerous variants of Chernoff and martingale tail bounds in the literature, 
and the ones we found in
\cite{MU05,CL06,tang15} are all slightly different from the above, in Appendix~\ref{sec:martingale}
we provide a pair of general martingale tail bounds that cover Facts~\ref{fact:chernoff} and~\ref{fact:martingalechernoff}
we need in this paper, with a complete proof.

For each $i = 1, 2, \ldots, \lfloor \log_2 n \rfloor -1$, 
let $x_i = n/2^i$, 
and let $\est^{(k)}_i$ be the value of $\est^{(k)}$ set in line~\ref{line:kmaxest} in the $i$-th
iteration of the for-loop (lines~\ref{line:phase1forb}--\ref{line:phase1fore})
of Phase 1. 
\begin{lemma} \label{lem:LBx}
	For each $i= 1, 2, \ldots, \lfloor \log_2 n \rfloor -1$,  
	\begin{itemize}
		\item[(1)]
		If $x_i = n/2^i > \psi^{(k)}$,  then with probability at least  $1 - \frac{1}{2n^\ell \log_2 n}$, 
		$n \cdot \est^{(k)}_i / \theta_i < (1+\varepsilon') \cdot x_i$.
		\item[(2)] 
		If $x_i = n/2^i \le \psi^{(k)}$,  then with probability at least  $1 - \frac{1}{2n^\ell \log_2 n}$, 
		$n \cdot \est^{(k)}_i / \theta_i < (1+\varepsilon') \cdot \psi^{(k)}$.
	\end{itemize}
\end{lemma}
\begin{proof}
	Let $\bR^{(1)}_1, \bR^{(1)}_2, \ldots, \bR^{(1)}_{\theta_i}$ be the $\theta_i$ generated
	RR sets by the end of the $i$-th iteration of the for-loop (lines~\ref{line:phase1forb}--\ref{line:phase1fore})
	of Phase 1. 
	For every $v\in V$, let $\est_{v,i}$ be the value of $\est_v$ at line~\ref{line:kmaxest} in the $i$-th
	iteration of the same for-loop. 
	Then we have $\est_{v,i} = \sum_{j=1}^{\theta_i} X_{\bR^{(1)}_j}(v)$.
	
	By Lemma~\ref{lem:martingale}, we have for every $1\le j \le \theta_i$, 
	$	\E[X_{\bR^{(1)}_{i}}(v) \mid  
	X_{\bR^{(1)}_{1}}(v), \ldots, X_{\bR^{(1)}_{i-1}}(v) ] = \psi_v / n$. 
	Then we can apply the martingale tail bound of Fact~\ref{fact:martingalechernoff} on the sequence.
	For the Statement (1) of the lemma, we consider $x_i = n/2^i > \psi^{(k)}$, and
	for every $v\in V$ such that $\psi_v \le \psi^{(k)}$, we obtain
	\begin{align}
	&\Pr\{n \cdot \est_{v,i} / \theta_i \ge (1+\varepsilon') \cdot x_i\}  \nonumber \\
	& = \Pr\{ \est_{v,i} \ge (1+\varepsilon') \cdot \theta_i \cdot x_i  / n\} \nonumber \\
	& \le \Pr\{ \est_{v,i} - \theta_i \cdot \psi_v / n\ge (\varepsilon'\cdot x_i / \psi_v ) \cdot \theta_i \cdot \psi_v / n\} \nonumber \\
	& \le \exp\left( - \frac{(\varepsilon'\cdot x_i / \psi_v )^2}{2+\frac{2}{3}(\varepsilon'\cdot x_i / \psi_v )}\cdot \theta_i \cdot \psi_v / n \right) \nonumber \\
	& = \exp\left( - \frac{\varepsilon'^2 \cdot x_i^2}{2\psi_v+\frac{2}{3}\cdot \varepsilon'\cdot x_i}\cdot \theta_i / n \right) \nonumber \\
	& \le \exp\left( - \frac{\varepsilon'^2 \cdot x_i}{2+\frac{2}{3}\varepsilon'}\cdot \theta_i / n \right) 
	\qquad \mbox{\{use $\psi_v \le \psi^{(k)} < x_i$\}}\nonumber \\
	& \le \frac{1}{2 n^{\ell+1} \log_2 n}. \label{eq:estvi1}
	\end{align}
	Note that $\est^{(k)}_{i}$ is the $k$-th largest among $\est_{v,i}$'s, while
	there are at most $k-1$ nodes $v$ with $\psi_v > \psi^{(k)}$.
	This means that there is at least one node $v$ with $\psi_v \le \psi^{(k)}$ and 
	$\est_{v,i} \ge \est^{(k)}_{i}$.
	Thus, by taking union bound on Eq.~\eqref{eq:estvi1}, we have
	\begin{align*}
	& \Pr\{n \cdot \est^{(k)}_{i} / \theta_i \ge (1+\varepsilon') \cdot x_i\}  \\
	& \le \Pr\{\exists v\in V, \psi_v \le \psi^{(k)},  n \cdot \est_{v,i} / \theta_i \ge (1+\varepsilon') \cdot x_i\}
	\\
	& \le \frac{1}{2 n^{\ell} \log_2 n}. 
	\end{align*}
	
	For the Statement (2) of the lemma, we consider $x_i = n/2^i \le \psi^{(k)}$, and 
	for every $v\in V$ with $\psi_v \le \psi^{(k)}$, we have
	\begin{align*}
	&\Pr\{n \cdot \est_{v,i} / \theta_i \ge (1+\varepsilon') \cdot \psi^{(k)}\} \\
	& = \Pr\{ \est_{v,i} \ge (1+\varepsilon') \cdot \theta_i \cdot \psi^{(k)}  / n\} \\
	& \le \Pr\{ \est_{v,i} - \theta_i \cdot \psi_v / n\ge (\varepsilon'\cdot \psi^{(k)} / \psi_v ) \cdot \theta_i \cdot \psi_v / n\} \\
	& \le \exp\left( - \frac{(\varepsilon'\cdot \psi^{(k)} / \psi_v )^2}{2+\frac{2}{3}(\varepsilon'\cdot \psi^{(k)} / \psi_v )}\cdot \theta_i \cdot \psi_v / n \right) \\
	& = \exp\left( - \frac{\varepsilon'^2 \cdot {\psi^{(k)}}^2}{2\psi_v+\frac{2}{3}\cdot \varepsilon'\cdot \psi^{(k)}}\cdot \theta_i / n \right) \\
	& \le \exp\left( - \frac{\varepsilon'^2 \cdot \psi^{(k)}}{2+\frac{2}{3}\varepsilon'}\cdot \theta_i / n \right) 
	\qquad \mbox{\{use $\psi_v \le \psi^{(k)}$\}}\\
	& \le \exp\left( - \frac{\varepsilon'^2 \cdot x_i}{2+\frac{2}{3}\varepsilon'}\cdot \theta_i / n \right) 
	\qquad \mbox{\{use $x_i \le \psi^{(k)}$\}}\\
	& \le \frac{1}{2 n^{\ell+1} \log_2 n}.
	\end{align*}
	Similarly, by taking the union bound, we have
	\begin{align*}
	& \Pr\{n \cdot \est^{(k)}_{i} / \theta_i \ge (1+\varepsilon') \cdot \psi^{(k)} \}  \\
	& \le \Pr\{\exists v\in V, \psi_v \le \psi^{(k)},  n \cdot \est_{v,i} / \theta_i 
	\ge (1+\varepsilon') \cdot \psi^{(k)}\}
	\\
	& \le \frac{1}{2 n^{\ell} \log_2 n}. 
	\end{align*}
	Thus the lemma holds.
\end{proof}

\begin{lemma} \label{lem:LB}
	Suppose that $\psi^{(k)} \ge 1$.
	In the end of Phase 1, with probability at least $1-\frac{1}{2n^\ell}$, 
	$\LB \le \psi^{(k)}$.
\end{lemma}
\begin{proof}
	Let $\LB_i = n\cdot \est^{(k)}_i / (\theta_i \cdot (1+\varepsilon'))$.
	Suppose first that $\psi^{(k)} \ge x_{\lfloor \log_2 n \rfloor - 1} $, and 
	let $i$ be the smallest index such that $\psi^{(k)} \ge x_i$.
	Thus, for each $i' \le i-1$, $\psi^{(k)} < x_{i'}$.
	By Lemma~\ref{lem:LBx} (1), for each $i' \le i-1$, with probability at most $\frac{1}{2n^\ell \log_2 n}$,
	$n \cdot \est^{(k)}_{i'} / \theta_{i'} \ge (1+\varepsilon') \cdot x_{i'}$.
	Taking union bound, we know that with probability at least $1-\frac{i - 1}{2n^\ell \log_2 n}$,
	for all $i' \le i-1$, $n \cdot \est^{(k)}_{i'} / \theta_{i'} < (1+\varepsilon') \cdot x_{i'}$.
	This means, with probability at least $1-\frac{i - 1}{2n^\ell \log_2 n}$,
	that the for-loop in Phase 1 would not break at the $i'$-th iteration for $i'\le i-1$,
	and thus $\LB = \LB_{i''}$ for some $i''\ge i$, or $\LB = 1$.
	Since for every $i'' \ge i$, we have $x_{i''} \le \psi^{(k)}$, by Lemma~\ref{lem:LBx} (2), 
	for every such $i''$, with probability at most $\frac{1}{2n^\ell \log_2 n}$, 
	$\LB_{i''} > \psi^{(k)}$.
	Taking union bound again, we know that with probability at most  $\frac{1}{2n^\ell}$,
	$\LB > \psi^{(k)}$.
	
	Finally, if $\psi^{(k)} < x_{\lfloor \log_2 n \rfloor - 1} $, use the similar argument as the above,
	we can show that, with probability at least $1-\frac{1}{2n^\ell}$,
	the for-loop would not break at any iteration, which means $\LB = 1$, which still
	implies that $\LB \le \psi^{(k)}$ since $\psi^{(k)} \ge 1$.
\end{proof}

\begin{lemma}[Robust Estimator] \label{lem:errorbound}
	Suppose that $\psi^{(k)} \ge 1$.
	With probability at least $1-\frac{1}{n^\ell}$, Algorithm {\ASVRR} returns 
	$\{\epsi_v\}_{v \in V}$ that satisfy:
	\[
	\left\{ 
	\begin{array}{lr}
	|\epsi_v - \psi_v| \le \varepsilon \psi_v & \forall v\in V \mbox{ with } \psi_v > \psi^{(k)},\\
	|\epsi_v - \psi_v| \le \varepsilon \psi^{(k)} & \forall v\in V \mbox{ with } \psi_v \le \psi^{(k)}.
	\end{array}
	\right.
	\]
\end{lemma}
\begin{proof}
	By Lemma~\ref{lem:LB}, we know that at the end of Phase 1, with
	probability at least $1 - \frac{1}{2n^\ell}$, $\LB \le \psi^{(k)}$.
	
	Then by Lemma~\ref{lem:suftheta}, we know that when we fix $\LB$
	to any fixed value $LB$ with $LB \le \psi^{(k)}$, 
	with probability at least $1-\frac{1}{2n^\ell}$, 
	we have 
	\[
	\left\{ 
	\begin{array}{lr}
	|\epsi_v - \psi_v| \le \varepsilon \psi_v & \forall v\in V \mbox{ with } \psi_v > \psi^{(k)},\\
	|\epsi_v - \psi_v| \le \varepsilon \psi^{(k)} & \forall v\in V \mbox{ with } \psi_v \le \psi^{(k)}.
	\end{array}
	\right.
	\]
	Taking the union bound, we know that with probability at least $1-\frac{1}{n^\ell}$, 
	we have
	\[
	\left\{ 
	\begin{array}{lr}
	|\epsi_v - \psi_v| \le \varepsilon \psi_v & \forall v\in V \mbox{ with } \psi_v > \psi^{(k)},\\
	|\epsi_v - \psi_v| \le \varepsilon \psi^{(k)} & \forall v\in V \mbox{ with } \psi_v \le \psi^{(k)}.
	\end{array}
	\right.
	\]
\end{proof}

\subsection{Time Complexity of {\ASVRR}} \label{app:time}

Finally, we argue about the time complexity of the algorithm.
For this purpose, we need to refer to the martingale stopping theorem, explained below.

A random variable $\btau$ is a {\em stopping time} for martingale $\{\bY_i, i\ge 1\}$
if $\btau$ takes positive integer values, and the event 
$\btau = i$ depends only on the values of $\bY_1, \bY_2, \ldots, \bY_i$.
The following martingale stopping theorem is an important fact for our analysis.

\begin{fact}[Martingale Stopping Theorem \cite{MU05}] \label{fact:stopping}
	Suppose that $\{\bY_i, i\ge 1\}$ is a martingale and $\btau$ is a stopping time for $\{\bY_i, i\ge 1\}$.
	If $\btau \le c$ for some constant $c$ independent of $\{\bY_i, i\ge 1\}$, then
	$\E[\bY_{\btau}] = \E[\bY_1]$.\footnote{There are two other alternative conditions
		besides that $\tau$ is bounded by a constant, but they are not needed in our analysis
		and thus are omitted.}
\end{fact}

Given a fixed set $R\subseteq V$, let the {\em width} of $R$, denoted
$\omega(R)$,
be the total in-degrees of nodes in $R$.
By Assumption \ref{assump:ComputationalTriggerModel},
the time complexity to generate the random RR set $\bR$ 
is  $\Theta(\omega(\bR)+1)$.
We leave the constant $1$ in the above formula because $\omega(\bR)$ could be less than $1$
or even $o(1)$ when $m < n$, while $\Theta(1)$ time is needed just to select a random root.
The expected time complexity to generate a
random RR set is $\Theta(\E[\omega(\bR)]+1)$. 

Let $EPT = \E[\omega(\bR)]$ be the expected width of a random RR set.
Let $\btheta'$ be the random variable denoting the number of RR sets
generated in Phase 1.
\begin{lemma} \label{lem:timefrombtheta}
	Under Assumption \ref{assump:ComputationalTriggerModel},
	the expected running time of {\ASVRR} is
	$\Theta((\E[\btheta']+\E[\btheta])\cdot (EPT+1))$.
\end{lemma}
\begin{proof}
	Let $\bR^{(1)}_1, \bR^{(1)}_2, \ldots, \bR^{(1)}_{\btheta'}$ be the
	RR sets generated in Phase 1.
	Under Assumption~\ref{assump:ComputationalTriggerModel}, 
	for each RR set $\bR^{(1)}_j$, the time to generate $\bR^{(1)}_j$ is
	$\Theta(\omega(\bR^{(1)}_j) + 1)$, where the constant $\Theta(1)$ term is to accommodate
	the time just to select a random root node for the RR set, and it is not absorbed by
	$ \Theta(\omega(\bR^{(1)}_j)$ because the width of an RR set could be less than $1$.
	After generating $\bR^{(1)}_j$, {\ASVRR} also needs to go through all
	entries $\bu\in \bR^{(1)}_j$ to update $\est_{\bu}$ (line~\ref{line:estnew1}),
	which takes $\Theta(|\bR^{(1)}_j|)$ time.
	Note that for every random RR set $\bR$, we have  $|\bR| \le \omega(\bR) + 1$, 
	because the RR set generation
	process guarantees that the induced sub-graph of any RR set must
	be weakly connected.
	Thus, for each RR set $\bR^{(1)}_j$, {\ASVRR} takes 
	$\Theta(\omega(\bR^{(1)}_j) + 1 + |\bR^{(1)}_j|) = \Theta(\omega(\bR^{(1)}_j) + 1)$
	time, and summing up for all $\btheta'$ RR sets,
	the total running time of Phase 1 is
	$\Theta(\sum_{j=1}^{\btheta'} (\omega(\bR^{(1)}_j)+1))$.
	
	We define $\bW_i = \sum_{j=1}^i (\omega(\bR^{(1)}_j) - EPT)$, for $i\ge 1$.
	By an argument similar to Lemma~\ref{lem:martingale}, 
	we know that $\{\bW_i, i\ge 1\}$
	is a martingale.
	Moreover, $\btheta'$ is a stopping time of $\{\bW_i, i\ge 1\}$ because its value
	is only determined by the RR sets already generated.
	The value of $\btheta'$ is upper bounded by 
	$\theta_{\lfloor \log_2 n \rfloor -1 }$, which is a constant set
	in line~\ref{line:setthetai}.
	Therefore, we can apply the martingale stopping theorem (Fact~\ref{fact:stopping})
	and obtain
	\begin{align*}
	0 & = \E[\bW_1] = \E[\bW_{\btheta'}] = 
	\E\left[\sum_{j=1}^{\btheta'} \omega(\bR^{(1)}_j) - \btheta' \cdot EPT \right]\\
	& = \E\left[\sum_{j=1}^{\btheta'} \omega(\bR^{(1)}_j)\right] - \E[\btheta']\cdot EPT.
	\end{align*}
	This implies that the expected running time of Phase 1 is 
	$\Theta(\E[\btheta']\cdot (EPT+1))$.
	
	For Phase 2, all $\btheta$ RR sets are independently generated, 
	and thus the expected running time of Phase 2 is
	$\Theta(\E[\btheta]\cdot (EPT+1))$.
	Together, we know that the expected running time of {\ASVRR} is
	$\Theta((\E[\btheta']+\E[\btheta])\cdot (EPT+1))$.
\end{proof}

We now connecting $EPT$ with the influence spread of a single node,
first established in \cite{tang14} (Lemma 7).
For completeness, we include a proof here.
\begin{lemma}[Expected Width of Random RR Sets] \label{lem:EPT}
	Let $\tilde{\bv}$ be a random node drawn from $V$ with probability proportional to the 
	in-degree of $\tilde{\bv}$.
	Let $\bR$ be a random RR set.
	Then: $$EPT = \E_{\bR}[\omega(\bR)] = \frac{m}{n} \E_{\tilde{\bv}}[\sigma(\{\tilde{\bv} \})].$$
\end{lemma}
\begin{proof}
	For a fixed set $R\subseteq V$, let $p(R)$ be the probability that
	a randomly selected edge (from $E$) points to a node in $R$.
	Since $R$ has $\omega(R)$ edges pointing to nodes in $R$, we have
	$p(R) = \omega(R)/m$.
	
	Let $d_v$ denotes the in-degree of node $v$.
	Let $\tilde{\bv}$ be a random node drawn from $v$ with probability proportional to the
	in-degree of $\tilde{\bv}$.
	We have:
	\begin{align*}
	p(R) = & \sum_{(u,v)\in E} \frac{1}{m} \cdot \I\{v \in R\}  \\
	= & \sum_{ v \in V} \frac{d_v}{m} \cdot \I\{v \in R\} = \E_{\tilde{\bv}}[\I\{\tilde{\bv} \in R\}].
	\end{align*}
	Let $\bR$ be a random RR set. Then, we have:
	\begin{align*}
	\E_{\bR}[\omega(\bR)] = & m \cdot \E_{\bR}[p(\bR)] \\
	= & m \cdot \E_{\bR}[\E_{\tilde{\bv}}[\I\{\tilde{\bv} \in \bR\}]] \\
	= & m \cdot \E_{\tilde{\bv}}[\E_{\bR}[\I\{\tilde{\bv} \in \bR\}]] \\
	= & m \cdot \E_{\tilde{\bv}}[\Pr_{\bR}(\tilde{\bv} \in \bR)] \\
	= & m \cdot \E_{\tilde{\bv}}[\sigma(\{\tilde{\bv}\}) / n],
	\end{align*}
	where the last equality is by Lemma~\ref{lem:margin}.
\end{proof}

Next we need to bound $\E[\btheta']$ and $\E[\btheta]$.
\begin{lemma} \label{lem:stopearly}
	For each $i = 1, 2, \ldots, \lfloor \log_2 n \rfloor -1$, 
	if $\psi^{(k)} \ge (1+\varepsilon')^2 \cdot x_i$, then
	with probability at least $1 - \frac{k}{2n^{\ell+1}\log_2 n}$, 
	$n \cdot \est^{(k)}_i / \theta_i > \psi^{(k)} / (1+\varepsilon')$, and
	$n \cdot \est^{(k)}_i / \theta_i > (1+\varepsilon') \cdot x_i$.
\end{lemma}
\begin{proof}
	Let $\bR^{(1)}_1, \bR^{(1)}_2, \ldots, \bR^{(1)}_{\theta_i}$ be the $\theta_i$ generated
	RR sets by the end of the $i$-th iteration of the for-loop (lines~\ref{line:phase1forb}--\ref{line:phase1fore})
	of Phase 1. 
	For every $v\in V$, let $\est_{v,i}$ be the value of $\est_v$ at line~\ref{line:kmaxest} in the $i$-th
	iteration of the same for-loop. 
	Then we have $\est_{v,i} = \sum_{j=1}^{\theta_i} X_{\bR^{(1)}_j}(v)$.

	
	Suppose that $\psi^{(k)} \ge (1+\varepsilon')^2 \cdot x_i$.
	For every node $v\in V$ with $\psi_v \ge \psi^{(k)}$, 
	we apply the lower tail of the martingale tail bound (Fact~\ref{fact:martingalechernoff}) and obtain
	\begin{align*}
	& \Pr \{n \cdot \est_{v,i} / \theta_i \le \psi_v / (1+\varepsilon')   \} \\
	& \le \Pr\{\est_{v,i}  \le \frac{1}{1+\varepsilon'} \cdot \theta_i \cdot \psi_v/ n\} \\
	& = \Pr\{\est_{v,i} - \theta_i\cdot \psi_{v} / n \le 
	- \frac{\varepsilon'}{1+\varepsilon'}  \cdot  \theta_i \cdot\psi_{v}/ n\} \\
	& \le \exp \left(-\frac{\varepsilon'^2}{2(1+\varepsilon')^2 } 
	\cdot \theta_i \cdot\psi_{v}/ n \right) \\
	& \le \exp \left(-\frac{\varepsilon'^2 \cdot x_i}{2n} \cdot \theta_i\right)
	\quad \qquad  \mbox{\{use $x_i \le \frac{\psi^{(k)}}{(1+\varepsilon')^2} \le
		\frac{\psi_v}{(1+\varepsilon')^2}$\} } \\
	& \le \frac{1}{2 n^{\ell+1} \log_2 n}.
	\end{align*}
	
	Note that $\est_i^{(k)}$ is the $k$-th largest value among $\{\est_{v,i}\}_{v\in V}$, or
	equivalently $(n-k+1)$-th smallest among $\{\est_{v,i}\}_{v\in V}$.
	But there are at most $n-k$ nodes $v$ with $\psi_v < \psi^{(k)}$, which means that
	there is at least one node $v$ with $\psi_v \ge \psi^{(k)}$ and
	$\est_{v,i} \le \est_i^{(k)}$.
	To be precise, such a $v$ has $\psi_v$ ranked before or the same as $\psi^{(k)}$, and thus there
	are at most $k$ such nodes.
	Then we have
	\begin{align*}
	& \Pr\{n \cdot \est^{(k)}_{i} / \theta_i \le \psi^{(k)} / (1+\varepsilon') \} \\
	& \le \Pr \{\exists v\in V, \psi_v \ge \psi^{(k)}, 
	n \cdot \est_{v,i} / \theta_i \le \psi_v / (1+\varepsilon')   \} \\
	& \le k \Pr \{n \cdot \est_{v,i} / \theta_i \le \psi_v / (1+\varepsilon')   \} \\
	& \le \frac{k}{2 n^{\ell+1}\log_2 n}.
	\end{align*}
	Since $\psi^{(k)} \ge (1+\varepsilon')^2 \cdot x_i$, we have that
	$n \cdot \est^{(k)}_i / \theta_i > \psi^{(k)} / (1+\varepsilon')$ implies that
	$n \cdot \est^{(k)}_i / \theta_i > (1+\varepsilon') \cdot x_i$.
\end{proof}

\begin{lemma} \label{lem:expectedthetas}
	For both $\btheta'$ and $\btheta$, we have
	$\E[\btheta'] = O(\ell n \log n / (\psi^{(k)} \varepsilon^2))$ and
	$\E[\btheta'] = O(\ell n \log n / (\psi^{(k)} \varepsilon^2))$, 
	when $\ell \ge (\log_2 k - \log_2 \log_2 n)/\log_2 n$.
\end{lemma}
\begin{proof}
	If $\psi^{(k)} < (1+\varepsilon')^2 \cdot x_{\lfloor \log_2 n \rfloor -1}$,
	then 	$\psi^{(k)} < 4(1+\varepsilon')^2$.
	In this case, in the worst case,
	\begin{align*}
	& \btheta' = \theta_{\lfloor \log_2 n \rfloor -1} \\
	& \le  \left\lceil \frac{ n \cdot 
		((\ell + 1)\ln n + \ln \log_2 n + \ln 2) \cdot (2+\frac{2}{3}\varepsilon')}
	{\varepsilon'^2 } \right \rceil \nonumber \\
	& \le \left\lceil \frac{ n \cdot 
		((\ell + 1)\ln n + \ln \log_2 n + \ln 2) \cdot (2+\frac{2}{3}\varepsilon') 
		\cdot 4(1+\varepsilon')^2}
	{\varepsilon'^2 \cdot \psi^{(k)}} \right \rceil \nonumber \\
	& = O(\ell n \log n / (\psi^{(k)} \varepsilon^2)), 
	\end{align*}
	where the last equality uses the fact that $\varepsilon' = \sqrt{2} \cdot \varepsilon$, 
	and the big O notation is for sufficiently small $\varepsilon$.
	Similarly, for $\btheta$, since $\LB \ge 1$, we have
	\begin{align*}
	& \btheta \le \left\lceil \frac{n ((\ell+1)\ln n + \ln 4)(2+ \frac{2}{3} \varepsilon) }{\varepsilon^2} \right\rceil \\
	& \le \left\lceil \frac{n ((\ell+1)\ln n + \ln 4)(2+ 
		\frac{2}{3} \varepsilon) \cdot 4(1+\varepsilon')^2}{\varepsilon^2\cdot \psi^{(k)}} \right\rceil\\
	& = O(\ell n \log n / (\psi^{(k)} \varepsilon^2)). 
	\end{align*}
	Therefore, the lemma holds when $\psi^{(k)} < (1+\varepsilon')^2 \cdot x_{\lfloor \log_2 n \rfloor -1}$.
	
	Now suppose that $\psi^{(k)} \ge (1+\varepsilon')^2 \cdot x_{\lfloor \log_2 n \rfloor -1}$.
	Let $i$ be the smallest index such that 
	$\psi^{(k)} \ge (1+\varepsilon')^2 \cdot x_{i}$.
	Thus $\psi^{(k)} < (1+\varepsilon')^2 \cdot x_{i-1}
	= 2(1+\varepsilon')^2 x_i$ (denote $x_0 = n$).
	
	For $\E[\btheta']$, 
	by Lemma~\ref{lem:stopearly}, with probability at least $1 - \frac{k}{2n^{\ell+1}\log_2 n}$,
	$n \cdot \est^{(k)}_i / \theta_i > (1+\varepsilon') \cdot x_i$, which means that
	Phase 1 would stop in the $i$-th iteration, and thus 
	\begin{align}
	& \btheta' = \theta_i = \left\lceil \frac{ n \cdot 
		((\ell + 1)\ln n + \ln \log_2 n + \ln 2) \cdot (2+\frac{2}{3}\varepsilon')}
	{\varepsilon'^2 \cdot x_i} \right \rceil \nonumber \\
	& \le \left\lceil \frac{ n \cdot 
		((\ell + 1)\ln n + \ln \log_2 n + \ln 2) \cdot (2+\frac{2}{3}\varepsilon') 
		\cdot 2(1+\varepsilon')^2}
	{\varepsilon'^2 \cdot \psi^{(k)}} \right \rceil \nonumber \\
	& = O(\ell n \log n / (\psi^{(k)} \varepsilon^2)), \label{eq:thetaprimebound}
	\end{align}
	where the last equality uses the fact that $\varepsilon' = \sqrt{2} \cdot \varepsilon$.
	
	When Phase 1 stops at the $i$-th iteration, we know that 
	$\LB = n\cdot \est^{(k)}_i / (\theta_i\cdot (1+\varepsilon')) \ge \psi^{(k)} / (1+\varepsilon')^2$,
	again by Lemma~\ref{lem:stopearly}.
	Then, for Phase 2 we have
	\begin{align}
	& \btheta \le \left\lceil \frac{n ((\ell+1)\ln n + \ln 4)(2+ 
		\frac{2}{3} \varepsilon) \cdot (1+\varepsilon')^2}{\varepsilon^2\cdot \psi^{(k)}} \right\rceil
	\nonumber \\
	& = O(\ell n \log n / (\psi^{(k)} \varepsilon^2)). \label{eq:thetaasympbound}
	\end{align}

	With probability at most $\frac{k}{2n^{\ell+1}\log_2 n}$, Phase 1 does not stop at
	the $i$-th iteration and continues to iterations $i' > i$.
	In the worst case, it continues to iteration $\lfloor \log_2 n \rfloor -1$,
	and $\btheta' = O(\ell n \log n / \varepsilon^2)$.
	Combining with Eq.~\eqref{eq:thetaprimebound}, we have
	\begin{align*}
	\E[\btheta'] & = O(\ell n \log n / (\psi^{(k)} \varepsilon^2)) + 
	\frac{k}{2n^{\ell+1}\log_2 n} \cdot O(\ell n \log n / \varepsilon^2) \\
	& = O(\ell n \log n / (\psi^{(k)} \varepsilon^2)),
	\end{align*}
	where the last equality uses the fact that $\psi^{(k)} \le n$, and
	the condition that $\ell \ge (\log_2 k - \log_2 \log_2 n)/\log_2 n$.
	Similarly, for Phase 2, in the worst case $\LB = 1$ and we have 
	$\btheta = O(\ell n \log n / \varepsilon^2)$.
	Combining with Eq.~\eqref{eq:thetaasympbound}, we have
	\begin{align*}
	\E[\btheta] & = O(\ell n \log n / (\psi^{(k)} \varepsilon^2)) + 
	\frac{1}{2n^{\ell+1}\log_2 n} \cdot O(\ell n \log n / \varepsilon^2) \\
	& = O(\ell n \log n / (\psi^{(k)} \varepsilon^2)).
	\end{align*}
	This concludes the lemma.
\end{proof}

We remark that, the setting of $\varepsilon' = \sqrt{2}\cdot \varepsilon$ is to
balance the complicated terms appearing in the upper bound
of $\E[\btheta'] + \E[\btheta]$, 
as suggested in~\cite{tang15}.

\begin{lemma}[Shapley Value Estimators: Scalability]  \label{lem:time}
	Under 
	Assumption \ref{assump:ComputationalTriggerModel}
	and the condition $\ell \ge (\log_2 k - \log_2 \log_2 n)/\log_2 n$, the expected running time
	of {\ASVRR} is 
	$
	O(\ell n \log n \cdot (EPT+1) / (\psi^{(k)} \varepsilon^2)) = 
	O(\ell (m+n) \log n \cdot \E[\sigma(\tilde{\bv})]/ (\psi^{(k)} \varepsilon^2)),
	$
	where $EPT=\E[\omega(\bR)]$ is the expected width of a random RR set,
	and $\tilde{\bv}$ is a random node drawn from $V$ with probability proportional to the 
	in-degree of $\tilde{\bv}$.
\end{lemma}
\begin{proof}
	The result is immediate from Lemmas~\ref{lem:timefrombtheta}, \ref{lem:EPT},
	and~\ref{lem:expectedthetas}.
\end{proof}

Together, Lemmas~\ref{lem:unbiassp}, \ref{lem:absnorm}, \ref{lem:errorbound}
and \ref{lem:time} establish Theorem~\ref{thm:ASVRR}. 

\subsection{On Theorem~\ref{thm:SNI}}

The proof of Theorem~\ref{thm:SNI} follows the same proof structure as the proof of Theorem~\ref{thm:ASVRR}.
The only difference is to replace the definition $X_R(v)$ with the following $X'_R(v)$:
\[
X'_R(v) = \left\{ \begin{array}{lr}
0 & \mbox{if $v\not\in R$; }\\
1 & \mbox{if $v\in R$.}
\end{array}
\right.
\]
Then, by Eq.~\eqref{eq:rrsetinf} in Lemma~\ref{lem:margin}, we know that the SNI centrality of node $v$ is:
$\psi^{\SNI}_v(\calI) = \sigma_{\calI}(\{v\}) = n \cdot  \E_{\bR}[X_{\bR}(v)]$.
This replaces the corresponding Lemma~\ref{lem:spexp} for Shapley centrality.
In the rest of the proof of Theorem~\ref{thm:ASVRR}, we replace every occurrence of $X_R(v)$ with
$X'_R(v)$ and keep in mind that now $\psi$ refers to the SNI centrality, then all the analysis went through
without change for SNI centrality, and thus Theorem~\ref{thm:SNI} holds
(Lemma~\ref{lem:absnorm} for absolute normalization does not apply to SNI centrality and no other
result depends on this lemma).

\subsection{Adaptation to Near-Linear-Time Algorithm} \label{app:nearlt}

We remark that, if we want a near-linear time algorithm for Shapley or SNI centrality
with some relaxation in robustness,
we can make an easy change to Algorithm {\ASVRR} or {\ASNIRR} as follows.
In line~\ref{line:estnew1}, replace
$\est_{\bu} = \est_{\bu} + 1/|\bR|$ with $\est_{\bu} = \est_{\bu} + 1$, and use parameter
$k=1$.
What this does is to estimate a lower bound $\LB$ of the largest single node influence
$\sigma^*_1 = \max_{v\in V} \sigma(\{v\})$.
Note that for line~\ref{line:estnew2}, it is still the case that,
for {\ASVRR} we use $\est_{\bu} = \est_{\bu} + 1/|\bR|$ while for {\ASNIRR} we use
$\est_{\bu} = \est_{\bu} + 1$.
Then we have this alternative result:
\begin{theorem} \label{thm:ASVRRsigma}
	If we use $k=1$ and replace $\est_{\bu} = \est_{\bu} + 1/|\bR|$ in line~\ref{line:estnew1} 
	of Algorithm~\ref{alg:rrshapleynew} with
	$\est_{\bu} = \est_{\bu} + 1$, while keeping the rest the same for {\ASVRR} and 
	{\ASNIRR} respectively, then the revised algorithm
	guarantees that with probability at least $1-\frac{1}{n^\ell}$;
	\[
	|\epsi_v - \psi_v| \le \varepsilon \sigma^*_1, \forall v\in V,
	\]
	with expected running time $O(\ell (m+n) \log n / \varepsilon^2)$.
	Note that $\psi_v$ above represents Shapley centrality of $v$ for {\ASVRR} and SNI centrality of $v$
	for {\ASNIRR}, and $\epsi_v$ is the algorithm output for the corresponding estimated centrality of $v$.
\end{theorem}

The proof of Theorem~\ref{thm:ASVRRsigma} would follow exactly the same structure as
the proof of Theorem~\ref{thm:ASVRR}.
To complete the proof, one only needs to observe that
by Lemma~\ref{lem:margin}, $\sigma(\{u\}) = n\cdot \E[\I\{u\in \bR\}]$ with a random
RR set $\bR$, and thus after changing to $\est_u = \est_u + 1$ in line~\ref{line:estnew1},
$n \cdot \est_u/\theta_i$ provides an estimate of $\sigma(\{u\})$ at the end of
the $i$-th iteration of Phase 1.
This means that $\LB$ obtained in Phase 1 is an estimate of the lower bound of the largest
single node influence $\sigma^*_1$.
Thus, essentially we only need to replace $\psi^{(k)}$ with $\sigma^*_1$ everywhere in the proof and
the theorem statement.
Finally, because $\E[\sigma(\tilde{\bv})] \le \sigma^*_1$, 
the time complexity no longer has the extra ratio term $\E[\sigma(\tilde{\bv})] / \sigma^*_1$.
The detailed proof is thus omitted.

\section{Extension to Weighted Influence Models} \label{app:weightedinf}

In this section, we extend 
 our results to models with weighted influence-spread functions.
We focus on the extension of Shapley centrality, and results on SNI centrality can be similarly derived.
The extended model uses node weights to capture 
  the practical ``nodes are not equal when activated'' 
  in network influence.
Let $w: V \rightarrow \R$ be a 
	non-negative weight function over $V$,
	 i.e., $w(v) \ge 0$, $\forall v\in V$. 
For any subset $S \subseteq V$, let $w(S) = \sum_{v\in S} w(v)$.
We can extend the cardinality-based influence spread $\sigma(S)$ to
  {\em weighted influence spread}: $\sigma^w(S) = \E[w(\bI(S))]$.
Here, the influence spread is weighted based on the value of activated nodes
in $\bI(S)$.
Note that, in the equivalent live-edge graph model for the triggering model, 
  we have:
$\sigma^w(S) = \E_{\bL}[w(\Gamma(\bL, S))]$.
Note also that set function $\sigma^w(S)$ is still monotone and submodular.
The influence instance $\calI$ is extended to include weight $w$.

\subsection{Algorithm {\ASVRRW}} 

Our Algorithm {\ASVRR} can be extended to the triggering model
  with weighted influence spreads.
Algorithm {\ASVRRW} follows essentially 
  the same steps of {\ASVRR}.
The only exception is that,
  when generating a random RR set $\bR$ 
  rooted at a random node $\bv$	(either in Phase 1 or Phase 2), 
  we select the root $\bv$ with probability 
  proportional to the weights of nodes.
To differentiate from random $\bv\sim V$, we use $\bv^w\sim_w V$ to denote
   a random node $\bv^w$ selected from $V$ according to
   node weights.
The random RR set generated from root $\bv^w$ is denoted as $\bR(\bv^w)$.
All the other aspects of the algorithm remains exactly the same.
In particular, the statement of Theorem~\ref{thm:ASVRR} remains essentially the same, except that
	$\psi$ is now the Shapley centrality of the weighted influence instance $\calI = (G, E, P_{\calI}, w)$.


The proof of Lemma~\ref{lem:margin} is changed accordingly to:
\begin{align*}
\sigma^w(S) = & n \cdot \E_{\bL}\left[ \E_{\bu^w} [\I\{\bu^w \in \Gamma(\bL, S) \} ] 
\right] \\
= &   n \cdot \E_{\bL,\bu^w} [\I\{\Gamma^-(\bL, \bu^w) \cap S \ne \emptyset \} ],
\end{align*}
where $\Gamma^-(L, u)$ is the set of nodes in graph $L$ that can reach $u$, and
$\bu^w$ is a random node drawn proportionally according to weight function $w$.
With random live-edge graph $\bL$, $\Gamma^-(\bL, u)$ is the same as the
RR set generated from root $u$, which is denoted as $\bR(u)$.
Thus, we have:
\begin{align*}
\sigma^w(S) = & n \cdot \E_{\bR(),\bu^w} [\I\{\bR(\bu^w) \cap S \ne \emptyset \}] \\
= &  n \cdot \Pr_{\bR(),\bu^w} (\bR(\bu^w) \cap S \ne \emptyset ),
\end{align*}
where the notation $\bR()$ means the randomness is only on the random generation
of reversed reachable set, but not on the random choice of the root node.
We use $\bR()$ to distinguish it from $\bR$, 
  which include the randomness of selecting the root node.
Weighted marginal spread $\sigma^w(S\cup \{v\}) - \sigma^w(S)$ 
  can be similarly argued.
  
The rest of the proof, including the proof on robustness and time complexity, essentially
	remains the same as given in Appendix~\ref{app:thm1}.

\subsection{Centrality Axioms for Weighted Influence Models}

In this section, we presented our axiomatic analysis for weighted influence
  models.

\subsubsection*{{\sc Weighted Social-Influence Instances}}

Mathematically, 
a weighted social-influence instance is a $4$-tuple $\calI^W=(V, E, P_{\calI},W)$,
 where (1) the influence instance $\calI=(V, E, P_{\calI})$ characterizes 
   the probabilistic profile of the influence model.
(2) $W$ is a 
	non-negative weight 
  function over $V$, i.e., $W(v) \geq 0, \forall v\in V$.
Although $W$ does not impact the influence process, 
  it defines the value of the activated set, and hence
  impacts the influence-spread profile of the model:
The weighted influence spread $\sigma_{\calI^W}$ is then given by:
\[
\sigma_{\calI^W}(S) = \E[W(\bI_{\calI}(S))] = \sum_{T\supseteq S} P_{\calI}(S, T) W(T).
\] 

Note that here we use the capital letter $W$ as the weight function that is integrated
	into the weighted influence instance $\calI^W = (V,E, P_{\calI}, W)$.
The capital letter $W$ is used to differentiate from the small letter $w$ used later
	as the parametrized weight function outside the influence instance.

Because $\calI$ and $W$ address different aspects of 
  the weighted influence model, $\calI^W=(V, E, P_{\calI},W)$,
 we assume they are independent of each other. 
We also extend the definition of centrality measure (Definition~\ref{def:CM}) to
{\em weighted centrality measure}, which is a mapping from a weighted influence
instance $\calI^W=(V, E, P_{\calI},W)$ to a real vector in $\R^{|V|}$.
We use $\psi^W$ to denote such a mapping.

\subsubsection*{{\sc Extension of Axioms 1-5}}

\begin{itemize}
\item Axiom \ref{axiom:anonymity} (Anonymity) has a natural extension, 
  if when we permute the influence-distribution-profile $\calI$ 
  with a $\pi$, we also permute weight function $W$ by $\pi$.
We will come back to this if-condition shortly.

\item Axiom \ref{axiom:normalization} (Normalization) is slightly changed such
	that the sum of the centrality measures is the total weights of all nodes:
\begin{axiom}[Weighted Normalization] \label{axiom:weightednormalization}
	For every weighted influence instance $\calI^W=(V, E, P_{\calI},W)$, 
	$\sum_{v\in V} \psi_{v}(\calI)  = W(V)$.
\end{axiom}

%
%
%

\item Axiom~\ref{axiom:sink} (Independence of Sink Nodes)
  remains the same.
%

\item Axiom \ref{axiom:bayesian} (Bayesian Influence) remains the same.

\item Axiom~\ref{axiom:critical} (Bargaining with Critical Sets) is replaced by 
  the following natural weighted version:

\begin{axiom} \mbox{\em (}{\sc Weighted Bargaining with Critical Sets}\mbox{\em )}
   \label{axiom:weightedcritical}
For the weighted critical set instance 
 $\calI^W_{R,v} = (R\cup\{v\}, E, P_{\calI_{R, v}}, W)$,
   the weighted centrality measure of $v$
is $\frac{|R| W(v)}{|R|+1}$, i.e. $\psi^W_v(\calI^W_{R,v}) = \frac{|R| W(v)}{|R|+1}$. 
\end{axiom}
The justification of the above axiom follows the same 
Nash bargaining argument for the non-weighted case.
Now the threat point is $(W(R), 0)$ and the slack is $W(v)$. 
The solution of 
$$(x_1,x_2) \in \argmax_{x_1\ge r, x_2\ge0, x_1+x_2=r+1} (x_1-W(R))^{1/r} \cdot x_2$$
	gives the fair share of $v$ as $\frac{|R| W(v)}{|R|+1}$.

\end{itemize}

\subsubsection*{\sc Characterization of Weighted Social Influence Model}

Let $\cA^W$ denote the set of Axioms~\ref{axiom:anonymity},
	\ref{axiom:sink}, \ref{axiom:bayesian}, , \ref{axiom:weightednormalization} and \ref{axiom:weightedcritical}.
Let {\em weighted Shapley centrality}, denoted as $\psi^{W,Shapley}$, 
	be the Shapley value of the weighted influence spread
	$\sigma_{\calI^W}$,
i.e., $\psi^{W,Shapley}(\calI^W) = \phi^{\Shapley}(\sigma_{\calI^W})$.
We now prove the following characterization theorem for weighted influence models:

\begin{theorem} {\sc (Shapley Centrality of Weighted Social Influence)} 
	\label{thm:weightedShapleyCen}
	Among all weighted centrality measures, 
	the weighted Shapley centrality $\psi^{W,Shapley}$ is 
	the unique weighted centrality measure that satisfies axiom set $\cA^W$
	(Axioms~\ref{axiom:anonymity},
	\ref{axiom:sink}, \ref{axiom:bayesian}, \ref{axiom:weightednormalization}, 
	and \ref{axiom:weightedcritical}).
\end{theorem}

The proof of Theorem~\ref{thm:weightedShapleyCen} follows the same proof structure of
	Theorem~\ref{thm:ShapleyCen}, and the main extension is on building a new full-rank
	basis for the space of weighted influence instances $\{\calI^W\}$, since this space has
	higher dimension than the unweighted influence instances $\{\calI\}$.

\begin{lemma}[Weighted Soundness] \label{lem:weightedshapleycen}
	The weighted Shapley centrality $\psi^{W,Shapley}$ satisfies all Axioms in $\cA^W$.
\end{lemma}
\begin{proof}[Sketch]
The proof essentially follows the same proof of Lemma~\ref{lem:shapleycen}, after replacing
	unweighted influence spread $\sigma_\calI$ with weighted influence
	spread $\sigma_{\calI^W}$.
Note that the proof of Lemma~\ref{lem:shapleycen} relies on earlier lemmas on the properties
	of sink nodes, which would be extended to the weighted version.
In particular, the result of Lemma~\ref{lem:sinkmargin} (a) is extended to:
	$$\sigma_{\calI^W}(S\cup \{v\}) - \sigma_{\calI^W}(S) = \Pr(v\not\in \bI_\calI(S)) \cdot W(v).$$
Lemma~\ref{lem:projectsigma} is extended to:
$$\sigma_{\calI\setminus \{v\}^W}(S) = \sigma_{\calI^W}(S)- \Pr(v\in \bI_\calI(S)) \cdot W(v).$$
All other results in Lemmas~\ref{lem:sinkmargin}--\ref{lem:projectmargin} are either the same,
	or extended by replacing $\sigma_{\calI}$ and $\sigma_{\calI\setminus \{v\}}$  to $\sigma_{\calI^W}$ and $\sigma_{\calI\setminus \{v\}^W}$, respectively.
With the above extension, the proof of Lemma~\ref{lem:weightedshapleycen} follows
	in the same way as the proof of Lemma~\ref{lem:shapleycen}.
\end{proof}

To prove the uniqueness, consider the profile of a weighted influence instance
	$\calI^W =(V, E, P_{\calI},W)$.
Comparing to the corresponding unweighted influence instance $\calI = (V, E, P_{\calI})$,
	$\calI^W$ has $n = |V|$ additional dimensions for the weights of the nodes, and thus
		we need $n$ additional parameters to specify node weights.
Recall that in the proof of Theorem~\ref{thm:ShapleyCen}, we overload the notation
	$P_{\calI}$ as a vector of $M$ dimensions to represent the influence probability profile
	of unweighted influence instance $\calI = (V, E, P_{\calI})$.
Similarly, we overload $W$ to represent a vector of $n$ dimensions 
  for the weights of
	$n$ nodes. 
Together, we use vector $(P_{\calI}, W)$ to represent a vector of $M' = M + n$ dimensions
	that fully determines a weighted influence instance $\calI$.

We now need to construct a set of basis vectors in $\R^{M'}$, each of which corresponds to
	a weighted influence instance.
The construction is still based on the critical set instance defined in
	Definition~\ref{def:critical}.
For every $R\subseteq V$ with $R\not\in \{\emptyset, V \}$ and every $U \supset R$,
	we consider the critical set instance $\calI_{R,U}$ with uniform weights (i.e.
	all nodes have weight $1$).
We use $\vec{1}$ to denote the uniform weight vector.
Then vector $(P_{\calI_{R,U}}, \vec{1}) \in \R^{M'}$ is the vector specifying the
	corresponding weighted critical set influence instance, denoted
	as $\calI_{R,U}^{\vec{1}}$.
Let $\vec{e}_i \in \R^{n}$ be the unit vector with $i$-th entry being $1$ and all other
	entries being $0$, for $i\in [n]$.
Then $\vec{e}_i$ corresponds to a weight assignment where the $i$-th node has weight $1$,
	and all other nodes have weight $0$.
Consider the null influence instance $\calI^N$, in which every node is an isolated node,
	same as defined in Lemma~\ref{lem:unique}.
We add weight vector $\vec{e}_i$ to the null instance $\calI^N$, to construct a unit-weight
	null instance $\calI^{N,\vec{e}_i}$, where every node is an isolated node, the
	$i$-th node has weight $1$, 
	and the rest have
	weight $0$, for every $i\in [n]$.
The vector representation of $\calI^{N,\vec{e}_i}$ is $(P_{\calI^N}, \vec{e}_i)$.
Note that, as already argued in the proof of Lemma~\ref{lem:unique}, 
	vector $P_{\calI^N}$ is the all-$0$ vector in $\R^M$.

Given the above preparation, we now define $\cV'$ as the set containing all the above
	vectors, that is:
\begin{align*}
\cV' = & 
\{(P_{\calI_{R,U}}, \vec{1}) \mid R, U \subseteq V, R \notin\{\emptyset, V\}, R\subset U  \} \\
	&  \quad	\cup \{(P_{\calI^N}, \vec{e}_i) \mid i \in [n]  \}.
\end{align*}
We prove the following lemma: 

\begin{lemma}[Independence of Weighted Influence] \label{lem:weightedlinearind}
	Vectors in $\cV'$ are linearly independent in the space $\R^{M'}$.
\end{lemma}
\begin{proof}
Our proof extends the proof of Lemma~\ref{lem:linearind}.
	Suppose, for a contradiction, that vectors in $\cV'$ are not linearly independent.
	Then for each $R$ and $U$ with $R, U \subseteq V, R \notin\{\emptyset, V\}, R\subset U $, we have a number $\alpha_{R,U} \in \R$, and for each $i$ we have a number $\alpha_i \in \R$,
	such that:
\begin{align} \label{eq:weightedlinear}
\sum_{R\not\in\{\emptyset,V\}, R\subset U} \alpha_{R,U} \cdot (P_{\calI_{R,U}}, \vec{1})
 + \sum_{i\in [n]} \alpha_i \cdot (P_{\calI^N}, \vec{e}_i)  = \vec{0},
\end{align}
	and at least some $\alpha_{R,U}\ne 0$ or some $\alpha_i \ne 0$.
	Suppose first that some $\alpha_{R,U}\ne 0$.
	Let $S$ be the smallest set with $\alpha_{S,U} \ne 0$ for some $U \supset S$, and
	let $T$ be any superset of $S$ with $\alpha_{S,T} \ne 0$.
	By the critical set instance definition, we have $P_{\calI_{S,T}}(S,T) = 1$.
	Since the vector does not contain any dimension corresponding to
	$P_\calI(S,S)$, we know that $T \supset S$.
	Moreover, since $P_{\calI^N}$ is an all-$0$ vector, we know that
	$P_{\calI^N}(S, T) = 0$.
	
	Then by the minimality of $S$, we have: 
	\begin{align}
	& 0 = \sum_{R, U:R\not\in\{\emptyset,V\}, R\subset U} \alpha_{R,U} \cdot P_{\calI_{R,U}}(S,T) 
	\nonumber \\
	& = \alpha_{S,T} \cdot P_{\calI_{S,T}}(S,T) + 
	\sum_{U: U \supset S,  U \ne T} \alpha_{S,U} \cdot P_{\calI_{S,U}}(S,T) + \nonumber \\
	& \quad \quad \sum_{R, U:|R| \ge |S|, R\ne S, U \supset R} \alpha_{R,U} \cdot P_{\calI_{R,U}}(S,T) 
	\nonumber 	\\
	& = \alpha_{S,T}  + 
	\sum_{U:U \supset S,  U \ne T} \alpha_{S,U} \cdot P_{\calI_{S,U}}(S,T) + \nonumber \\
	& \quad \quad \sum_{R, U:|R| \ge |S|, R\ne S, U \supset R} \alpha_{R,U} \cdot P_{\calI_{R,U}}(S,T). 	\nonumber
	\end{align}
	Following the same argument as in the proof of Lemma~\ref{lem:linearind}, we have
	$\alpha_{S,T} = 0$, which is a contradiction.

Therefore, we know that $\alpha_{R,U} = 0$ for all $R,U$ pairs, and there must be some
	$i$ with $\alpha_i \ne 0$.	
However, when all $\alpha_{R,U}$'s are $0$, what left in Eq.~\eqref{eq:weightedlinear} is
	$\sum_{i\in [n]} \alpha_i \cdot \vec{e}_i = \vec{0}$.
But since vectors $\vec{e}_i$'s are obviously linearly independent, the above cannot be true
	unless all $\alpha_i$'s are $0$, another contradiction.

Therefore, vectors in $\cV'$ are linearly independent.
\end{proof}

\begin{lemma}[Centrality Uniqueness of the Basis] \label{lem:weighteduniquebasis}
	Fix a set $V$.
	Let $\psi^W$ be a weighted centrality measure that satisfies axiom set $\cA^W$.
	For any instance $\calI^W$ that corresponds to a vector in $\cV'$,
	the centrality	$\psi(\calI^W)$ is unique.
\end{lemma}
\begin{proof}
Suppose first that $\calI^W$ is a weighted critical set instance
	$\calI_{R,U}^{\vec{1}}$.
Since $\calI_{R,U}^{\vec{1}}$ has the same weight
	for all nodes, its weighted centrality uniqueness can be argued in
	the exact same way
	as in the proof of Lemma~\ref{lem:uniquecritical}
	 (except that
	the unweighted Axioms~\ref{axiom:normalization} and~\ref{axiom:critical} are
	replaced by the corresponding weighted Axioms~\ref{axiom:weightednormalization}
	and~\ref{axiom:weightedcritical}).

Now suppose that $\calI^W$ is one of the instances $\calI^{N,\vec{e}_i}$, for some
	$i\in [n]$.
Since in instance $\calI^{N,\vec{e}_i}$ all nodes are isolated nodes, and thus sink nodes,
	for each node $v$, we can repeatedly apply the Sink Node Axiom (Axiom~\ref{axiom:sink})
	to remove all other nodes until $v$ is the only node in the graph, and this repeated
	projection will not change the centrality of $v$.
When $v$ is the only node in the graph, by the Weighted Normalization Axiom
	(Axiom~\ref{axiom:weightednormalization}), we know that $v$'s weighted centrality
	measure is $W(v)$.
Since the weights of all nodes are determined by the vector $\vec{e}_i$, 
	the weighted centrality of $\calI^{N,\vec{e}_i}$ is fully determined and is unique.
\end{proof}

\begin{lemma}[Weighted Completeness] \label{lem:weightedunique}
	The weighted centrality measure satisfying axiom set $\cA^W$ is unique.
\end{lemma}
\begin{proof}[Sketch]
The proof follows the proof structure of Lemma~\ref{lem:unique}.
Lemma~\ref{lem:weightedlinearind} already show that $\cV'$ is a set of basis vectors
	in the space $\R^{M'}$, and Lemma~\ref{lem:weighteduniquebasis} further shows that
	instances corresponding to these basis vectors have unique weighted centrality measures.
In addition, we define the $0$-weight null instance $\calI^{N,\vec{0}}$ to be
	an instance in which all nodes are isolated nodes, and all nodes have 
	weight $0$.
Then the vector corresponding to $\calI^{N,\vec{0}}$ is the all-$0$ vector in $\R^{M'}$.
Moreover, similar to $\calI^{N,\vec{e}_i}$, the weighted centrality of $\calI^{N,\vec{0}}$
	satisfying axiom set $\cA^W$ is also uniquely determined.
	
With the above preparation, the rest of the proof follows exactly the same logic as
	the one in the proof of Lemma~\ref{lem:unique}.
\end{proof}

\begin{proof}[of Theorem~\ref{thm:weightedShapleyCen}]
Theorem~\ref{thm:weightedShapleyCen} follows from 
  Lemmas~\ref{lem:weightedshapleycen}
	and~\ref{lem:weightedunique}.
\end{proof}
\subsubsection*{\sc Axiom Set Parametrized by Node Weights}

The above axiomatic characterization is based on the direct axiomatic extension 
	from unweighted influence models to the weighted influence models, where node
	weight function $W$ is directly added as part of the influence instance.
One may further ask the question: ``What if we treat node weights as parameters outside
	the influence instance? Is it possible to have an axiomatic characterization
	on such parametrized influence models, for {\em every} weight function?'' 

The answer to the above question would further highlight 
	the impact of the weight function to the influence model.
Since our goal is to achieve axiomatization that works for {\em every} weight function,
	we may need to seek for stronger axioms.

To achieve the above goal, for a given set $V$, we assume that the node weight function 
   cannot be permuted. 
To differentiate parametrized weight function from the integrated weight function $W$
	discussed before, we use small letter $w$ to represent the parametrized
	weight function: $w: V\rightarrow \R^+ \cup \{0\}$. 
The weight parameter $w$ appearing on the superscripts
	of notations such as influence instance $\calI$ and influence spread $\sigma$ 
	denotes that these quantities are parametrized by weight function $w$.
The influence spread $\sigma^w_{\calI}$ in influence instance
	$\calI = (V,E,P_{\calI})$ parametrized by weight $w$ is defined as:
\[
\sigma^w_{\calI}(S) = \E[w(\bI_{\calI}(S))] = \sum_{T\supseteq S} P_{\calI}(S, T) w(T).
\]

We would like to provide a natural axiom set $\cA^w$ parametrized by 
	$w: V\rightarrow \R^+\cup \{0\}$, such that
	the Shapley value for the weighted influence spread $\sigma^w$, 
	denoted as $\psi^{w,\Shapley}(\calI) = \phi^{\Shapley}(\sigma^w_\calI)$, 
  is the unique weighted centrality measure
	satisfying the axiom set $\cA^w$, for {\em every} such weight function $w$.
Recall that the weight function $w$ satisfies that $w(v)\ge 0$ for all $v\in V$.
Let $\psi^w$ denote a centrality measure satisfying the axiom set $\cA^w$.

Our Axiom set $\cA^w$ contains the weighted version of 
	Axioms~\ref{axiom:normalization}--\ref{axiom:critical},
	namely Axioms~\ref{axiom:sink}, \ref{axiom:bayesian},
	\ref{axiom:weightednormalization}, and \ref{axiom:weightedcritical}
	(of course, notation $W(v)$ is replaced by $w(v)$).
But it also needs an replacement of the Anonymity Axiom (Axiom~\ref{axiom:anonymity}).


By making  $w$ ``independent'' of the distribution profile of the influence 
  model $\calI = (V,E,P_\calI)$,
  the extension of  Axiom  Anonymity does not seem 
  to have a direct weighted version.
Conceptually, Axiom Anonymity is about node symmetry in the influence model.
However, when influence instance is parametrized by
	node weights, which cannot be 
  permuted and may not be uniform, 
  even if the influence instance $\calI$ has node symmetry, it does not imply that
  their weighted centrality is still the same.
This is precisely the reason we assume $w$ can not be permuted.

Therefore, we are seeking a new property 
  about node symmetry in the influence model parametrized by node weights
  to replace Axiom Anonymity.
We first define node pair symmetry as follows.
We denote $\pi_{uv}$ as the permutation in which $u$ and $v$ are mapped to each other while other nodes are mapped to themselves.
\begin{definition} \label{def:symmetric}
A node pair $u,v \in V$ is symmetric in the influence instance $\calI$ if
	for every $S, T \subseteq V$, 
	$P_\calI(S, T) = P_\calI(\pi_{uv}(S), \pi_{uv}(T))$, where
	$\pi_{uv}(S) = \{\pi_{uv}(v') \mid v'\in S\}$.
\end{definition}

We now give the axiom about node symmetry 
  in the weighted case, related to sink nodes and social influence
	projections.
\begin{axiom}[Weighted Node Symmetry] \label{axiom:weightedsymmetry}
In an influence instance $\calI=(V, E, P_{\calI})$, 
	let $S$ be the set of sink nodes.
If every pair of none-sink nodes are symmetric, then for any $v\in S$ and any $u \not\in S$, 
	$\psi^w_u(\calI) = \psi^w_u(\calI\setminus \{v\}) + \frac{1}{|V\setminus S|} (w(v) - \psi^w_v(\calI))$.
\end{axiom}
We justify the above axiom as follows.
Consider a sink node $v \in S$.
$\psi^w_v(\calI)$ is its fair share to the influence game.
Since $v$ cannot influence other nodes but may be influenced by others, its fair share is at most its
	weight $w(v)$ (can be formally proved).
Thus the leftover share of $v$, $w(v) - \psi^w_v(\calI)$, is divided among the rest nodes.
Since sink nodes do not influence others, they should have no contribution for the above
	leftover share from $v$.
Thus, the leftover share should be divided only among the rest non-sink nodes.
By the assumption of the axiom, all non-sink nodes are symmetric to one another, therefore
	they equally divide $w(v) - \psi^w_v(\calI)$, leading to $\frac{1}{|V\setminus S|} (w(v) - \psi^w_v(\calI))$
	contribution from each non-sink node.
Here an important remark is that, the weights of the non-sink nodes do not play a role
	in dividing the leftover share form $v$.
This is because, the weight of a node is an indication of the node's importance when it is influenced,
	but not its power in influencing others.
In other words, the influence power is determined by the influence instance $\calI$, in particular $P_\calI$,
	and it is unrelated to node weights.
Therefore, the above equal division of the leftover share is reasonable.
After this division, we can apply the influence projection to remove sink node $v$,
	and the remaining share
	of a non-sink node $u$ is simply the share of $u$ in the projected instance.	

The parametrized weighted axiom set $\cA^w$ is formed by Axioms~\ref{axiom:normalization}, \ref{axiom:sink},
	\ref{axiom:bayesian}, \ref{axiom:weightednormalization}, \ref{axiom:weightedcritical},
	and \ref{axiom:weightedsymmetry} (after replacing the weight notation $W()$ with $w()$
	in the corresponding axioms).
We define the weighted Shapley centrality $\psi^{w,\Shapley}(\calI)$ as the Shapley value of the 
	weighted influence spread $\phi^{\Shapley}(\sigma^w)$.
Note that this definition coincides with the definition of $\psi^{W,Shapley}(\calI^W)$, that is,
	whether or not we treat the weight function as an outside parameter or
	integrated into the influence instance, the weighted version of Shapley
	centrality is the same.
The following theorem summarizes the axiomatic characterization for the case of 
	parametrized weighted influence model.

\begin{theorem} {\sc (Parametrized Weighted Shapley Centrality of Social Influence)} \label{thm:weightedShapleyCen2}
	Fix a node set $V$. 
	For any normalized and non-negative node weight function $w: V\rightarrow R^+\cup \{0\}$, 
	the weighted Shapley centrality $\psi^{w,\Shapley}$ is 
	the unique weighted centrality measure that satisfies
	axiom set $\cA^w$ (Axioms~\ref{axiom:normalization}, \ref{axiom:sink},
	\ref{axiom:bayesian}, \ref{axiom:weightednormalization}, \ref{axiom:weightedcritical},
	and \ref{axiom:weightedsymmetry}).
	
\end{theorem}

\begin{lemma} \label{lem:weightedsink}
	If $v$ is a sink node in $\calI$, then for any $S\subseteq V\setminus \{v\}$, 
	(a) $\sigma^w_\calI(S\cup \{v\}) - \sigma^w_\calI(S) = w(v) \Pr(v\not\in \bI_\calI(S))$; and 
	(b) $\sigma^w_{\calI\setminus \{v\}}(S) = \sigma^w_\calI(S)- w(v)\Pr(v\in \bI_\calI(S))$.
\end{lemma}
\begin{proof}
The proof follows the proofs of Lemma~\ref{lem:sinkmargin} (a) and Lemma~\ref{lem:projectsigma}, except replacing $1$ with weight $w(v)$.
\end{proof}

\begin{lemma} \label{lem:symmetryinf}
If node pair $u, u'$ are symmetric in $\calI$, then for any $v \in V\setminus \{u,u'\}$,
	(a) for any
	$S\subseteq V$, 
	$\Pr(v \in \bI_{\calI}(S)) = \Pr(v \in \bI_{\calI}(\pi_{uu'}(S)))$'
	(b) for any random permutation $\bpi'$ on $V\setminus \{v\}$,
	$\E_{\bpi'}[\Pr(v\in \bI_{\calI}(S_{\bpi',u}))]
	= \E_{\bpi'}[\Pr(v\in \bI_{\calI}(S_{\bpi',u'}))]$, and
	$\E_{\bpi'}[\Pr(v\in \bI_{\calI}(S_{\bpi',u}\cup \{u\}))]
		= \E_{\bpi'}[\Pr(v\in \bI_{\calI}(S_{\bpi',u'}\cup \{u'\}))]$.
\end{lemma}
\begin{proof}
For (a), by the definition of symmetric node pair (Definition~\ref{def:symmetric}), we have
\begin{align*}
& \Pr(v \in \bI_{\calI}(S)) = \sum_{T\supseteq S\cup\{v\}} P_{\calI}(S, T) \\
& = \sum_{T\supseteq S\cup\{v\}} P_{\calI}(\pi_{uu'}(S), \pi_{uu'}(T)) \\
& = \sum_{\pi^{-1}_{uu'}(T)\supseteq S\cup\{v\}} P_{\calI}(\pi_{uu'}(S), T) \\
& = \sum_{T\supseteq \pi_{uu'}(S)\cup\{v\}} P_{\calI}(\pi_{uu'}(S), T) 
	= \Pr(v \in \bI_{\calI}(\pi_{uu'}(S))).
\end{align*}
For (b), we use (a) and obtain
\begin{align*}
\E_{\bpi'}[\Pr(v\in \bI_{\calI}(S_{\bpi',u}))] =  \E_{\bpi'}[\Pr(v\in \bI_{\calI}(\pi_{uu'}(S_{\bpi',u})))]. 
\end{align*}
Note that $\pi_{uu'}(S_{\bpi',u})$ is a random set obtained by first generating a random
	permutation $\bpi'$, then selecting the prefix node set $S_{\bpi',u}$ 
	before node $u$ in $\bpi'$, and finally
	replacing the possible occurrence of $u'$ in $S_{\bpi',u}$ with $u$
	($u$ cannot occur in $S_{\bpi',u}$ so there is no replacement of $u$ with $u'$).
This random set can be equivalently obtained by first generating the
	random permutation $\bpi'$, then switching the position of $u$ and $u'$
	(denote the new random permutation $\pi_{uu'}(\bpi')$),
	and finally selecting the prefix node set $S_{\pi_{uu'}(\bpi'),u'}$ before
		$u'$ in $\pi_{uu'}(\bpi')$.
We note that random permutations $\bpi'$ and $\pi_{uu'}(\bpi')$ follow
	the same distribution, and thus we have
\begin{align*}
\E_{\bpi'}[\Pr(v\in \bI_{\calI}(S_{\bpi',u}))] =  \E_{\bpi'}[\Pr(v\in \bI_{\calI}(S_{\bpi',u'}))]. 
\end{align*}
The equality $\E_{\bpi'}[\Pr(v\in \bI_{\calI}(S_{\bpi',u}\cup \{u\}))]
= \E_{\bpi'}[\Pr(v\in \bI_{\calI}(S_{\bpi',u'}\cup \{u'\}))]$ can
	be argued in the same way.
\end{proof}

\begin{lemma}[Weighted Soundness] \label{lem:weightedShapley2}
Weighted Shapley centrality $\psi^{w,\Shapley}(\calI)$ satisfies all axioms in $\cA^w$.
\end{lemma}
\begin{proof}
Satisfaction of Axioms~\ref{axiom:normalization}, \ref{axiom:sink}, 
\ref{axiom:bayesian}, \ref{axiom:weightednormalization} and \ref{axiom:weightedcritical}
	can be similarly verified as in the proof of Lemma~\ref{lem:weightedshapleycen}.
We now verify Axiom~\ref{axiom:weightedsymmetry}.

Let $v$ be a sink node and $u$ be a non-sink node.
Let $\bpi'$ be a random permutation on node set $V\setminus \{v\}$.
We have
	\begin{align}
	& \psi^{w,\Shapley}_u(\calI) =  \E_{\bpi}[\sigma^w_\calI(S_{\bpi,u} \cup \{u\}) - \sigma^w_\calI(S_{\bpi,u})] \nonumber \\
	& = \Pr(u \prec_{\bpi} v)\E_{\bpi}[\sigma^w_\calI(S_{\bpi,u} \cup \{u\}) - \sigma^w_\calI(S_{\bpi,u}) \mid u \prec_{\bpi} v] + \nonumber \\
	& \quad \Pr(v \prec_{\bpi} u)\E_{\bpi}[\sigma^w_\calI(S_{\bpi,u} \cup \{u\}) - \sigma^w_\calI(S_{\bpi,u}) \mid v \prec_{\bpi} u] \nonumber \\
	& = \E_{\bpi'}[\sigma^w_\calI(S_{\bpi',u} \cup \{u\}) - \sigma^w_\calI(S_{\bpi',u})] /2 + \nonumber \\
	& \quad \E_{\bpi}[\sigma^w_\calI(S_{\bpi,u} \cup \{u\}) - \sigma^w_\calI(S_{\bpi,u}) \mid v \prec_{\bpi} u] / 2 \nonumber \\
	& = \E_{\bpi'}[\sigma^w_\calI(S_{\bpi',u} \cup \{u\}) - \sigma^w_\calI(S_{\bpi',u})] /2 + \nonumber \\
	& \quad \E_{\bpi}[\sigma^w_\calI(S_{\bpi,u} \setminus \{v\} \cup \{u\}) - \sigma^w_\calI(S_{\bpi,u}\setminus \{v\} ) \mid v \prec_{\bpi} u] / 2 \nonumber \\
	& \quad + w(v)\E_{\bpi}[ \Pr(v\not\in \bI_\calI(S_{\bpi,u} \setminus \{v\} \cup \{u\})) - \nonumber \\
	& \quad \quad  \Pr(v\not\in \bI_\calI(S_{\bpi,u} \setminus \{v\}))  \mid v \prec_{\bpi} u] / 2
	\label{eq:wusesinkind} \\
	& = \E_{\bpi'}[\sigma^w_\calI(S_{\bpi',u} \cup \{u\}) - \sigma^w_\calI(S_{\bpi',u})] \nonumber \\
	& \quad + w(v)\E_{\bpi'}[ \Pr(v\not\in \bI_\calI(S_{\bpi',u} \cup \{u\})) - \nonumber \\
	& \quad \quad  \Pr(v\not\in \bI_\calI(S_{\bpi',u} )) ] / 2 \nonumber \\
	& = \E_{\bpi'}[\sigma^w_{\calI\setminus \{v\}}(S_{\bpi',u} \cup \{u\}) - \sigma^w_{\calI\setminus \{v\}}(S_{\bpi',u})] \nonumber \\
	& \quad + w(v)\E_{\bpi'}[ \Pr(v\in \bI_\calI(S_{\bpi',u} \cup \{u\})) - \nonumber \\
	& \quad \quad  \Pr(v\in \bI_\calI(S_{\bpi',u} )) ]  \nonumber\\
	& \quad + w(v)\E_{\bpi'}[ \Pr(v\not\in \bI_\calI(S_{\bpi',u} \cup \{u\})) - \nonumber \\
	& \quad \quad  \Pr(v\not\in \bI_\calI(S_{\bpi',u} )) ] / 2 
	\label{eq:wuseprojectmar} \\
	& = \psi^{w,\Shapley}_u({\calI\setminus \{v\}}) \nonumber \\
	& \quad + w(v)\E_{\bpi'}[ \Pr(v\in \bI_\calI(S_{\bpi',u} \cup \{u\})) - \nonumber \\
	& \quad \quad  \Pr(v\in \bI_\calI(S_{\bpi',u} )) ] / 2. \label{eq:shapleysink}
	\end{align}
Eq.\eqref{eq:wusesinkind} is by Lemma~\ref{lem:weightedsink} (a), and Eq.\eqref{eq:wuseprojectmar} is by Lemma~\ref{lem:weightedsink} (b).
For $v$'s weighted Shapley centrality, we have
\begin{align}
& \psi^{w,\Shapley}_v(\calI) = \E_{\bpi}[\sigma^w_\calI(S_{\bpi,v} \cup \{v\}) - \sigma^w_\calI(S_{\bpi,v})] \nonumber \\
& \quad = w(v) \E_{\bpi}[ \Pr(v\not\in \bI_\calI(S_{\bpi,v} )) ], \label{eq:wsinkv}
\end{align}
where the last equality above is also by Lemma~\ref{lem:weightedsink} (a).

Recall that in Axiom~\ref{axiom:weightedsymmetry} $S$ is the set of sink nodes
	and $u \in V\setminus S$ is a non-sink node.
Then we have
\begin{align}
& 1 =  \E_{\bpi}[\Pr(v \in \bI_{\calI}(V) )]  \nonumber \\
&  =  \E_{\bpi}[\sum_{u'\in V} (\Pr(v\in \bI_{\calI}(S_{\bpi,u'}\cup \{u'\})) - \Pr(v\in \bI_{\calI}(S_{\bpi,u'})))]  \label{eq:telescoping} \\
&  =  \sum_{u'\in V\setminus \{v\}} \Pr(u' \prec_{\bpi} v) \E_{\bpi}[\Pr(v\in \bI_{\calI}(S_{\bpi,u'}\cup \{u'\})) - \nonumber \\
& \quad \Pr(v\in \bI_{\calI}(S_{\bpi,u'})) \mid u' \prec_{\bpi} v ] + \nonumber \\
& \quad  \E_{\bpi}[\Pr(v\in \bI_{\calI}(S_{\bpi,v}\cup \{v\})) - \Pr(v\in \bI_{\calI}(S_{\bpi,v}))] \label{eq:includecancel} \\
&  =  \sum_{u'\in V\setminus \{v\}} \E_{\bpi'}[\Pr(v\in \bI_{\calI}(S_{\bpi',u'}\cup \{u'\})) - \nonumber \\
& \quad \quad \Pr(v\in \bI_{\calI}(S_{\bpi',u'}))] /2  + \E_{\bpi}[\Pr(v \not\in \bI_{\calI}(S_{\bpi,v}))] \nonumber \\
& =  \sum_{u' \in V\setminus S} \E_{\bpi'}[\Pr(v\in \bI_{\calI}(S_{\bpi',u'}\cup \{u'\})) - \nonumber \\
& \quad \quad \Pr(v\in \bI_{\calI}(S_{\bpi',u'}))] /2  
   + \E_{\bpi}[\Pr(v \not\in \bI_{\calI}(S_{\bpi,v}))] \label{eq:sinknomargin} \\
& =  |V\setminus S|\cdot \E_{\bpi'}[\Pr(v\in \bI_{\calI}(S_{\bpi',u}\cup \{u\})) - \nonumber \\
& \quad \quad \Pr(v\in \bI_{\calI}(S_{\bpi',u}))]/2   
	+ \E_{\bpi}[\Pr(v \not\in \bI_{\calI}(S_{\bpi,v}))] \label{eq:symmetryinf} 
\end{align}
Eq.~\eqref{eq:telescoping} is a telescoping series where all middle terms are canceled out.
Eq.~\eqref{eq:includecancel} is because when $v \prec_{\bpi} u'$, $v\in S_{\bpi,u'}$ and thus 
	$\Pr(v\in \bI_{\calI}(S_{\bpi,u'}\cup \{u'\})) = \Pr(v\in \bI_{\calI}(S_{\bpi,u'}))) = 1$.
Eq.~\eqref{eq:sinknomargin} is by Lemma~\ref{lem:sinkmargin} (b), and
Eq.~\eqref{eq:symmetryinf} is by Lemma~\ref{lem:symmetryinf} (b).
Therefore, from Eq.~\eqref{eq:symmetryinf}, we have
\begin{align*}
&\E_{\bpi'}[\Pr(v\in \bI_{\calI}(S_{\bpi',u}\cup \{u\})) - \Pr(v\in \bI_{\calI}(S_{\bpi',u}))]/2 \\
& \quad = \frac{1}{|V\setminus S|} (1 - \E_{\bpi}[\Pr(v \not\in \bI_{\calI}(S_{\bpi,v}))]).
\end{align*}
Plugging the above equality into Eq.~\eqref{eq:shapleysink}, we obtain
\begin{align*}
& \psi^{w,\Shapley}_u(\calI)  \\
& \quad = \psi^{w,\Shapley}_u({\calI\setminus \{v\}}) +
	\frac{w(v)(1 - \E_{\bpi}[\Pr(v \not\in \bI_{\calI}(S_{\bpi,v}))])}{|V\setminus S|}  \\
& \quad = \psi^{w,\Shapley}_u({\calI\setminus \{v\}}) + \frac{1}{|V\setminus S|} (w(v) - \psi^{w,\Shapley}_v(\calI)),
\end{align*}
where the last equality above uses Eq.~\eqref{eq:wsinkv}.
The above equality is exactly the one in Axiom~\ref{axiom:weightedsymmetry}.
\end{proof}

For the uniqueness of the parametrized axiom set $\cA^w$, the proof follows the same
	structure as the proof for $\cA$.
The only change is in the proof of Lemma~\ref{lem:uniquecritical}, which
	we provide the new version for the weighted case below.
\begin{lemma}[Weighted Critical Set Instances]
Fix a $V$.  
For any normalized and non-negative node weight function $w: V\rightarrow R^+\cup \{0\}$, 
let $\psi^w$ be a weighted centrality measure that satisfies axiom set $\cA^w$.
	For any $R, U \subseteq V$ with $R\ne \emptyset$ and $R\subseteq U$,
	and the critical set instance $\calI_{R,U}$ as defined
	in Definition~\ref{def:critical}, its weighted centrality
	$\psi^w(\calI_{R,U})$ must be unique.
\end{lemma}
\begin{proof}
	Consider the critical set instance $\calI_{R,U}$.
	We first consider a node $v \in V\setminus U$.
	By Lemma~\ref{lem:isosink}, every node $u \in V\setminus R$ is a sink node.
	Then we can apply the Sink Node Axiom (Axiom~\ref{axiom:sink}) to iteratively
	remove all nodes in $U\setminus R$ without changing $v$'s centrality measure.
	After removing nodes in $U\setminus R$, we know that in the remaining projected instance,
	every node becomes an isolated node.
	Then we can further remove all other nodes and only leave $v$ in the graph, still
	not changing $v$'s centrality measure.
	When $v$ is the only node left in the graph, 
	we then apply the Weighted Normalization Axiom (Axiom~\ref{axiom:weightednormalization})
	and obtain that 
	$\psi^w_v(\calI_{R,U}) = w(v)$.
	Thus, every node $v \in V\setminus U$ has uniquely determined centrality measure $w(v)$.
	
	Next, we consider a node $v\in U\setminus R$.
	By Lemma~\ref{lem:isosink}, every node $v \in V\setminus R$ is a sink node.
	Then we can apply the Sink Node Axiom (Axiom~\ref{axiom:sink}) to iteratively
	remove all these sink nodes except $v$, such that the centrality measure of $v$ does not change
	after the removal.
	By Lemma~\ref{lem:projcritical}, the remaining instance with node set $R\cup\{v\}$ is 
	still a critical set instance with critical set $R$ and target set $R\cup \{v\}$.
	Thus we can apply the Weighted Bargaining with Critical Set Axiom (Axiom~\ref{axiom:weightedcritical}) to this remaining
	influence instance, and know $\psi^w_v(\calI_{R,U}) = |R|w(v)/(|R|+1)$,
	for every node $v\in U\setminus R$.
	
	Finally, we consider a node $v\in R$.
	Again we can remove all sink nodes in $V\setminus R$ iteratively by influence projection until
	we only have nodes in $R$ left, which is the instance $\calI_{R,R}$ in the graph with node set $R$.
	It is straightforward to verify that every pair of nodes in $R$ are symmetric.
	Therefore, we can apply the Weighted Node Symmetry Axiom (Axiom~\ref{axiom:weightedsymmetry})
	to node $v\in R$.
	In particular, for every isolated node $u \in V\setminus U$, since we have $\psi^w_u(\calI_{R,U}) = w(u)$,
	there is no leftover share from $u$ that $v$ could claim.
	For every node $u' \in U \setminus R$, we have $\psi^w_{u'}(\calI_{R,U}) = |R|w(u')/(|R|+1)$,
	and thus the leftover share from $u'$ is $w(u')/(|R|+1)$.
	By Axiom~\ref{axiom:weightedsymmetry}, node $v\in R$ would obtain $w(u')/(|R|(|R|+1))$ from $u'$.
	In the final projected instance $\calI_{R,R}$ with node set $R$, it is easy to check that every node
	is an isolated node. 
	Thus by a similar argument of removing all other nodes and applying Weighted Normalization Axiom,
	we know that in this final projected instance $v$'s weighted centrality is $w(v)$.
	Summing them up by Axiom~\ref{axiom:weightedsymmetry}, we have
	$\psi^w_v(\calI_{R,U}) = w(v) + \frac{w(U\setminus R)}{|R|(|R|+1)}$.
	
	Therefore, the weighted centrality measure for instance $\psi^w(\calI_{R,U})$ is uniquely determined.
\end{proof}

Once we set up the uniqueness for the critical set instances in the above lemma, the rest proof
	follows the proof for the unweighted axiom set $\cA$.
In particular, the linear independence lemma (Lemma~\ref{lem:linearind}) remains the same, since it
	only concerns about influence instances and is not related to node weights.
Lemma~\ref{lem:unique} also follows, excepted that when arguing the centrality uniqueness for
	the null influence instance $\calI^N$, we again use repeated projection and apply
	the Weighted Normalization Axiom (Axiom~\ref{axiom:weightednormalization}) 
	to show that each node $v$ in the null instance has the unique centrality measure
	of $w(v)$.
Therefore, Theorem~\ref{thm:weightedShapleyCen2} holds.

\section{Martingale Tail Bounds} \label{sec:martingale}

There are numerous variants of Chernoff bounds and the more general martingale tail bounds
	in the literature (e.g.~\cite{MU05,CL06,tang15}).
However, they either cover the case of independent variables, or Bernoulli variables, or a bit
	looser bounds, or some general cases with different conditions.
In this section, for completeness, we state the general martingale tail bounds that we need
	for this paper, and provide a complete proof.
The proof structure follows that of \cite{MU05} for Chernoff bounds, but the result is more general.

\begin{theorem}[Martingale Tail Bounds] \label{thm:martingaletail}
	Let $\bX_1, \bX_2, \ldots, \bX_t$ be random variables 
	with range $[0,1]$.
\begin{enumerate}
	\item[(1)] 
	Suppose that $\E[\bX_i \mid \bX_1, \bX_2, \ldots, \bX_{i-1}] \le \mu_i$ for every
	$i\in [t]$.
	Let $\bY = \sum_{i=1}^t \bX_i$, and $\mu = \sum_{i=1}^t \mu_i$.
	For any $\delta > 0$, we have:
	\begin{align*}
	\Pr\{\bY - \mu \ge \delta \cdot \mu \} &\le \exp\left( - \frac{\delta^2}{2+\frac{2}{3}\delta} \mu \right).
	\end{align*}
	
	\item[(2)] 
	Suppose that $\E[\bX_i \mid \bX_1, \bX_2, \ldots, \bX_{i-1}] \ge \mu_i$, $\mu_i \ge 0$, for every
	$i\in [t]$.
	Let $\bY = \sum_{i=1}^t \bX_i$, and $\mu = \sum_{i=1}^t \mu_i$.
	For any $0< \delta < 1$, we have:
	\begin{align*}
	\Pr\{\bY - \mu \le - \delta \cdot  \mu \} &\le \exp\left( - \frac{\delta^2}{2} \mu  \right).
	\end{align*}
\end{enumerate}	
\end{theorem}

\begin{lemma} \label{lem:bernoullitightbound}
Let $\bX_1, \bX_2, \ldots, \bX_t$ be Bernoulli random variables with range $\{0,1\}$.
\begin{enumerate}
	\item[(1)] 
	Suppose that $\E[\bX_i \mid \bX_1, \bX_2, \ldots, \bX_{i-1}] \le \mu_i$ for every
	$i\in [t]$.
	Let $\bY = \sum_{i=1}^t \bX_i$, and $\mu = \sum_{i=1}^t \mu_i$.
	For any $\delta > 0$, we have:
	\begin{align} 
	\Pr\{\bY - \mu \ge \delta \cdot \mu \} &\le 
		\left( \frac{e^\delta}{(1+\delta)^{(1+\delta)}}\right)^\mu. \label{eq:uppertail}
	\end{align}
	
	\item[(2)] 
	Suppose that $\E[\bX_i \mid \bX_1, \bX_2, \ldots, \bX_{i-1}] \ge \mu_i$, $\mu_i \ge 0$, for every
	$i\in [t]$.
	Let $\bY = \sum_{i=1}^t \bX_i$, and $\mu = \sum_{i=1}^t \mu_i$.
	For any $ 0 < \delta < 1$, we have:
	\begin{align}
	\Pr\{\bY - \mu \le - \delta \cdot  \mu \} &\le \left( \frac{e^{-\delta}}{(1-\delta)^{(1-\delta)}}  \right)^\mu. \label{eq:lowertail}
	\end{align}
\end{enumerate}	
\end{lemma}
\begin{proof}
Let $\bY_i = \sum_{j=1}^i \bX_j$, for $i\in [t]$.
For (1), applying Markov's inequality, for any $\alpha > 0$, we have
\begin{align}
& \Pr\{\bY - \mu \ge \delta \cdot \mu \}  \nonumber \\
	& = \Pr\{ e^{\alpha\bY} \ge e^{\alpha(1+\delta)\mu} \} \nonumber\\
	& \le \frac{\E[e^{\alpha\bY}]}{e^{\alpha(1+\delta)\mu}} \label{eq:markov}
		\\
	& = \frac{\E[e^{\alpha(\bX_t + \bY_{t-1})}]}{e^{\alpha(1+\delta)\mu}} \nonumber\\
	& = \frac{\E[ \E[e^{\alpha(\bX_t + \bY_{t-1})} \mid \bX_1,\ldots, \bX_{t-1}]]}{e^{\alpha(1+\delta)\mu}} \nonumber\\
	& = \frac{\E[ e^{\alpha \bY_{t-1}}\E[e^{\alpha \bX_t} \mid \bX_1,\ldots, \bX_{t-1}]]}{e^{\alpha(1+\delta)\mu}}, \label{eq:Yt}
\end{align}
where Inequality~\eqref{eq:markov} is by the Markov's inequality.
Next, for the term $\E[e^{\alpha \bX_t} \mid \bX_1,\ldots, \bX_{t-1}]$, we utilize the fact that $X_t$ is a Bernoulli
	random variable and have
\begin{align}
& \E[e^{\alpha \bX_t} \mid \bX_1,\ldots, \bX_{t-1}] \nonumber \\
	& = \Pr\{\bX_t = 0 \mid \bX_1,\ldots, \bX_{t-1}\} \E[1 \mid \bX_1,\ldots, \bX_{t-1}]  \nonumber \\
	& \quad + \Pr\{\bX_t = 1 \mid \bX_1,\ldots, \bX_{t-1}\} \E[e^\alpha \mid \bX_1,\ldots, \bX_{t-1}]  \nonumber \\
	& = (e^\alpha - 1) \E[\bX_t \mid \bX_1,\ldots, \bX_{t-1}] + 1 \nonumber \\
	& \le (e^\alpha - 1) \mu_t + 1 \nonumber \\
	& \le e ^{(e^\alpha - 1) \mu_t}, \nonumber
\end{align}
where the last inequality uses the fact that $1+z \le e^z$ for any $z$.
Plugging the above inequality into Eq.\eqref{eq:Yt}, we have
\begin{align*}
& \Pr\{\bY - \mu \ge \delta \cdot \mu \} \\
	& \le \frac{ e ^{(e^\alpha - 1) \mu_t} \E[ e^{\alpha \bY_{t-1}}]}{e^{\alpha(1+\delta)\mu}} \\
	& \le \frac{ e ^{(e^\alpha - 1) (\mu_{t-1} + \mu_t)} \E[ e^{\alpha \bY_{t-2}}]}{e^{\alpha(1+\delta)\mu}} \\
	& \le \cdots \le \frac{ e ^{(e^\alpha - 1) \mu} }{e^{\alpha(1+\delta)\mu}}.
\end{align*}
Finally, by setting $\alpha = \ln(\delta + 1)$, we obtain Inequality~\eqref{eq:uppertail}.

For (2), the analysis follows the same strategy: for any $\alpha>0$, we have
\begin{align*}
& \Pr\{\bY - \mu \le - \delta \cdot \mu \} \\
& = \Pr\{e^{-\alpha\bY} \ge e^{-\alpha(1-\delta)\mu} \} \\
	& \le \frac{\E[e^{-\alpha\bY}]}{e^{-\alpha(1-\delta)\mu}} 
	\\
	& = \frac{\E[ e^{-\alpha \bY_{t-1}}\E[e^{-\alpha \bX_t} \mid \bX_1,\ldots, \bX_{t-1}]]}{e^{-\alpha(1-\delta)\mu}} \\
	& \le \frac{ e ^{(e^{-\alpha} - 1) \mu_t} \E[ e^{\alpha \bY_{t-1}}]}{e^{-\alpha(1-\delta)\mu}} \\
	& \le \frac{ e ^{(e^{-\alpha} - 1) \mu} }{e^{-\alpha(1-\delta)\mu}}.
\end{align*}
Finally, by setting $\alpha = -\ln(1-\delta)$, we obtain Inequality~\eqref{eq:lowertail}.
\end{proof}

Recall that a function $f$ is {\em convex} if for any $x_1$ and $x_2$ and for any $0\le \lambda \le 1$,
\[
f(\lambda x_1 + (1-\lambda) x_2) \le \lambda f(x_1) + (1-\lambda) f(x_2).
\]

\begin{lemma} \label{lem:convex}
Let $\bX$ be a random variable with range $[0,1]$, and let $p = \E[\bX]$.
Let $\bZ$ be the Bernoulli random variable with $\Pr(\bZ=1) = p$.
For any convex function $f$, we have $\E[f(\bX)] \le \E[f(\bZ)]$.
\end{lemma}
\begin{proof}
Let $\bX$ be a random variable defined on the probability space $(\Omega, \Sigma, P)$, then we have
\[
p = \E[\bX] = \int_{\Omega} \bX(\omega) P({\rm d} \omega),
\]
and
\[
\E[f(\bX)] = \int_{\Omega} f(\bX(\omega)) P({\rm d} \omega).
\]
Applying the convexity of $f$, together with the assumption that the range of $\bX$ is $[0,1]$, we have
\begin{align*}
\E[f(\bX)] & = \int_{\Omega} f(\bX(\omega)) P({\rm d} \omega) \\
& \le \int_{\Omega} ((1-\bX(\omega)) f(0) + \bX(\omega) f(1)) P({\rm d} \omega)\\
& = \left(1 - \int_{\Omega} \bX(\omega) P({\rm d} \omega) \right) f(0) \\
& \quad + \left(\int_{\Omega} \bX(\omega) P({\rm d} \omega)\right) \cdot f(1) \\
& = (1-p) f(0) + p f(1) = \E[f(\bZ)].
\end{align*}
%
\end{proof}

\begin{lemma}
	Let $\bX_1, \bX_2, \ldots, \bX_t$ be random variables with range  $[0,1]$.
Items (1) and (2) in Lemma~\ref{lem:bernoullitightbound} still hold.
\end{lemma}
\begin{proof}
For item (1), following the proof of Lemma~\ref{lem:bernoullitightbound}, we can still
	obtain Inequality~\eqref{eq:Yt}.
For the term $\E[e^{\alpha \bX_t} \mid \bX_1,\ldots, \bX_{t-1}]$, 
	notice that function $f(x) = e^{\alpha x}$ is convex for any $\alpha > 0$.
Therefore, we can apply Lemma~\ref{lem:convex}.
In particular, let $\bZ_t$ be a Bernoulli random variable with 
	$\E[\bZ_t] = \E[\bX_t \mid \bX_1, \ldots, \bX_{t-1}]$.
By Lemma~\ref{lem:convex}, we have
\begin{align*}
& \E[e^{\alpha \bX_t} \mid \bX_1,\ldots, \bX_{t-1}]  \\
& \le \E[e^{\alpha \bZ_t}] \\
& = \Pr\{\bZ_t = 0 \} \cdot 1  + \Pr\{\bZ_t = 1\} \cdot e^\alpha   \\
& = (e^\alpha - 1) \E[\bZ_t] + 1  \\
& \le (e^\alpha - 1) \mu_t + 1  \\
& \le e ^{(e^\alpha - 1) \mu_t}.
\end{align*}
The rest of the proof of item (1) is the same.

For item (2), the treatment is the same, as long as we notice that
	function $g(x) = e^{-\alpha x}$ is also convex for any $\alpha > 0$.
\end{proof}

\begin{lemma}
For $\delta > 0$, we have
\begin{align} \label{eq:relaxuppertail}
\frac{e^\delta}{(1+\delta)^{(1+\delta)}} \le \exp\left( - \frac{\delta^2}{2+\frac{2}{3}\delta} \right).
\end{align}
For $0< \delta < 1$, we have
\begin{align} \label{eq:relaxlowertail}
\frac{e^{-\delta}}{(1-\delta)^{(1-\delta)}} \le \exp\left( - \frac{\delta^2}{2} \right).
\end{align}
\end{lemma}
\begin{proof}
For Inequality~\eqref{eq:relaxuppertail}, we take logarithm of both sides, and obtain
	the following equivalent inequality:
\[
f(\delta) = \delta - (1+\delta)\ln (1+\delta) + \frac{\delta^2}{2+\frac{2}{3}\delta} \le 0.
\]
Taking the derivatives of $f(\delta)$, we have:
\begin{align*}
f'(\delta) & = 1 - \frac{1+\delta}{1+\delta} - \ln(1+\delta) + \frac{2\delta}{2+\frac{2}{3}\delta}
	- \frac{\frac{2}{3}\delta^2}{(2+\frac{2}{3}\delta)^2} \\
	& = - \ln(1+\delta) + \frac{18\delta + 3\delta^2}{2 \cdot (3+\delta)^2} \\
	& = - \ln(1+\delta) + \frac{3 \cdot ((3 +\delta)^2 - 9)}{2 \cdot (3+\delta)^2} \\
	& = - \ln(1+\delta) + \frac{3}{2} - \frac{27}{2 \cdot (3+\delta)^2}; \\
f''(\delta) & = -\frac{1}{1+\delta} + \frac{27}{(3+\delta)^3}; \\
f'''(\delta) & = \frac{1}{(1+\delta)^2} - \frac{81}{(3+\delta)^4}. \\
\end{align*}
When $\delta \ge 0$, $f'''(\delta)=0$ has exactly two solutions at $\delta_1=0$ and
	$\delta_2 = 3$.
When $\delta \in (0,3)$, $f'''(\delta) <0$, and
	when $\delta > 3$, $f'''(\delta) > 0$.
	
Looking at $f''(\delta)$, we have $f''(0)=0$ and $\lim_{\delta\rightarrow +\infty} f''(\delta) = 0$.
When $\delta$ increases from $0$ to $3$, since $f'''(\delta) <0$, $f''(\delta)$ decreases, which means
	$f''(\delta) <0$; when $\delta$ increases from $3$, since $f'''(\delta)>0$,
	$f''(\delta)$ keeps increasing, but never increases above $0$.
Thus, for all $\delta \ge 0$, $f''(\delta) \le 0$.

Looking at $f'(\delta)$, we have $f'(0) = 0$.
Since $f''(\delta)\le 0$ for all $\delta \ge 0$, $f'(\delta)$ is monotonically non-increasing, and thus
	$f'(\delta) \le 0$ for all $\delta \ge 0$.
	
Finally, looking at $f(\delta)$, we have $f(0) = 0$.
Since $f'(\delta) \le 0$ for all $\delta \ge 0$, $f(\delta)$ is monotonically non-increasing, and thus
$f(\delta) \le 0$ for all $\delta \ge 0$.

For Inequality~\eqref{eq:relaxlowertail}, we take logarithm of both sides, and obtain the following
	equivalent inequality:
\begin{align*}
g(\delta) = -\delta - (1-\delta) \ln (1-\delta) + \frac{\delta^2}{2} \le 0.
\end{align*}
Taking the derivatives of $g(\delta)$, we have:
\begin{align*}
g'(\delta) & = -1 +\ln(1-\delta) + \frac{1-\delta}{1-\delta} + \delta \\
		& = \ln(1-\delta) + \delta; \\
g''(\delta) & = -\frac{1}{1-\delta} + 1.
\end{align*}
Looking at $g''(\delta)$, it is clear that $g''(\delta) < 0$ for $\delta\in (0,1)$, and $g''(0)=0$.
This implies that $g'(\delta)$ is monotonically decreasing in $(0,1)$.
Since $g'(0) = 0$, we have $g'(\delta) \le 0$ for $\delta \in (0,1)$.
This implies that $g(\delta)$ is monotonically increasing in $(0,1)$.
Finally, since $g(0) = 0$, we have $g(\delta)\le 0$ for all $\delta \in (0,1)$.
\end{proof}
}

\end{document}